\documentclass[11pt,draftclsnofoot,peerreviewca,onecolumn,letterpaper]{IEEEtran}

\advance \topmargin by -\headheight
\advance \topmargin by -\headsep
\evensidemargin \oddsidemargin
\usepackage{amsmath}
\usepackage{amssymb,theorem}
\usepackage[pdftex]{graphics}
\usepackage{epsfig}
\usepackage[pdftex]{color}
\usepackage{psfrag,bbm}
\usepackage{relsize}

\usepackage{latexsym}
\usepackage{epsfig}

\usepackage{graphicx}
\usepackage[font=small]{caption}
\usepackage[font=footnotesize]{subcaption}
\usepackage{wrapfig}
\usepackage{cite}

\long\def\comment#1{}

\newfont{\bbb}{msbm10 scaled 800}

\newfont{\bb}{msbm10 scaled 1100}
\newcommand{\CC}{\mbox{\bb C}}
\newcommand{\PP}{\mbox{\bb P}}
\newcommand{\RR}{\mbox{\bb R}}

\newcommand{\ZZ}{\mbox{\bb Z}}

\newcommand{\cv}{{\bf c}}

\newcommand{\qv}{{\bf q}}

\newcommand{\sv}{{\bf s}}

\newcommand{\uv}{{\bf u}}

\newcommand{\vv}{{\bf v}}
\newcommand{\xv}{{\bf x}}
\newcommand{\yv}{{\bf y}}
\newcommand{\zv}{{\bf z}}

\newcommand{\Am}{{\bf A}}

\newcommand{\Fm}{{\bf F}}
\newcommand{\Gm}{{\bf G}}
\newcommand{\Hm}{{\bf H}}

\newcommand{\Pm}{{\bf P}}
\newcommand{\Qm}{{\bf Q}}

\newcommand{\Um}{{\bf U}}

\newcommand{\Vm}{{\bf V}}

\newcommand{\Ec}{{\cal E}}

\newtheorem{thm}{Theorem}%[section]
\newtheorem{lemma}{Lemma}%[section]
\newtheorem{defn}{Definition}%[section]
\newtheorem{remark}{\indent \bf Remark}%[section]

\begin{document}

\title{Cellular Interference Alignment}

\author{Vasilis~Ntranos$^\dagger$\thanks{This work is the outcome of a collaboration that started while V. Ntranos was a research intern at Bell Labs, Alcatel-Lucent. Emails: ntranos@usc.edu,
mohammadali.maddah-ali@alcatel-lucent.com, caire@usc.edu.}, 
        Mohammad~Ali~Maddah-Ali$^\ast$,  and
        Giuseppe~Caire$^\dagger$ \\
        $^\dagger$University of Southern California, Los Angeles, CA,  USA\\ 
				$^\ast$Bell Labs, Alcatel-Lucent, Holmdel, NJ, USA }
        
\maketitle

%\begin{abstract}
%Interference alignment promises that in interference channels, each link can support half of a degree of freedom (DoF) per pair of transmit-receive antennas. However, this surprising theoretical result is achieved asymptotically with exponentially-long delay and/or exponentially-high signal to noise ratio, with no gain in practical settings.  Here, we aim to propose a communication scenario in wireless cellular systems, where the promised DoF gain of alignment is achieved without using time expansion or lattices. We focus on uplink cellular systems,  where each cell is split into three sectors and assume that interference is generated locally between transmitters and receivers of neighboring cells. We  consider a message passing network architecture, in which nearby sectors can exchange  already decoded messages and propose an alignment solution that can achieve the optimal DoF in this framework.
%To avoid  signaling schemes relying on the strength of interference for communication, we further introduce the notion of \emph{topologically robust} schemes, which are able to guarantee a minimum rate (or degrees of freedom) no matter if the interference link are strong or weak. Towards this end, we design an alignment scheme for cellular systems which is topologically robust  and still achieves the same optimum performance.
%\end{abstract}

\begin{abstract}
Interference alignment promises that, in Gaussian interference channels, 
each link can support half of a degree of freedom (DoF) per pair of transmit-receive antennas. 
However,  in general, this result requires to precode the data bearing signals over a signal space of asymptotically large diversity, e.g., 
over an infinite number of dimensions 
%in time-frequency 
for time-frequency varying fading channels, 
or over an infinite number of rationally independent signal levels, in the case of time-frequency invariant channels. 
In this work we consider a wireless cellular system scenario where the promised optimal DoFs are achieved with linear precoding in 
one-shot (i.e., over a single time-frequency slot). 
We focus on the uplink of a symmetric cellular system, where each cell is split into three sectors with orthogonal intra-sector multiple access. 
In our model, interference is ``local'', i.e., it is due to transmitters in neighboring cells only. We consider a message-passing backhaul network architecture,  in which nearby sectors can exchange already decoded messages and propose an alignment solution that can achieve the optimal DoFs.
To avoid  signaling schemes relying on the strength of interference, we further introduce the notion of \emph{topologically robust} schemes, 
which are able to guarantee a minimum rate (or DoFs) irrespectively of the strength of the interfering links. Towards this end, we design an alignment scheme  
which is topologically robust  and still achieves the same optimum DoFs.
\end{abstract}


%%%%%%%%%%%%%%%%%%%%%%%%%%%%%%%%%%%%%%%%%%%%%%%
\section{Introduction} \label{intro}

Interference is the dominant limiting factor in the performance of today's wireless networks. 
Recent theoretical results \cite{cj08,mgmk09,ergodic} have shown that transmission schemes based on {\it interference alignment} 
\cite{mmk08,cj08} are able to provide half of the Degrees of Freedom (DoFs) of the interference-free rates\footnote{
In our context, degrees of freedom are defined in Section \ref{sec:probstate}.}
to each user in the network.  
While these results promise significant gains compared to conventional interference mitigation techniques, the extent to which such gains 
can be realized in practice has been so far limited. 

Most  interference alignment schemes are either restricted to networks with a small number of users (typically three transmit-receive pairs) \cite{mmk08,sht11,katabi09,sy13isit}, or rely  infinite channel diversity/resolution (e.g., through symbol extensions) for the more general cases  \cite{mmk08,cj08,mgmk09,sy13}.  In fact, it has been shown that without symbol extensions, the DoF gain of any linear interference alignment scheme  in a fully-connected network vanishes as the number of users increases~\cite{ygjk10, Razaviyayn, Bresler}. 
%Asymptotic interference alignment has not shown gain in practical Signal-to-Noise ratios (SNR). 
On the other hand, splitting the network into smaller sub-networks does not seem to be the solution either, as the remaining interference  between sub-networks can eliminate the potential gain of 
 interference alignment.  

%Most of interference alignment schemes are either restricted to networks with a small number of users (typically three transmit-receive pairs), or rely on infinite channel diversity/resolution (symbol extensions, rational-irrational gains, etc) for the more general cases. In fact, it has been shown that without symbol extensions, the DoF gain of any linear interference alignment scheme  in a fully-connected network vanishes as the number of users increases~\cite{Razaviyayn, Bresler}. 

Another class of interference management techniques for wireless systems relies on utilizing the backhaul connections  in  order to enable cooperation between base stations. 
For the uplink, it is assumed that  all base stations can share their received signal samples over the backhaul of the network and then jointly decode the corresponding user messages. Similarly, for the downlink, it is assumed that all user messages can be shared across the entire network, so that base stations can cooperatively transmit the messages to the corresponding users and manage the interference. This technique, often referred to as "Network MIMO" in the literature \cite{kfv06,fkv06,multicell10,htc12}, effectively reduces the system to a (network-wide) multiple-antenna multiaccess channel for the uplink, or a multiple-antenna broadcast channel \cite{cs03,wss06} for the downlink. In an effort to reduce the significant backhaul load requirements of the above technique, limited base station collaboration has also been considered for the downlink \cite{zcagh09, sw11, UIUC11, UIUC12, UIUC14} where user message information is locally shared within smaller clusters of the network,   and for the uplink \cite{lsw07, venkat07, ssps09, llsw12} where local receiver collaboration is enabled by sharing sampled (or quantized) received signals under backhaul connectivity (or capacity) constraints.  

%\begin{table}
%\begin{center}
%    \begin{tabular}{ | l    p{8cm} |}
%    \hline
%      \hline
%      
%     
%    
%   & $\bullet$ Hexagonal network topology with three sectors per cell \hfill \\
%    
%    \hspace{0.2in}{\normalsize Cellular Network} & $\bullet$ Cellular uplink with any number of transmit-receive pairs  \hfill 
%    
%    
%   
%    
%    $\bullet$ Interference generated from all neighboring sectors \\  \hline
%    
%     & 
%    
%    $\bullet$ Linear beamforming strategies without symbol extensions. \hfill \\
%    
%  \hspace{0.1in}{\normalsize Interference Alignment}  & $\bullet$ Constant (frequency-flat) channel gains. \hfill \\  
%  & $\bullet$ Local channel-state information. \hfill \\ \hline
%    
%    
%     & 
%    
%    $\bullet$ Collaboration restricted only to already decoded (informa-\\
%     \hspace{0.05in} {\normalsize Backhaul Infrastructure} & tion) messages; if a sector receiver can decode its own user's message, it can share it with its neighboring sector receivers.  \\ \hline 
%     \hline
%       
%     
%    \end{tabular}
%\end{center}\vspace{-0.25in}
%\caption {Overview of the Cellular Interference Alignment framework.\vspace{-0.25in}}
%\end{table}
%
% \vspace{-0.25in}
Here we propose a framework that can take advantage of the  {\em partial connectivity} of extended\footnote{Following \cite{extendedkumar04,extendedtelatar05,extendedtse10} 
we refer to an ``extended'' network as a network with a fixed spatial density of cells and increasing total coverage area, in 
contrast to a ``dense'' network where the total coverage area is fixed and the cell density increases.}
cellular networks and  
%lead to  efficient and scalable interference alignment schemes. Our results  
provide insights and guidelines for the design of the next generation advanced inter-cell interference management in wireless systems. Within our framework, we are interested in the design of interference alignment schemes for cellular networks 
with the following  three basic principles in mind.

\begin{itemize}
\item{\it Scalability:} The overall performance of the scheme should materialize irrespectively of the size of the cellular network, i.e., when the number of 
number of transmit-receive pairs becomes arbitrarily large. 

\vspace{0.05in}
 
\item{\it Locality:} The transmission scheme should  operate under local information exchange, and
exploit the distributed nature of the cellular network.  
%(e.g., channel state, limited cooperation, etc) that can eventually lead to a practical system implementation.

\vspace{0.1in}

\item{\it Spectral Efficiency:} The scheme should aim for high spectral-efficiency by allowing more (interference-free) parallel transmissions 
to take place within the same spectrum. 
\end{itemize}
\vspace{0.05in}

In this paper we focus on the uplink of a sectored cellular system. Hence, receivers are located at the base station sites.  Motivated by results embraced in practice (see \cite{NICE} for an example),  we assume that if a sector receiver can decode its own user's message, it can share it with its neighboring sector receivers. This can be easily done for
sectors located in the same base station site (co-located) and it can also be done with today's technology and moderate infrastructure effort 
through local backhaul connections to neighboring cells. In particular, we show that this {\it local} and {\it one directional} data exchange --- 
restricted only to decoded messages ---  is enough to reduce the uplink of a sectored cellular network to a topology in which the optimal degrees 
of freedom can be achieved without requiring time-frequency expansion or lattice alignment. Notice that in the proposed architecture we do not require
that the sector receivers share received signal samples and/or perform joint decoding of multiple user messages, 
in contrast with existing works on ``distributed antenna systems'' and the popular and widely studied ``Wyner model'' \cite{wyner94,shamaiwyner97} for cellular systems. 
We emphasize that locally sharing decoded information messages over the backhaul and restricting to single-user decoding 
can be easily implemented within the current technology.

\vspace{0.05in}

In general, in coordinated cell processing strategies, there is always the risk that the signaling scheme 
relies on the strength of interference in order to achieve reliable communication. 
However,  practical systems are not designed to guarantee that strength. 
On the contrary, current system deployment is geared to making interfering links as weak as possible. Hence, a scheme that relies on ``strong interference'' links
would fail if applied to a system which was designed according to the current design guidelines.  In order to address this issue,  we introduce the 
concept of \emph{topological robustness}, where the goal is to design communication schemes that can maintain a minimum rate 
(or degrees of freedom) no matter if the interference links are strong or weak. In particular, we show that such schemes exists in our framework 
and prove their optimality using a compound network formulation.

This paper is organized as follows. First, in Section \ref{problem} we describe the cellular model that we consider in this work and give a formal problem statement. Then, in Section \ref{nointracell}  we state our results for networks with no intra-cell interference and give the corresponding achievability and converse theorems.
In Section \ref{sec:intracell} we extend our model to incorporate both out-of-cell and intra-cell interference and in Section \ref{sec:topo} we focus on the design and optimality of topologically robust transmission schemes. Finally,
we conclude this paper with Section \ref{conclusions}.

%%%%%%%%%%%%%%%%%%%%%%%%%%%%%%%%%%%%%%%%%%%%%%%%%%%%%%%%%%%%%
\section{Problem Formulation}  \label{problem}

\subsection{Cellular Model}

Consider a large multiple-input multiple-output (MIMO) cellular network with three sectors per cell. 
As in current 4G cellular systems \cite{sesia2009lte}, orthogonal intra-sector multiple access is used in the uplink, such that, 
without loss of generality, we can consider a single user per sector, as shown in Fig.~\ref{cell}. 
Within each sector, the receiver is interested in decoding the uplink message of the user associated with it 
and observes all other transmissions as interference. %[{\color{red}{frequency reuse 1}}].
We consider here a symmetric  configuration in which  all transmitters and  receivers in the network are equipped with $M$ 
antennas each and assume frequency-flat channel gains that remain constant throughout the entire 
communication.

 \begin{figure}[t]
        \centering
        \begin{subfigure}[b]{0.5\columnwidth}
                \centering
                \includegraphics[width=\columnwidth]{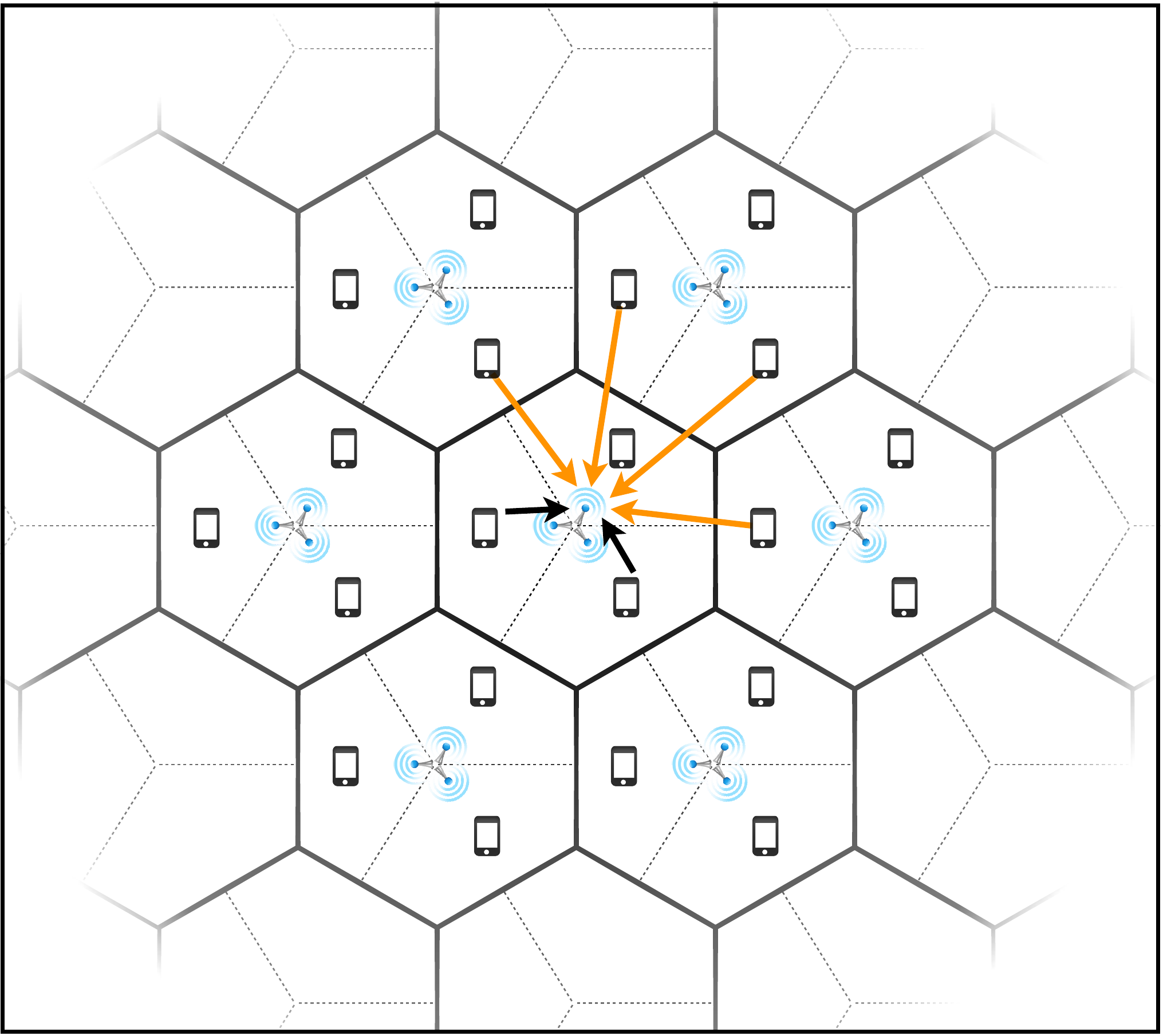}
                \caption{Cellular network}
\label{cell}
        \end{subfigure}%
        ~       \begin{subfigure}[b]{0.5\columnwidth}
                \centering
                \includegraphics[width=\columnwidth]{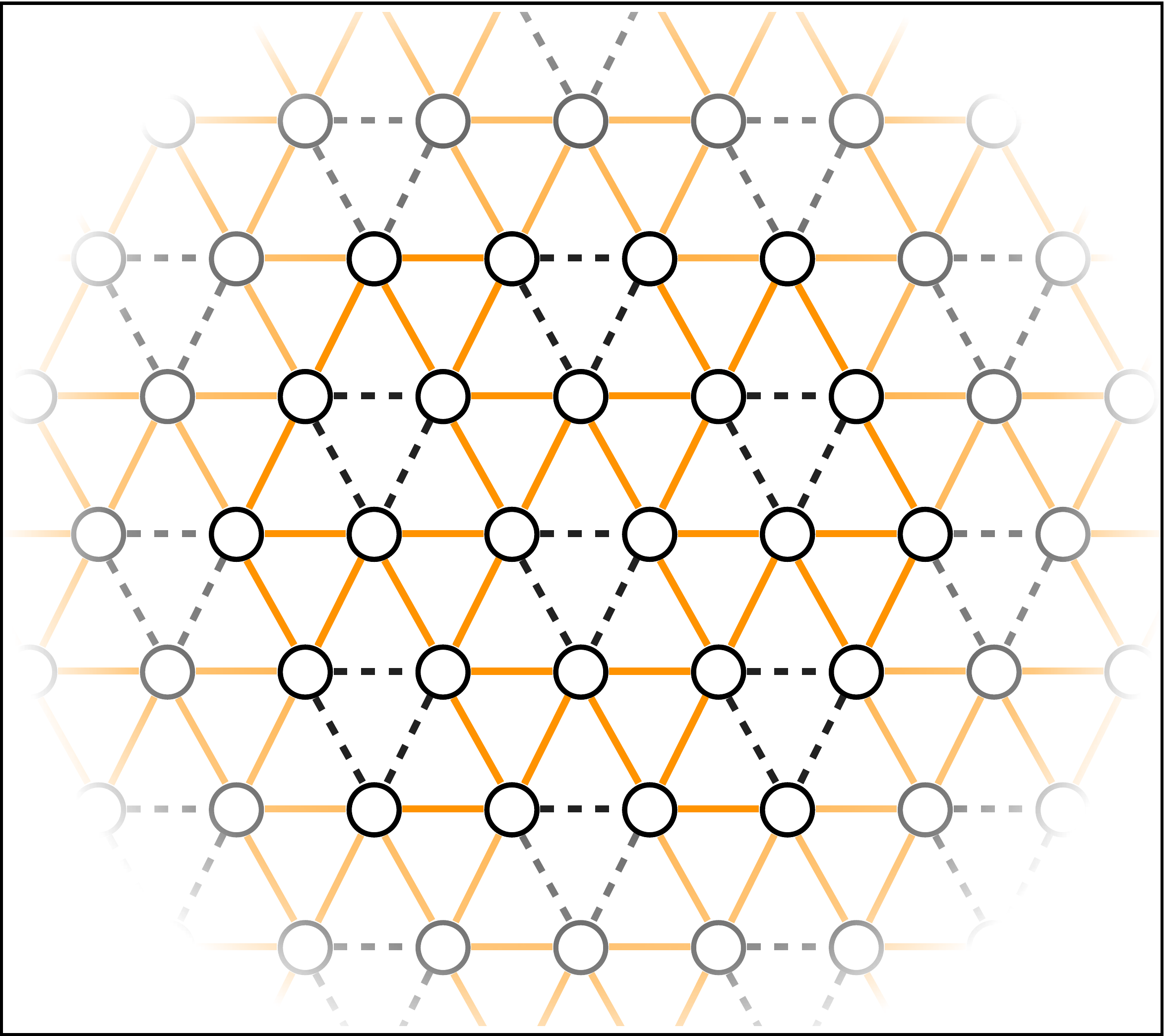}
\caption{Interference graph}
\label{fig-igraph}
        \end{subfigure}
        
       \caption{ The cellular network topology and the corresponding interference graph: we consider the uplink of a MIMO cellular network with $120^{o}$ sector receivers as depicted in Fig.~\ref{cell}. Each receiver is interested in decoding the message of the mobile terminal associated with it and observes all other transmissions as interference. In our cellular model we assume four dominant sources of interference for each sector shown  as orange arrows originating from its closest out-of-cell transmitters. Interference  between the sectors of the same cell is depicted with black arrows. Fig.~\ref{fig-igraph} 
       shows the corresponding interference graph by taking into account all interfering links in a given cellular network, in which vertices represent transmit-receive pairs within sectors and edges indicate interfering neighbors. The dashed black edges in the above graph correspond to interference between sectors of the same cell that we are going to ignore until Section \ref{sec:intracell}.}\label{fig:1}
       \vspace{-0.2in}
\end{figure}

Because of shadowing effects and distance-dependent pathloss, that are inherent to wireless communications \cite{MolischBook},  we assume that the interference seen at each receiver is generated {\it locally}, by transmitters located in neighboring sectors.\footnote{
In practice, the aggregate effect of non-neighboring transmitters contributes to the ``noise floor'' of the system. 
In \cite{gnaj13}, necessary and sufficient 
conditions on the channel gain coefficients of a Gaussian $K$-user interference channels are found such that ``treating interference as noise'' (TIN) is 
approximately optimal in the sense that, subject to these conditions, the TIN-achievable region is within an SNR-independent gap of the capacity region. 
%In the context of this paper, we may imagine that such TIN optimality conditions apply to non-neighboring sectors, such that wetake care explicitly only of neighboring sectors in the proposed interference alignment strategy.
}

%We will further assume in the following sections that transmit-receive pairs located in sectors of the same cell do not interfere with each other. This modeling assumption can be motivated by taking into account the physical orientation and radiation patterns of the receive antennas used in sectored cellular systems: 
%The received signal power from interfering users of the same cell is typically much less than  
%the interference power observed from out-of-cell users located in a sector's line of sight. Note 
%however that in Section \ref{sec:intracell} we are going to lift this assumption and focus on a  cellular model that  incorporates both out-of-cell and intra-cell interference. 
Let $\mathcal S$ be the sector index set and let $\mathcal{N}(i)$ denote the set of the interfering neighbors of the $i$th sector. 
The received signals in our model can be written as
\begin{equation}
\yv_{i} = \Hm_{ii}\xv_{i} + \sum_{j\in {\mathcal N}(i)}\Hm_{ij}\xv_{j} + \zv_{i},\,\; i\in {\cal S}
\end{equation}
where $\Hm_{ij}$ is the $M\times M$ matrix of channel gains between the transmitter (user terminal) associated with sector $j$ and the receiver of sector $i$ and  $\xv_{i}$ are the corresponding transmitted signals satisfying the average power constraint $\mathbb{E}\big[||\xv_{i}||^{2}\big]\leq P$.

%Fig.~1a shows the corresponding . Each sector observes four dominant interfering links depicted as orange arrows originating from its closest out-of-cell neighboring transmitters, and two interfering links coming from the sectors located in the same cell.

%We further assume that sector receivers located in the same cell can jointly process their received observations.

In this paper, we will consider two  interference models based on the choice of the sets ${\cal N}(i),\,i\in \cal S$. 
In the first part, we will assume that the sectors located  in the same cell do not interfere with each other and focus only on interference generated by  nearby out-of-cell  transmitters.   This assumption can be motivated by taking into account the physical orientation and radiation patterns of the 
antennas  used in sectored cellular systems, where the interference power from users in different sectors of the same cell should be much 
less than  the interference power observed from out-of-cell users located in the sector's line of sight. Then, in Section IV, we are going to lift 
this assumption and consider the case where sector receivers observe both out-of-cell and intra-cell interference.
This extension takes into account the fact that 
%it has been observed that 
users near the sector boundary may produce significant interference
to the neighboring sector in the same cell, due to possibly non-ideal sectored antenna radiation patterns. 

%The cellular network topology we consider here is depicted in Fig.~1. Each sector observes four dominant interfering links depicted as orange arrows originating from its closest out-of-cell neighboring transmitters, and the black arrows correspond to interference between sectors of the same cell. 

\subsection{Interference Graph}\label{igraph}

A useful representation of our cellular model can be given by the corresponding {\it interference graph} ${\cal G}({\cal V},{\cal E})$ shown in Fig.~\ref{fig-igraph}. In this graph, vertices represent transmit-receive pairs within each sector and  edges indicate interfering neighboring links: 
the transmitter associated with a node $u\in \cal V$ causes interference to all receivers associated with nodes $v\in \cal V$ if there is an 
edge $(u,v)\in \cal E$. Notice that  the interference graph is undirected and hence interference between sectors in our model goes 
in both directions. 
%That is, in our cellular model, transmit-receive pairs associated with an edge $(u,v)\in E$ correspond to sectors that interfere with each other. 

%More formally, we can define the interference graph $\cal G (\cal V, \cal E)$ as follows. First, we are going to define the set $\cal V$ by describing an one-to-one mapping between the vertices of the graph and a set of complex numbers that we will refer to as node labels. 
%The real and imaginary parts of these labels can be interpreted  as the coordinates of the corresponding nodes embedded on the complex plane in a way that resembles the specific sector layout of our cellular system. 
%%Then we will describe the set of edges ${\cal E} = \{(u,v) : u,v\in\cal V\}$  as a function of the corresponding node labels. 
%%We will later use the properties of these numbers to refer to nodes in $\cal V$ and describe the topology of our cellular model. 
%A natural choice for this labeling is the set of the Eisenstein integers $\mathbb{Z}(\omega)$ that exhibits the triangular lattice structure shown in Fig.~\ref{eisen}.  

%It is convenient to represent the interference graph $\cal G(V,E)$ by identifying $\cal V$ with a set of points on the complex plane whose coordinates are referred to as the node labels. The geometry of such labels in the complex plane corresponds to the hexagonal lattice layout of the sectors. A natural choice for this labeling is the set of the Eisenstein integers $\mathbb{Z}(\omega)$ shown in Fig.~\ref{eisen}.
%

More formally, we can define the interference graph $\cal G (\cal V, \cal E)$ as follows. First, we are going to define the set $\cal V$ through a 
one-to-one mapping between the vertices of the graph and a set of complex numbers that we will refer to as node labels. 
The real and imaginary parts of these labels can be interpreted  as the coordinates of the corresponding nodes embedded on the complex 
plane in a way that resembles the specific sector layout of our cellular system. 
A natural choice for this labeling is the set of the Eisenstein integers $\mathbb{Z}(\omega)$ that exhibits the hexagonal lattice 
structure shown in Fig.~\ref{eisen}.

\begin{defn}[Eisenstein integers]
The set of Eisenstein integers, denoted as $\mathbb{Z}(\omega)$, is formed by all complex numbers of the form 
$z=a+b\omega$, where $a,b \in \mathbb{Z}$ and $\omega = \frac{1}{2}(-1 + i\sqrt{3})$. \hfill $\lozenge$
\end{defn}

\begin{figure}[ht]
                \centering
                \includegraphics[width=.6\columnwidth]{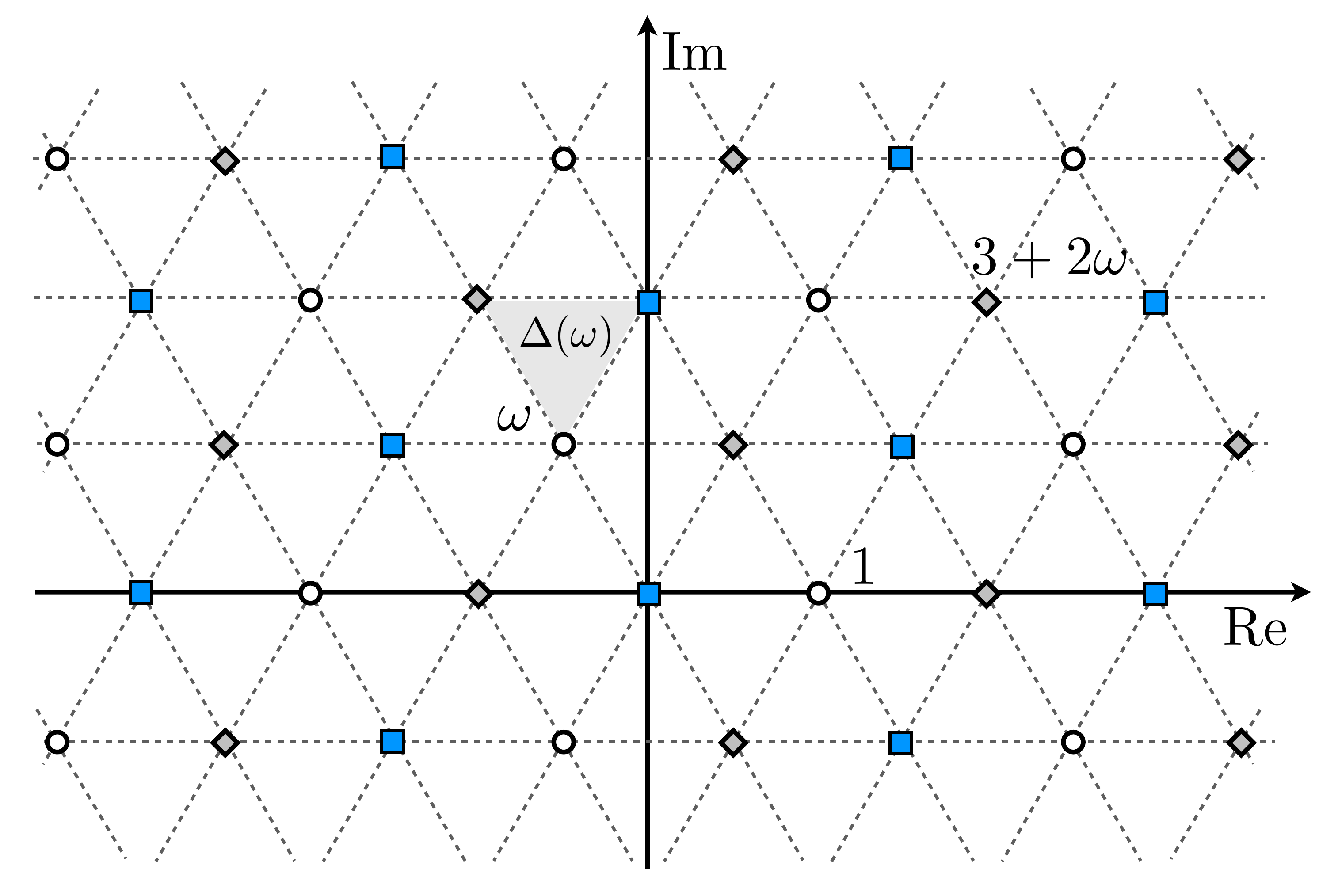}
                \caption{ The Eisenstein integers $Z(\omega)$ on the complex plane. }
                \label{eisen}

\end{figure}

Define ${\cal B}_{r} \triangleq \{z\hspace{-0.05in }\in \hspace{-0.05in }\mathbb{C}\hspace{-0.05in } :\hspace{-0.02in } |{\rm Re}(z)|\leq r , |{\rm Im}(z)|\leq \frac{\sqrt{3}r}{2}\}$
and let $\phi: {\cal V} \rightarrow \mathbb{Z}(\omega)\cap{\cal B}_{r}$ be an one-to-one mapping between the elements of $\cal V$ and the set of bounded Eisenstein integers given by $\mathbb{Z}(\omega)\cap{\cal B}_{r}$. For any $v\in \cal V$ we say that $\phi(v)$ is the label of the corresponding vertex in 
the interference graph. Correspondingly, the set of vertices $\cal V$ is given by
\begin{equation}
{\cal V} = \left\{ \phi^{-1}(z) : z\in \mathbb{Z}(\omega)\cap{\cal B}_{r}\right\}.
\label{eq:V}
\end{equation}
We now explicitly describe the set of edges ${\cal E}$ in the interference graph in terms of the function $\phi$. 
Consider the set of three segments in $\CC$
\begin{equation*}
{\Delta}(z) = \{(z,z+\omega),\,(z,z+\omega+1),\,(z+\omega,z+\omega+1) \}
\end{equation*} 
and define the set%$\cal D$
\begin{equation}\label{def:D}
{\cal D}\triangleq \hspace{-0.2in}\underset{\substack{ a,b\in\mathbb{Z}: \\ \left[a+b\right]{\rm mod}\,3 \neq 0}}{\bigcup}\hspace{-0.2in}{\Delta}(a+b\omega)
\end{equation} 
to be the union of ${\Delta}(a+b\omega)$ over all $a,b\in \mathbb{Z}$ such that 
%to be the union of ${\cal T}(z)$ over all $z = a+b\omega \in \mathbb{Z}(\omega)$ such that 
$\left[a+b\right]{\rm mod}\,3 \neq 0$. 
Observe that the segments in $\Delta(z)$ form a triangle with vertices in the Eisenstein integers $z$, $z+\omega$ and $z+\omega+1$, 
as shown in Fig.~\ref{eisen}. The function $f(a+b\omega)\triangleq\left[a+b\right]{\rm mod}\,3$ partitions the hexagonal lattice $\mathbb{Z}(\omega)$ into 
three cosets. In particular, all points $z$ such that $f(z)= 0$ form a sublattice $\Lambda_0$ of $\mathbb{Z}(\omega)$, 
and the points $z$ for which $f(z) = 1$ and $f(z) = 2$ corresponds to its cosets $\Lambda_0 + 1$ and $\Lambda_0 - 1$. 
In Fig.~\ref{eisen}, the points of $\Lambda_0$, $\Lambda_0 + 1$ and $\Lambda_0 - 1$ are shown with squares, circles and diamonds, respectively. 
Without loss of generality, we assume that for all $z\in \Lambda_0$ the segments in $\Delta(z)$ correspond to links  between the three sectors of the 
same cell. Hence, under the assumption that such sectors do not interfere,  we exclude the corresponding 
$\{\Delta(z) : z \in \Lambda_0\}$ in the definition of $\cal D$ in (\ref{def:D}). % $\cal D$. 
Eventually, the set of edges ${\cal E}$ representing out-of-cell interference is given by
\begin{equation}
{\cal E} = \left\{(u,v) : u,v\in{\cal V} \mbox{ and} \left(\phi(u),\phi(v)\right)\in {\cal D}\right\}.
\label{eq:E}
\end{equation}

\begin{defn}[Interference Graph]
The out-of-cell interference graph $\cal G (\cal V, \cal E)$ is an undirected graph defined by the set of vertices $\cal V$ given in (\ref{eq:V}) and the corresponding set of edges $\cal E$ given in (\ref{eq:E}). The graph vertices represent transmit-receive pairs in our cellular model and edges indicate interfering neighbors.
\hfill $\lozenge$
\end{defn}

\subsection{Network Interference Cancellation}\label{sec:NICE} \vspace{-0.1in}

We further consider a message-passing network architecture for our cellular system, in which  sector receivers  communicate locally in order  to exchange  decoded messages.
%We further consider an asynchronous decoding framework for our cellular model in which transmitted messages are decoded sequentially.
%We further consider a network-wide successive interference cancellation scheme, where user messages decoded at a base station are passed to neighboring base stations in a pipelined manner.
%Within this framework, we assume that  sector receivers  communicate locally only to exchange  decoded messages;  
Any  receiver that has already decoded its own user's message can use the backhaul of the network and  pass it as side information to one or more of its neighbors. 
In turn, the neighboring sectors can use the received decoded messages in order to  reconstruct the corresponding interfering  
signals and subtract them from their observation. It is important to note that this scheme only requires sharing (decoded) information messages between sector receivers and does not require sharing the baseband signal samples, which is much more demanding for the backbone network.
%
%can send it to one or more of its neighbors  over the backhaul of the network.  The neighboring sectors can then use the received decoded message to reconstruct the corresponding interfering  signal and subtract it from their observation before proceeding to decode their own desired message. 
%We further consider in our cellular model that neighboring sector receivers can exchange already decoded messages over the backhaul network lines. We assume that communication over the backhaul can be performed in a much shorter time scale, such that a receiver that has already decoded its own user's message can send it to one or more of its neighbors almost instantaneously. The neighboring sectors can then  use the received decoded message to reconstruct the corresponding interfering  signal and subtract it from their observation before proceeding to decode their own desired message. 

The above operation effectively cancels  interference in one direction: 
all decoded messages propagate through the backhaul of the network, successively eliminating certain interfering 
links between neighboring sectors according to a specified decoding order.  
Fig.~\ref{NICE} illustrates the above network interference cancellation process in our cellular graph model assuming a ``left-to-right, top-down'' 
decoding order.  Notice that edges are now {\it directed} in order to indicate the interference flow over the network. 
For example, if an undirected edge $(u,v)$ exists in $\Ec$ and, under this message-passing architecture, 
node $v$ decodes its message before node $u$ and passes it to node $u$ through the backhaul, 
then the resulting interference graph will contain the directed link $[u,v]$, indicating that the interference is from node (sector)  $u$
node (sector) $v$ only.
\clearpage

\begin{figure}[ht]
                \centering
                \includegraphics[width=0.8\columnwidth]{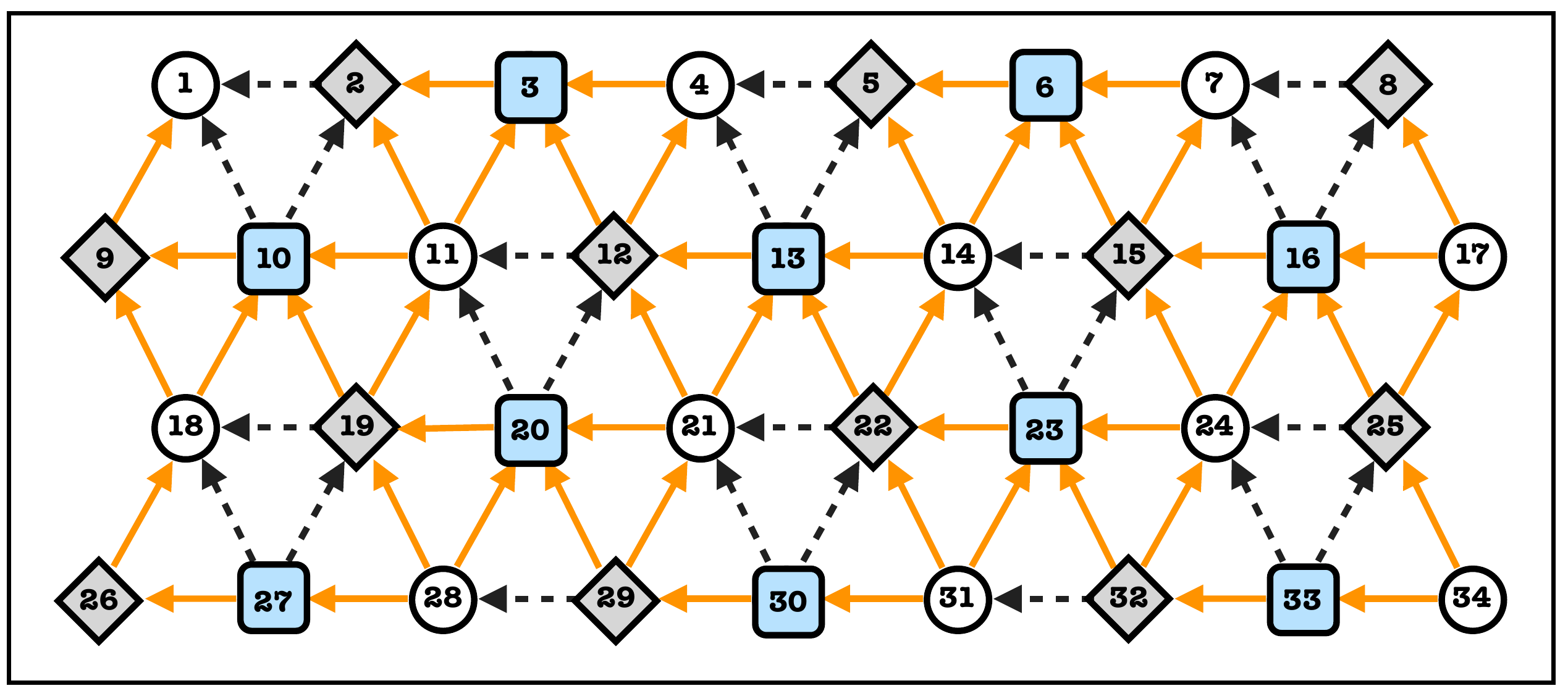}
                \caption{ The directed interference graph $\cal G_{\pi^{*}}(V,E_{\pi^{*}})$ after network interference cancellation according to the ``left-to-right, top-down'' decoding order $\pi^{*}$. The transmitter of a sector associated with node $i$ causes interference only to its neighboring sector receivers $j$ with $j<i$. Orange arrows indicate out-of-cell interference while black dashed arrows correspond to interference from within each cell. }
                \label{NICE}

\end{figure}

A decoding order $\pi$ can be specified by defining a %(strict) 
partial order ``$\prec_{\pi}$'' over the  set of vertices $\cal V$ in our interference graph.  Then, the message of the user associated with vertex $v \in \cal V$ will be decoded before the one associated with vertex $u \in \cal V$ if $v \prec_{\pi}u$. In principle, we can choose any decoding order that partially orders the set $\cal V$ and hence $\pi$ can be treated as an optimization parameter in our model. 

\begin{defn}
[Directed Interference Graph ${\cal G}_{\pi}$] For a given partial order ``$\prec_{\pi}$'' on $\cal V$, the directed interference graph is defined as ${\cal G}_{\pi}({\cal V}, {\cal E}_{\pi})$ 
where 
${\cal E}_{\pi}$ is a set of ordered pairs $[u,v]$ given by 
$
{\cal E}_{\pi} = \left\{[u,v] : (u,v) \in {\cal E} \mbox{ and }  v\prec_{\pi}u   \right\}
$. \hfill $\lozenge$
\end{defn}

Next, we formally specify the ``left-to-right, top-down'' decoding order $\pi^{*}$ that has been chosen in Fig.~\ref{NICE}. As we will show in the following section,  this decoding order is indeed optimum and can lead to the maximum possible DoF per user in large cellular networks.
\vspace{0.3in}

\begin{defn}[The Decoding Order $\pi^{*}$]
The ``left-to-right, top-down'' decoding order $\pi^{*}$ is defined by the partial ordering $\prec_{\pi^{*}}$ over $\cal V$ such that for any $u, v\in \cal V$, $v\prec_{\pi^{*}}u \Leftrightarrow$
\begin{equation*}
 \begin{cases} {\rm Im}\left(\phi(v)\right) > {\rm Im}\left(\phi(u)\right) \mbox{, or} \\ 
{\rm Im}\left(\phi(v)\right) = {\rm Im}\left(\phi(u)\right) \mbox{and}\;  {\rm Re}\left(\phi(v)\right) < {\rm Re}\left(\phi(u)\right)
\end{cases}
\end{equation*}
%where ${\rm Im}(z)$ and ${\rm Re}(z)$ denote the imaginary and real parts of $z\in \mathbb{C}$ . 
\hfill $\lozenge$

\end{defn}

\subsection{Problem Statement}\label{sec:probstate}
\vspace{-0.15in}
Our main goal is to design efficient communication schemes for the cellular model previously introduced.
As a first-order approximation of a scheme's efficiency, we will consider here the achievable DoFs, broadly defined as the number of point-to-point interference-free channels that can be created between transmit-receive pairs in the network.

More specifically, we are going to limit ourselves to linear beamforming strategies over multiple antennas assuming constant (frequency-flat) 
channel gains  without allowing symbol extensions. We refer to such schemes as ``one-shot'', indicating that precoding is achieved over a single
time-frequency slot (symbol-by-symbol).  Our goal it to maximize, over all decoding orders $\pi$, the average (per sector) achievable DoFs
\vspace{-0.1in}
\begin{equation}
d_{{\cal G},{\pi}} \triangleq \frac{1}{|{\cal V}|} \sum_{v\in \cal V}d_{v}\, ,
\end{equation}
where  $\cal G(V,E)$ is the  interference graph defined in Section \ref{igraph} and 
$d_{v}$ denotes the DoFs achieved by the transmit-receive pair associated with the node 
$v\in\cal V$, where 
$$d_{v} = \lim_{P\rightarrow \infty}\frac{R_{v}(P)}{\log(P)},$$ 
and $R_{v}(P)$ is  the achievable rate in sector 
$v\in \cal V$ under the per-user transmit power constraint $P$.

\vspace{-0.1in}
%\section{\color{red} Main Results}
\section{Networks with No Intra-cell Interference}\label{nointracell}
\vspace{-0.1in}
Here we state our main results for the case where there is no interference between the sectors of the same cell. It is worth pointing out that in this section we do not assume any form of collaboration between sector receivers other than the message passing  scheme described in Section~\ref{sec:NICE}. 
The main results of this section are given by the following achievability and converse theorems. The complete proof of these results is provided in 
Appendices \ref{proofthm1} and \ref{proof:thm2}.  
For the sake of clarity and in order to build intuition on both the achievability coding scheme and the converse proof technique, 
 we treat in detail the case of two-antenna terminals ($M = 2$) in Sections \ref{bfscheme} and \ref{sec:converse1}.
%following theorems show that [...]
%{\color{blue}[State the main results for sectors MxM and motivate the rest of the paper. Theorem 1: Achievability and Theorem 2: Converse.
%Say why this result is important.]}
\begin{thm} \label{thm1}
For a sectored cellular system $\cal G(V,E)$ in which transmitters and receivers are equipped with $M$ antennas each, there exist a 
one-shot linear beamforming scheme that achieves the average (per sector) DoFs
\begin{equation}d_{{\cal G},{\pi^{*}}} = \begin{cases}\frac{M}{2}, \;\;\;\;\;\;\;\;\;\;\;\mbox{$M$ is even}\\ \frac{M}{2} -\frac{1}{6}, \;\;\;\;\;\mbox{$M$ is odd}\end{cases}
%\frac{M-1}{2} +\frac{1}{3}
\end{equation}
under the network interference cancellation framework with decoding order $\pi^{*}$.
%are achievable in the corresponding directed interference graph $\cal G_{\pi^{*}}$ under the network interference cancellation framework with decoding order $\pi^{*}$ . 
\hfill \QED
\end{thm}

\begin{thm}  \label{thm2}
For a sectored cellular system $\cal G(V,E)$ in which transmitters and receivers are equipped with $M$ antennas each and for any network interference 
cancellation decoding order $\pi$, the average (per sector) DoFs $d_{{\cal G},{\pi}}$  that can be achieved by any one-shot 
linear beamforming scheme are bounded by 
$\textstyle
 d_{{\cal G},{\pi^{*}}} + {\cal O}\left( \scriptstyle{1}/{{\sqrt{|{\cal V}|}}}\right)
$, where $d_{{\cal G},{\pi^{*}}}$ is given by Theorem \ref{thm1}.
\hfill \QED
\end{thm}
%and hence the achievable scheme of Theorem~1 is asymptotically optimal as $|V|\rightarrow \infty$

The above theorems yield a tight DoFs result for large extended cellular networks, for which $|\cal V| \rightarrow \infty$. 
The term  ${\cal O}\left( \scriptstyle{1}/{{\sqrt{|{\cal V}|}}}\right)$  comes from the fact that sectors on the boundary observe less interference, 
and therefore can achieve higher DoFs. However, the number of sectors on the boundary is small compared to the total number or 
sectors $|\cal V|$, and therefore, their effect vanishes as the size of the network increases. 

\begin{remark}\vspace{-0.06in}
Notice that $d_{{\cal G},{\pi^{*}}}$ is not exactly $M/2$ for odd values of $M$. 
This is because we have  insisted on one-shot schemes. By precoding over two time-frequency varying slots 
%(even without assuming time-frequency variability), 
it is not difficult to show that $M/2$ DoFs  per sector are indeed achievable also for odd $M$. %If we extend the channel over two slots we will be able to achieve $M/2$ for this case as well.
\hfill $\lozenge$
\end{remark}

%\subsection*{Cellular Interference Alignment: $2\times 2$ Example}
%For the purpose of illustrating  our main ideas, we consider here a scenario in which both sector 

%for our cellular model that is able to create one interference-free communication channel per sector across the entire  network. and hence allow all receivers to decode their desired messages. 
%
%Then, in the second part of this section, we will provide a converse argument for any linear beamforming strategy that can be used under the network interference cancellation framework and show that our proposed scheme is in fact very close to optimal in terms of spectral efficiency.

\subsection{Achievability}\label{bfscheme}
For the purpose of illustrating  our main ideas, we  will consider here  the case where sector 
receivers and mobile terminal transmitters are equipped with $M=2$  antennas and describe  the linear beamforming scheme  that is able to achieve 
one DoF per link for the entire network.

Consider the directed interference graph ${\cal G}_{\pi^{*}}({\cal V}, {\cal E}_{\pi^{*}})$ shown in Fig.~\ref{NICE} and assume that all  user terminals $v\in \cal V$  are simultaneously transmitting their signals  $\xv_{v}$ to their corresponding receivers. Recall that each sector receiver that is able to decode its own message, is also able to pass it as side information to its neighbors, effectively eliminating interference in that direction. 
Hence, following  the ``left-to-right, top-down'' decoding order $\pi^{*}$ introduced in Section \ref{sec:NICE}, the sector receiver associated with the node $u\in \cal V$ is able to  eliminate  interference from  all neighboring sectors $v\prec_{\pi^{*}}u$ 
%(by reconstructing the corresponding signals $\Hm_{uv}\xv_{v}$, $v\prec_{\pi^{*}}u$) 
and attempt to decode its own message from the two-dimensional received signal observation $\yv_{u}$ given by
\begin{equation}
\yv_{u} = \Hm_{uu}\xv_{u} + \sum_{v : [v,u]\in {\cal E}_{\pi^{*}}} \Hm_{uv}\xv_{v} +\zv_{u}.
\end{equation}
%in which all interference coming from neighboring sectors $v\prec_{\pi^{*}}u$ has been eliminated.

%Following the network interference cancellation process introduced in Section \ref{sec:NICE}, before the receiver associated with node $u\in \cal V$ attempts to decode its own desired message, it will create  a two-dimensional received signal observation $\yv_{u}$ given by
%\begin{equation}
%\yv_{u} = \Hm_{uu}\xv_{u} + \sum_{v : [v,u]\in {\cal E}_{\pi^{*}}} \Hm_{uv}\xv_{v} +\zv_{u},
%\end{equation}
%in which all interference coming from neighboring sectors $v\prec_{\pi^{*}}u$ has been eliminated. 

Our goal is to design the transmitted signals $\xv_{v}$ such that  all interference observed in  
$\yv_{u}$ is aligned in one dimension for every sector receiver $u$ in our cellular system. 
Let $\uv_{u}$ and $\vv_{u}$ denote the 2-dimensional receive and transmit beamforming vectors associated with node $u\in \cal V$ and assume that every user terminal in the network has encoded its message in the corresponding codeword. 
Although codewords span many slots (in time), we focus here on a single slot and denote the corresponding coded symbol of user $u$ by 
$s_{u}$. Then, the vector transmitted by  user $u$ 
%two-antenna array 
is given by $\xv_{u}=\vv_{u}s_{u}$ and each receiver can project its observation $\yv_{u}$ along $\uv_{u}$ to obtain
$$\hat y_{u} = \uv_{u}^{\rm H}\Hm_{uu}\vv_{u}s_{u} 
+ \sum_{v : [v,u]\in {\cal E}_{\pi^{*}}} \uv_{u}^{\rm H}\Hm_{uv}\vv_{v}s_{v} + \hat z_{u}.
$$
We will show next that it is possible to design $\uv_{u}$ and $\vv_{u}$ across the entire network ${\cal G}_{\pi^{*}}({\cal V}, {\cal E}_{\pi^{*}})$ such that the following interference alignment conditions are satisfied:
\begin{eqnarray}
&\uv_{u}^{\rm H}\Hm_{uu}\vv_{u} \neq 0,&\; \forall u\in {\cal V}  \;\;\; \mbox{and}\\
&\uv_{u}^{\rm H}\Hm_{uv}\vv_{v} =0,& \; \forall [v,u]\in {\cal E}_{\pi^{*}}.
\end{eqnarray}
Hence, each receiver in the network can decode its own desired symbol $s_{u}$ from an interference-free channel observation of the form \begin{equation} \hat y_{u} = \hat h_{u} s_{u} + \hat z_{u}\label{form}\end{equation}
where $\hat h_{u} = \uv_{u}^{\rm H}\Hm_{uu}\vv_{u}$ and $\hat z_{u} = \uv_{u}^{\rm H}\zv_{u}$.

\begin{figure}[h]

                \centering
                \includegraphics[width=0.8\columnwidth]{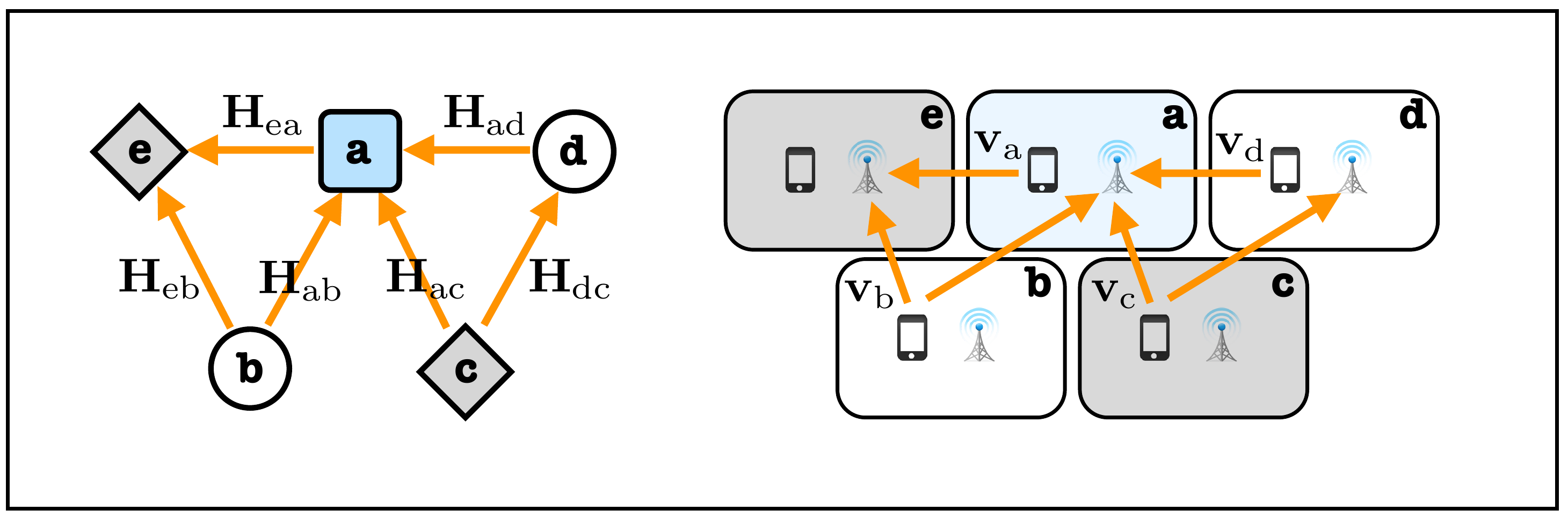}
                \caption{ Out-of-Cell Interference in the neighborhood of a blue node. Sectors are labeled here with letters to avoid confusion with the underlying decoding order (cf. Fig.~\ref{NICE}). }
\label{hood}
\end{figure}

In order to describe the alignment precoding scheme, we will partition the nodes in ${\cal G}_{\pi^{*}}({\cal V}, {\cal E}_{\pi^{*}})$ into three sets based on their interference in-degree,  defined as the number of incoming interfering links. Notice that in Fig.~\ref{NICE} all the {\it square} nodes observe at most three incoming interfering links, while the in-degrees of all {\it diamond} and {\it circle} nodes are at most two and one respectively. 
Let ${\cal V}_{\rm square} = \{v : \phi(v) \in \Lambda_0\}$, ${\cal V}_{\rm circle} = \{v : \phi(v) \in \Lambda_0+1\}$
and ${\cal V}_{\rm diamond} = \{v : \phi(v) \in \Lambda_0-1\}$ 
denote the sets of square, diamond, and circle nodes respectively,  as introduced in Section \ref{igraph}.

First we are going to propose an interference alignment solution for a small part of the network that we will refer to as the {\it neighborhood} of a square node, denoted as ${\cal S}({u})$,  $u\in {\cal V}_{\rm square}$, and then explain how this solution can be extended and applied in the entire network. Fig.~\ref{hood} shows the interfering links and transmit-receive pairs that belong to the  {neighborhood} ${\cal S}({a})$.

%Fig.~3 shows the interfering links in the {\it neighborhood} of a blue node (sector). We will describe an alignment solution for this neighborhood and then argue how this solution can be applied in the entire interference graph. 

In the above neighborhood, the goal is to design the $2$-dimensional beamforming vectors ${\bf v}_{ a}$, ${\bf v}_{ b}$, ${\bf v}_{ c}$ and ${\bf v}_{ d}$ such that all interference occupies a single dimension in every receiver. We will hence  require that 
$\mbox{span}(\Hm_{ea}\vv_{a}) = \mbox{span}({\bf H}_{ eb}{\bf v}_{ b})$ for receiver $e$ and 
$\mbox{span}(\Hm_{ab}\vv_{b}) = \mbox{span}({\bf H}_{ ac}{\bf v}_{ c}) = \mbox{span}({\bf H}_{ ad}{\bf v}_{ d})$ for receiver $a$.
These interference alignment conditions can be satisfied if we choose:
\begin{eqnarray}
{\bf v}_{ a} &\doteq& {\bf H}_{ ea}^{-1}{\bf H}_{ eb}{\bf v}_{ b}\label{Iacond1}\\
{\bf v}_{ b} &\doteq& {\bf H}_{ ab}^{-1}{\bf H}_{ ac}{\bf v}_{ c}\label{Iacond2} \\
{\bf v}_{ c} &\doteq& {\bf H}_{ ac}^{-1}{\bf H}_{ ad}{\bf v}_{ d},\label{Iacond3}
\end{eqnarray}
where ${\bf v}\doteq {\bf u}$ is a shorthand notation for ${\bf v} \in \mbox{span}({\bf u})$. 
%Indeed we can check that at receiver $e$ we have 
%$\mbox{span}(\Hm_{ea}\vv_{a}) = \mbox{span}(\Hm_{ea}{\bf H}_{ ea}^{-1}{\bf H}_{ eb}{\bf v}_{ b})
%= \mbox{span}({\bf H}_{ eb}{\bf v}_{ b})$. Similarly we can see that the conditions  $\mbox{span}(\Hm_{ab}\vv_{b}) = \mbox{span}({\bf H}_{ ac}{\bf v}_{ c}) = \mbox{span}({\bf H}_{ ad}{\bf v}_{ d})$ are also  satisfied for receiver $a$.
%
%\begin{eqnarray}
%\mbox{span}(\Hm_{ea}\vv_{a}) &=& \mbox{span}(\Hm_{ea}{\bf H}_{ ea}^{-1}{\bf H}_{ eb}{\bf v}_{ b}) \\
%&=& \mbox{span}({\bf H}_{ eb}{\bf v}_{ b})
%\end{eqnarray}
%at receiver $e$ and that
%\begin{eqnarray}
%\mbox{span}(\Hm_{ab}\vv_{b}) &=& \mbox{span}(\Hm_{ab}{\bf H}_{ ab}^{-1}{\bf H}_{ ac}{\bf v}_{ c}) \\
%&=& \mbox{span}({\bf H}_{ ac}{\bf v}_{ c}) \\ 
%&=& \mbox{span}({\bf H}_{ ac}{\bf H}_{ ac}^{-1}{\bf H}_{ ad}{\bf v}_{ d}) \\
%&=& \mbox{span}({\bf H}_{ ad}{\bf v}_{ d})
%\end{eqnarray} at receiver $a$.
Notice that in the above solution
%in equations (\ref{Iacond1})-(\ref{Iacond3}), 
the beamforming vectors ${\bf v}_{ a}$, ${\bf v}_{ b}$ and ${\bf v}_{ c}$ depend on the chosen direction for ${\bf v}_{ d}$. This is a key observation in order to embed the above beamforming strategy in the entire network.

\begin{figure}[ht]

                \centering
                \includegraphics[width=0.6\columnwidth]{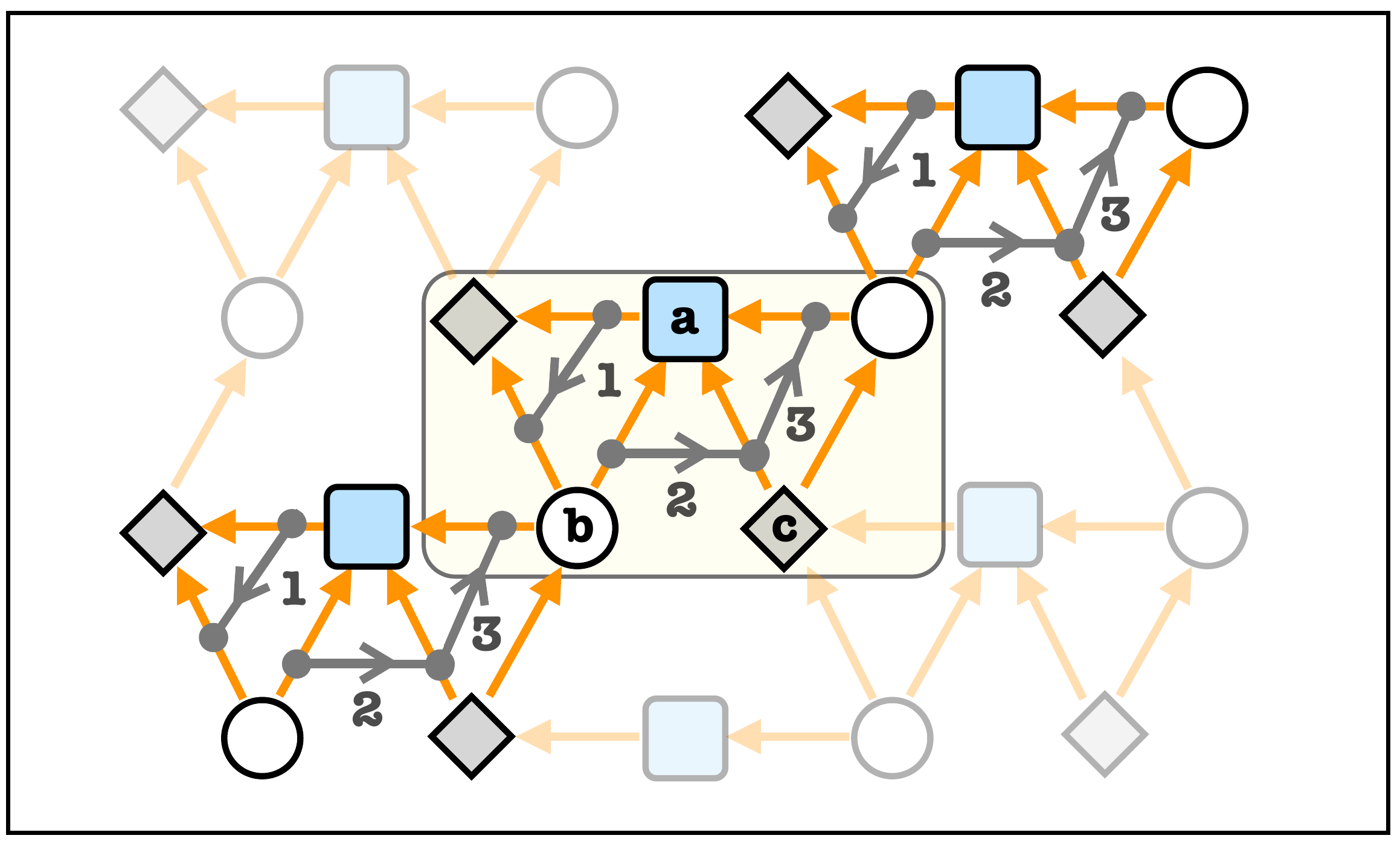}
                \caption{ Interference Alignment Scheme. The alignment conditions are depicted here with arrows connecting interfering streams that have to be aligned. The direction of the arrows show the corresponding beamforming dependencies (e.g., the  arrow labeled with the number $1$ requires that ${\bf v}_{ a}$ is chosen as a function of ${\bf v}_{ b}$).}
                \label{dependencies}

\end{figure}

All the transmitters associated with square nodes $a\in {\cal V}_{\rm square}$ can choose their beamforming vectors $\vv_{a}$ such that the first alignment condition (Eq.~\ref{Iacond1}) is satisfied in every neighborhood ${\cal S}(a)$. This beamforming choice is shown 
in Fig.~\ref{dependencies} with an arrow labeled with the number $1$, connecting the two interfering links that have to be aligned. The direction of the arrow indicates that ${\bf v}_{ a}$ has been chosen as a function of ${\bf v}_{ b}$. In a similar fashion, following the arrows labeled with the number $2$, every circle node $b\in {\cal V}_{\rm circle}$ can beamform to satisfy the second alignment condition (Eq.~\ref{Iacond2}) by choosing ${\bf v}_{ b}$ as a function of ${\bf v}_{ c}$. 
Now, in order to ensure that the third condition (Eq.~\ref{Iacond3}) is also satisfied in every neighborhood we can choose the beamforming vectors of diamond nodes $c\in {\cal V}_{\rm diamond}$ according to the arrows labeled with the number $3$, as shown in  Fig~\ref{dependencies}. 

Notice that, from each neighborhood's perspective,  ${\bf v}_{ c}$ is chosen as a function of an arbitrary vector ${\bf v}_{ d}$ that has in turn been 
chosen to satisfy an alignment condition in a different neighborhood. 
Following this procedure, all the transmitters are able design their beamforming vectors sequentially, as functions of their neighbors' choices,
starting from the boundary of the network.
%,  the above beamforming dependencies propagate through the network as shown in Fig.~\ref{dependencies}.
It is not hard to verify that with the above beamforming strategy,  every receiver in the network will observe all interference aligned 
in one dimension that can subsequently  be zero-forced in order to obtain an observation in the form of (\ref{form}). In that way,  under the network interference cancellation framework with decoding order $\pi^{*}$, all transmit-receive pairs in ${\cal G}_{\pi^{*}}({\cal V}, {\cal E}_{\pi^{*}})$ can successively 
create an {one-dimensional} interference-free channel for communication and hence achieve  $d_{v}=1$, $\forall v \in \cal V$.
\vspace{-0.4in}

\subsection{Converse}\label{sec:converse1}

In the previous section we described a linear beamforming scheme that can be applied in 
%$\cal G_{\pi^{*}}(V,E_{\pi^{*}})$  
$\cal G(V,E)$
when $M=2$, and achieve $d_{v} = 1$, for all $v\in\cal V$, following the ``left-to-right, top-down'' decoding order $\pi^{*}$.
Here, we are going to show that the above DoFs are {\it almost} optimal for our cellular network in the sense that for {any decoding order} $\pi$, the average (per sector) DoFs 
$d_{{\cal G},{\pi}}$ achievable by any linear scheme are upper bounded by $1+{\cal O}\left( \scriptstyle{1}/{{\sqrt{|{\cal V}|}}}\right)$. 

Let $\Vm_{v},\Um_{v} \in\mathbb{C}^{2\times d_{v}}$ denote the  transmit and receive beamforming matrices associated with  a node $v\in \cal V$. 
Any linear scheme that achieves the  DoFs $\{d_{v},v\in\cal V\}$ in $\cal G_{\pi}(V,E_{\pi})$ has to satisfy the  interference alignment conditions: 
\begin{align}
&\Um_{v}^{\rm H}\Hm_{vu}\Vm_{u} = 0,\; \forall [u,v] \in \cal E_{\pi}\label{eq:condA}\\
&\mbox{rank}\left(\Um_{v}^{\rm H}\Hm_{vv}\Vm_{v}\right) = d_{v},\; \forall v\in \cal V.\label{eq:condB}
\end{align}
For any receiver associated with a node $v\in \cal V$, the first condition corresponds to zero-forcing all interference from transmitters $\{u\in {\cal V}:[v,u]\in \cal E_{\pi}\}$ and the second one requires that its own desired symbols can be successfully resolved. 

\clearpage

From the above we can obtain the following {\it necessary} conditions such that any $\{d_{v}:v\in\cal V\}$ is achievable in $\cal G_{\pi}(V,E_{\pi})$:
\begin{align}
&d_{v}\in \{0,1,2\}\,,\;\forall v\in\cal V\label{eq:cond1}\\
&d_{v} + d_{u} \leq 2 \,,\;\forall [u,v]\in\cal E_{\pi},\label{eq:cond2}
\end{align}
where (\ref{eq:cond1}) follows directly from (\ref{eq:condB}) and (\ref{eq:cond2}) follows from (\ref{eq:condA}) assuming that $\mbox{rank}(\Hm_{vu}\Vm_{u}) = \mbox{rank}(\Vm_{u}) = d_{u}$, $\forall [v,u]\in \cal E_{\pi}$.

In order to obtain an upper bound on the achievable average DoFs, we shall consider the  optimization problem
\begin{align}
\hspace{-0.2in}{{ \rm Q}_{1}({{\cal G_{\pi}}})}:\;\;\;\;\; &\underset{\{d_{v}:v\in\cal V\}}{\mbox{maximize}}\;\;\;  \frac{1}{|{\cal V}|} \sum_{v\in \cal V}d_{v} \nonumber\\
&\mbox{subject to:} \;\;(\ref{eq:cond1}),(\ref{eq:cond2}).
\nonumber
\end{align}
In particular, we will derive an upper bound $\hat d_{\cal G}$ for the optimal value of ${ \rm Q}_{1}({{\cal G_{\pi}}})$, such that $\hat d_{\cal G} \geq {\rm opt}({ \rm Q}_{1}({{\cal G_{\pi}}}))$, $\forall \pi$ and show that 
$
\hat d_{\cal G} = 1 + {\cal O}\left( \scriptstyle{1}/{{\sqrt{|{\cal V}|}}}\right)$.

As a first step, we are going to rewrite the sum in the objective of 
${ \rm Q}_{1}({{\cal G_{\pi}}})$ as a sum over  connected vertex triplets $[u,v,w]$ that we are going to call the {\it triangles} $\cal T$ of our graph. 

%{\RED [NOTE: I suggest to limit the proliferation of notation. You have defined already $\Delta(z)$ as the set of segments (links) 
%corresponding to the sides of these triangles, and $\Dc$ in (3) as the set of all such segments. 
%Now, it would be sufficient to say that $\Tc$ is the set of triplets of vertices such that the corresponding  
%triangles $\Delta(z)$ satisfy $\Delta(z) \subseteq \Dc$ (or equivalently,  such that $z \in \ZZ(\omega) \setminus \Lambda_0$. While this definition does not make things dramatically simpler, 
%it avoids the introduction of yet another symbol $\cal P$, and in addition has the advantage of linking what introduced before to what it is said here, 
%such that the unitary presentation of the paper is improved]}
\begin{figure}[ht]

                \centering
                \includegraphics[width=.55\columnwidth]{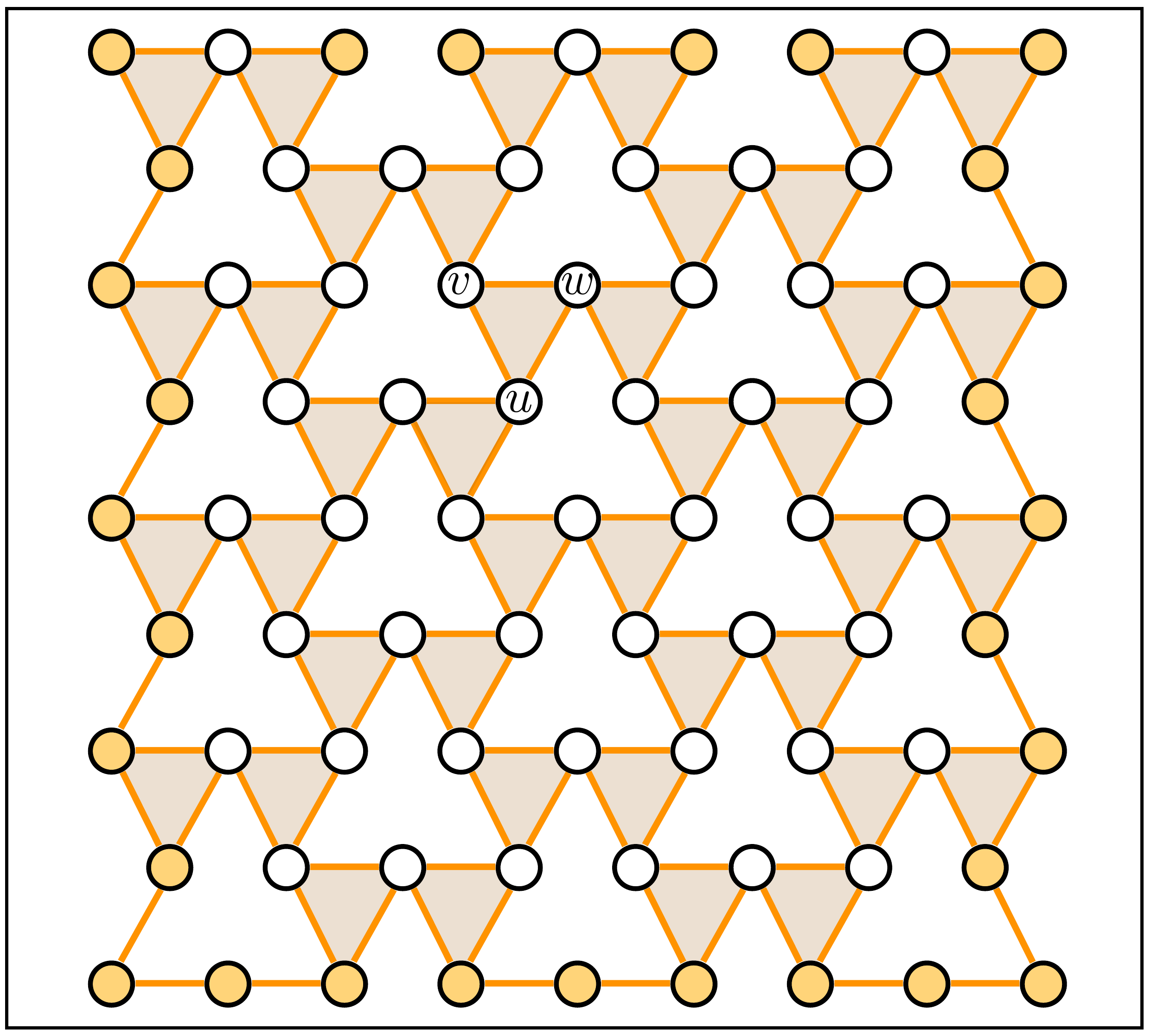}
                \caption{The set of triangles $[u,v,w]\in \cal T$ for  $\cal G(V,E)$. All the circle nodes belong to $\cal V_{\rm in}$ and participate in exactly two triangles ($n_{v}=2$). The set $\cal V_{\rm ex}$ contains the  colored nodes on the 
                boundary for which $n_{v}<2$.}
                \label{Triangles}

\end{figure}

%{In order to formally describe the set of triangles $\cal T$ in $\cal G(V,E)$, we consider the set of ordered Eisenstein integer triplets $${ \cal P} = \{[z,z+\omega,z+\omega+1]: z\in\mathbb{Z}(\omega)\setminus \Lambda_{0} \}.$$
%Recall from Section \ref{igraph} that for all $z\in\mathbb{Z}(\omega)\setminus \Lambda_{0}$, the points $z$, $z+\omega$ and $z+\omega+1$ form a triangle $\Delta(z)\subseteq {\cal D}$ and correspond to  the interference graph vertices $\phi^{-1}(z)$, $\phi^{-1}(z+\omega)$ and $\phi^{-1}(z+\omega+1)$. We can hence define  the set of vertex triangles as 
%\begin{equation}
%{\cal T} \hspace{-0.05in}=\hspace{-0.05in} \{[u,v,w]: [\phi(u),\phi(v),\phi(w)]\in {\cal P}, u,v,w\in {\cal V}\}.
%\label{triangledef}
%\end{equation}
%\vspace{-0.1in}
%}
In order to formally describe the set of triangles $\cal T$ in $\cal G(V,E)$, we consider the set of ordered Eisenstein integer triplets $${ \cal P} = \{[z,z+\omega,z+\omega+1]: z\in\mathbb{Z}(\omega)\setminus \Lambda_{0}  \}.$$
Recall from Section \ref{igraph} that when $z\in\mathbb{Z}(\omega)\setminus \Lambda_{0}$, the points $z$, $z+\omega$ and $z+\omega+1$ form the line segments $\Delta(z)\subseteq {\cal D}$ and the corresponding graph vertices $\phi^{-1}(z)$, $\phi^{-1}(z+\omega)$ and $\phi^{-1}(z+\omega+1)$ form a connected triangle in $\cal G(V,E)$. We can hence define  the set of vertex triangles as 
\begin{equation}
{\cal T} \hspace{-0.05in}=\hspace{-0.05in} \{[u,v,w]: [\phi(u),\phi(v),\phi(w)]\in {\cal P}, u,v,w\in {\cal V}\}.
\label{triangledef}
\end{equation}

The above definition is illustrated in Fig.~\ref{Triangles} in which shaded triangles connect the corresponding vertex triplets $[u,v,w]\in \cal T$.
Notice that apart from some vertices on the external boundary of the graph, all other nodes participate in exactly two triangles in $\cal T$. This observation will be particularly useful in rewriting the sum in the objective function of ${ \rm Q}_{1}({{\cal G_{\pi}}})$ as a sum over $\cal T$ instead of $\cal V$.

Let 
\begin{equation}
n_{v}\triangleq \sum_{[i,j,k]\in {\cal T}}\mathlarger{\mathbbm{1}}\Big\{v\in\{i,j,k\}\Big\}
\label{nv}
\end{equation}
denote the number of triangles $[i,j,k]\in {\cal T}$ that include a given vertex $v\in \cal V$.
As we have seen, $n_{v}$ can only take values in $\{0, 1, 2\}$ for any $v\in \cal V$. More specifically $n_{v}= 2$ for all  internal vertices in $\cal G(V,E)$, while  $n_{v}< 2$ only for some external vertices that lie on the outside boundary of our graph.

We define the set of internal and external vertices as follows.
\begin{eqnarray}
{\cal V}_{\rm in} &=& \{v\in {\cal V} : n_{v}=2\}, \mbox{ and} \\
{\cal V}_{\rm ex} &=& \{v\in {\cal V} : n_{v}<2\}.
\end{eqnarray}

 In Fig.~\ref{Triangles} we show the above distinction by coloring all graph vertices $v$ that belong to the set ${\cal V}_{\rm ex}\subseteq\cal V$.  

%Notice that the white vertices participate in exactly two triangles in $\cal G$ while   
%not all $v\in \cal V$ participate in vertex triangles. In fact, due to boundary conditions, a vertex can participate in either $0$, $1$, or $2$ triangles in $\cal G$. 
%In Fig.~\ref{Triangles} for example all the white nodes participate in exactly two triangles, while the colored nodes .  

\begin{lemma}[Triangle sums]Consider the  interference graph $\cal G(V,E)$ and let $\{x_{v}:\;v\in \cal V\}$  be a set of values associated with $\cal V$. 
%and $\{x_{ij},\;[i,j]\in {\cal E}_{\pi}\}$   a set of values associated with $\cal E_{\pi}$. 
The sum of $x_{v}$ over all vertices $v\in \cal V$ can be written as
\begin{equation}
\sum_{v\in\cal V}x_{v} = \hspace{-0.1in} \sum_{[i,j,k]\in \cal T}\left(\frac{x_{i}+x_{j}+x_{k}}{2}\right) + \sum_{u\in {\cal V}_{\rm ex}}\left(1-\frac{n_{u}}{2}\right)x_{u}.
\end{equation}
\label{lem1}
\end{lemma}
%\begin{proof}$\displaystyle \sum_{[i,j,k]\in \cal T}\left({x_{i}+x_{j}+x_{k}}\right)=$
%\begin{eqnarray}
% &=& \sum_{v\in \cal V}n_{v}x_{v}\nonumber\\
% &=&\sum_{v\in \cal V_{\rm in}}n_{v}x_{v} + \sum_{u\in \cal V_{\rm ex}}n_{u}x_{u}\nonumber\\
%&=& 2\sum_{v\in \cal V}x_{v} + \sum_{v\in \cal V_{\rm in}}(n_{v}-2)x_{v} \nonumber\\& &+ \sum_{u\in \cal V_{\rm ex}}(n_{u}-2)x_{u}\nonumber \\
%&=& 2\sum_{v\in \cal V}x_{v} + \sum_{v\in \cal V_{\rm ex}}(n_{v}-2)x_{v}.\nonumber
%\end{eqnarray}
%Rearranging the terms and dividing by $2$ gives the desired result. 
%\end{proof}
\vspace{-0.2in}
\begin{proof}
From the definition of $n_{v}$ in (\ref{nv}), we have that
$$
\sum_{[i,j,k]\in \cal T}\left({x_{i}+x_{j}+x_{k}}\right) = \sum_{v\in \cal V}n_{v}x_{v}.
$$
Splitting the sum in terms of $\cal V_{\rm in}$ and $\cal V_{\rm ex}$ we get
\begin{eqnarray}
\sum_{v\in \cal V}n_{v}x_{v} &=&\sum_{v\in \cal V_{\rm in}}n_{v}x_{v} + \sum_{u\in \cal V_{\rm ex}}n_{u}x_{u}\nonumber\\
&=& 2\sum_{v\in \cal V}x_{v} + \sum_{v\in \cal V_{\rm in}}(n_{v}-2)x_{v} + \sum_{u\in \cal V_{\rm ex}}(n_{u}-2)x_{u}\nonumber \\
&=& 2\sum_{v\in \cal V}x_{v} + \sum_{v\in \cal V_{\rm ex}}(n_{v}-2)x_{v}\,,\nonumber
\end{eqnarray}
where the last step follows from the fact that $n_{v} =2$ for all $v \in \cal V_{\rm in}$. Rearranging the terms and dividing by $2$ gives the desired result.
\end{proof}

In view of the above lemma, we can rewrite the average DoFs in the objective  of $\rm Q_{1}(\cal G_{\pi})$ as 
\begin{equation}
\frac{1}{|{\cal V}|} \sum_{v\in \cal V}d_{v} = \frac{1}{2|{\cal V}|}\sum_{[i,j,k]\in \cal T}
\hspace{-0.09in}(d_{i}+d_{j}+d_{k}) + \frac{D_{\rm ex}}{|\cal V|}\, ,
\label{eq24}
\end{equation}  
where 
 $$D_{\rm ex} = \sum_{v\in {\cal V}_{\rm ex}}\left(1-\frac{n_{v}}{2}\right)d_{v}.$$

Notice that from (\ref{eq:cond1}) and (\ref{eq:cond2}), the maximum sum $d_{i}+d_{j}+d_{k}$ that any triangle $[i,j,k]\in\cal T$ can achieve in our setting is $3$ and hence we can bound the sum in (\ref{eq24}) as
\begin{equation*}
\sum_{[i,j,k]\in \cal T}
\hspace{-0.09in}(d_{i}+d_{j}+d_{k})\leq 3|\cal T|.
\end{equation*}

Similarly, since $d_{v}\leq 2$ and $n_{v}\geq 0$ for all $v\in \cal V$, we have that $$D_{\rm ex}\leq 2|{\cal V_{\rm ex}}|.$$ 
 
\vspace{-0.15in}
Using the above inequalities we can upper bound  the average achievable DoFs in the cellular network as
\vspace{-0.15in}
\begin{equation}
\frac{1}{|{\cal V}|} \sum_{v\in \cal V}d_{v} \leq \frac{3|\cal T|}{2|{\cal V}|}  + \frac{2|\cal V_{\rm ex}|}{|\cal V|}.
\end{equation} 
\begin{lemma}   \label{lem:2}
By construction, the interference graph $\cal G(V,E)$ satisfies
  \begin{align}&|{\cal T}|\leq \frac{2}{3}|\cal V|, \;\;\;\mbox{and}\\
  &|{\cal V_{\rm ex}}| = {\cal O}\left(\sqrt{|\cal V|}\right).
   \end{align}
\end{lemma}

\begin{proof} See Appendix \ref{proof-lemma2}. \end{proof}

Using the result of the above lemma, we can obtain 
\begin{equation}
{\rm opt}({ \rm Q}_{1}({{\cal G_{\pi}}})) \leq 1 + {\cal O}\left( \scriptstyle{1}/{{\sqrt{|{\cal V}|}}}\right),
\end{equation}
and hence conclude that when $M=2$, the average DoFs achievable in the sectored cellular system are bounded by $1 + {\cal O}\left( \scriptstyle{1}/{{\sqrt{|{\cal V}|}}}\right)$ for any decoding order $\pi$.

% 
%Intuitively, we can see that every triangle $[u,v,w]\in \cal T$ is uniquely associated with a  node $u\in \cal V$ that satisfies $f(\phi(u))\neq 0$. 

%we should expect that the number of triangles are approximately 2/3 times the number of vertices in our graph by associating every triangle $[u,v,w]\in \cal T$ with its starting node $u\in V$. Looking at the definition of $\cal T$, a node $u\in \cal V$ that is a starting node of a triangle has to satisfy $f(\phi(u))\neq 0$. Recalling that $f(z)$ partitions $\cal V$ into three sets we can argue that $|{\cal T}|\approx \frac{2}{3}|\cal V|$.

%Since $d_{v}\leq 2$ and $n_{v}\geq 0$ for all $v\in \cal V$, we have that $$D_{\rm ex}\leq 2|{\cal V_{\rm ex}}|$$.

%We can then argue that the achievable degrees of freedom can only be improved in a small fraction of the network by an amount that vanishes as 

%Intuitively, this means that the achievable degrees of freedom can only be improved by an amount that vanishes as 

\section{Networks with Intra-Cell Interference}\label{sec:intracell}

In this section we extend our cellular model to 
incorporate both out-of-cell and intra-cell interference. Namely we will assume here that a sector receiver observes interference not only from its out-of-cell neighbors but also from the other transmitters located within the same cell. 
These intra-cell interfering links are shown as black arrows in Fig.~\ref{cell}   
and correspond to the dashed edges in the interference graph shown in Fig.~\ref{fig-igraph}. 

The interference graph, denoted here as $\hat{\cal G}\big(\cal V,\hat{\cal E}\big)$, is  the same as the graph defined in Section \ref{igraph} with the only difference that the set  $\hat{\cal E}$ now includes both out-of-cell and intra-cell interference edges.  Similarly we can define the directed interference graph $\hat{\cal G}_{\pi}\big(\cal V,\hat{\cal E}_{\pi}\big)$ for any network interference cancellation decoding order $\pi$.

We will see next that these additional interfering links in $\hat{\cal E}$ do not affect the achievable degrees of freedom in our  cellular system as long as we allow  the sectors of each cell to jointly process their received signals.\footnote{It is interesting to notice that joint sector processing 
at the same cell base station site is implemented in current technology.} 
%{\RED[REF: some paper ... by Qualcomm/by Lucent??]}.}
Again, we state here our main achievability result and focus on the case where $M = 2$ in Section \ref{sec:intracellA}, while the full proof
is postponed to Appendix \ref{proof:thm3}.
%
%Let $\hat{\cal G}\big(\cal V,\hat{\cal E}\big)$ denote the interference graph defined as in Section \ref{igraph} with the only difference that the set $\hat{\cal E}$ here includes both out-of-cell and intra-cell interference edges. 

%Let $\hat{\cal G}\big(\cal V,\hat{\cal E}\big)$ denote the interference graph except that $\hat{\cal E}$ is now defined to include both out-of-cell and intra-cell interference edges. 

\begin{thm} \label{thm:sectors}
For a sectored cellular system $\hat{\cal G}\big(\cal V,\hat{\cal E}\big)$ in which transmitters and receivers are equipped with $M$ antennas each,  
there exists a one-shot linear beamforming scheme that achieves the average (per sector) DoFs
%\vspace{-0.15in}
\begin{equation}d_{\hat{\cal G},{\pi^{*}}} = \begin{cases}\frac{M}{2}, \;\;\;\;\;\;\;\;\;\;\;\mbox{$M$ is even}\\ \frac{M}{2} -\frac{1}{6}, \;\;\;\;\;\mbox{$M$ is odd}\end{cases}
%\frac{M-1}{2} +\frac{1}{3}
\end{equation}
under the network interference cancellation framework with decoding order $\pi^{*}$, 
and with joint processing within the sectors of each cell.
\hfill \QED
\end{thm}

%{\RED [QUESTION: why there is no converse theorem for this setting? Is the converse a straightforward corollary of Theorem 2?
%In order to show this, we have to show that the network without intra-cell interference is always better or equal, in terms of DoFs, than the network
%with intra-cell interference **even though we allow joint processing of intra-cell sector receivers**. Is this straightforward to see?
%In any case, IT IS BETTER TO ADD A COMMENT]}

\subsection{Achievability}\label{sec:intracellA}

Consider the beamforming scheme described in Section \ref{bfscheme} for $M=2$ and focus on the cell $\{a,b,c\}$ shown in Fig.~\ref{intracell}. 
Without loss of generality, we will describe here how to jointly process the received observations in 
$\yv_{a}$, $\yv_{b}$ and $\yv_{c}$ such that all intra-cell interference can be eliminated and show that under the network interference cancellation framework and the beamforming choices of Section \ref{bfscheme},  every sector receiver in the network is able to decode its own desired message.

\begin{figure}[ht]
                \centering
                \includegraphics[width=0.6\columnwidth]{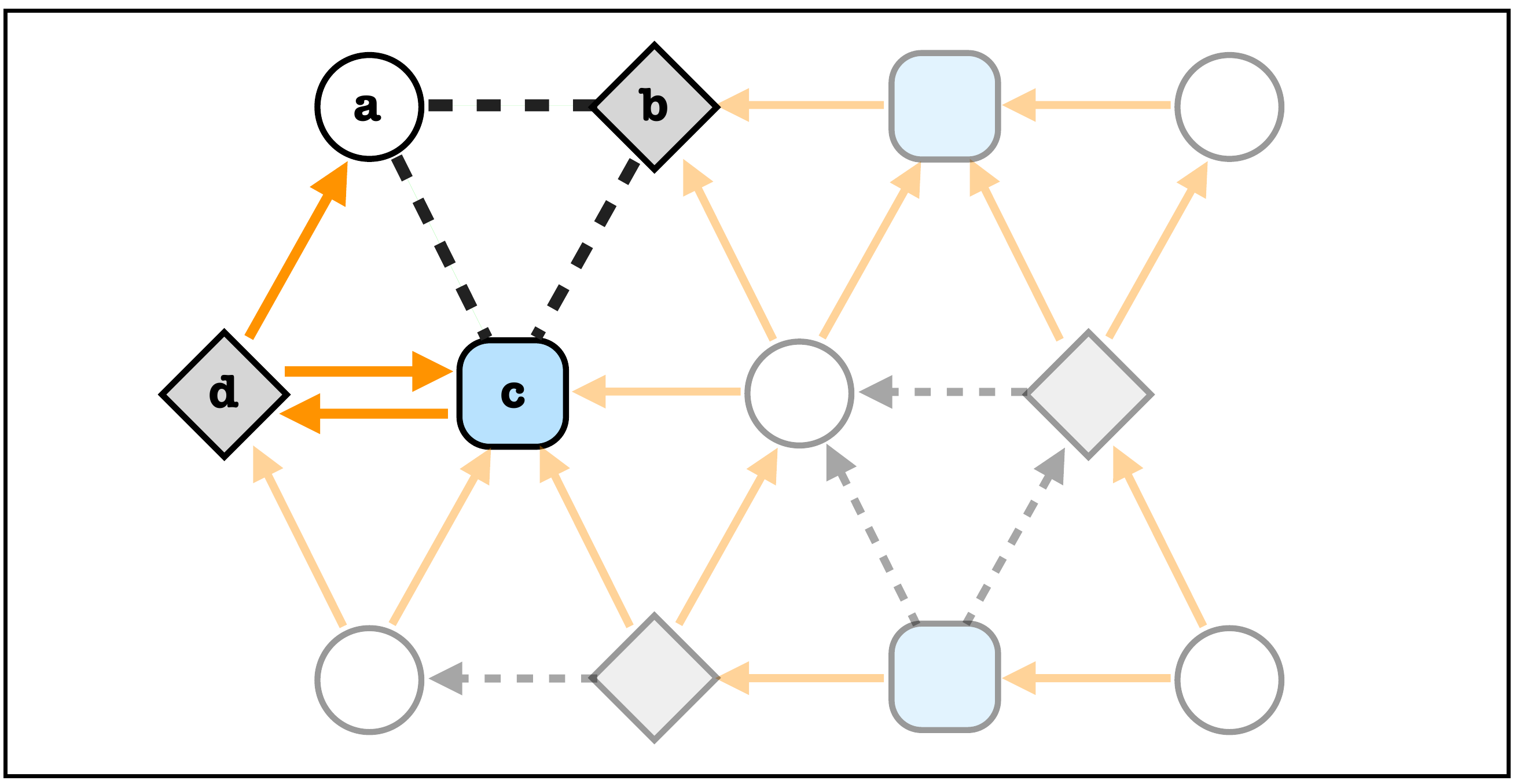}
                \caption{ Intra-Cell Interference Elimination. The sectors of each cell can jointly process their received signals and successively decode their  desired  messages. 
                %The receiver in sector $a$ can use the projected observations from sectors {$b$} and {$c$} to eliminate interference and decode ${ s}_{ a}$. Then the receivers in sectors {$b$} and {$c$} can use ${s}_{ a}$ and successively decode their own messages.} 
                }
\label{intracell}
\end{figure}

According to  the ``left-to-right, top-down'' decoding order $\pi^{*}$, at the time when sector $a$ attempts to decode,  all the interfering links from transmitters located ``above and to the left'' of $a$ have already been eliminated. As we can also see in Fig.~\ref{intracell}, sector $a$ will observe  intra-cell interference from sectors $b$ and $c$, and out-of-cell interference from sector $d$.  Hence, the received signal available to sector $a$ is  given by \begin{equation}
\yv_{a} = \Hm_{aa}\vv_{a}s_{a} + \sum_{u\in \{b,c,d\}}\Hm_{au}\vv_{u}s_{u} + \zv_{a}.
\label{obsA}
\end{equation}

At the same time, the receivers $b$ and $c$ will be observing interference from all their neighboring sectors that have not decoded their messages yet. Notice however that with the specific beamforming choices described in Section \ref{bfscheme}, all interference that comes from sectors whose messages will be decoded after sectors $b$ and $c$ according to $\pi^{*}$, occupy a single dimension in each receiver and can hence be zero-forced.  It is only the transmitter associated with sector $d$ that is going to cause interference after the projection. 
Therefore,
the corresponding  observations from sectors $b$ and $c$ that are available when receiver $a$ attempts to decode  are given by 
\begin{align}
&\uv_{b}^{\rm H}\yv_{b}= \uv_{b}^{\rm H}\Hm_{bb}\vv_{b}s_{b} + \sum_{u\in \{a,c\}}\uv_{b}^{\rm H}\Hm_{bu}\vv_{u}s_{u} + \uv_{b}^{\rm H}\zv_{b},\label{obsB}\\
&\uv_{c}^{\rm H}\yv_{c} = \uv_{c}^{\rm H}\Hm_{cc}\vv_{c}s_{c} + \sum_{u\in \{a,b,d\}}\uv_{c}^{\rm H}\Hm_{cu}\vv_{u}s_{u} + \uv_{c}^{\rm H}\zv_{c}.\label{obsC}
\end{align}

We will see next that the cell with sectors $\{a,b,c\}$ can jointly process the above observations such that all the corresponding sector receivers will be able to  decode their desired messages in the order $s_{a}, s_{b}, s_{c}$  specified by $\pi^{*}$. 
Indeed, if we let $\sv = [s_{a}, s_{b}, s_{c}, s_{d}]^{\rm T}$, the observations (\ref{obsA}), (\ref{obsB}) and (\ref{obsC}) can be written in vector form as
\begin{equation}
{ \bf \tilde y} = {\bf \tilde H} \sv + { \bf \tilde z}
\label{vecobs}
\end{equation}
where 
\begin{align}
&{ \bf \tilde y} = [\yv_{a}, \uv_{b}^{\rm H}\yv_{b}, \uv_{c}^{\rm H}\yv_{c}]^{\rm T},\nonumber\\
&{ \bf \tilde z} = [\zv_{a}, \uv_{b}^{\rm H}\zv_{b}, \uv_{c}^{\rm H}\zv_{c}]^{\rm T} \nonumber
\end{align}
and
\begin{align}
{\bf \tilde H} = \begin{bmatrix} \Hm_{aa}\vv_{a} & \Hm_{ab}\vv_{b} & \Hm_{ac}\vv_{c} & \Hm_{ad}\vv_{d} \\ 
 \uv_{b}^{\rm H}\Hm_{ba}\vv_{a} & \uv_{b}^{\rm H}\Hm_{bb}\vv_{b} & \uv_{b}^{\rm H}\Hm_{bc}\vv_{c} & 0 \\ 
 \uv_{c}^{\rm H}\Hm_{ca}\vv_{a} & \uv_{c}^{\rm H}\Hm_{cb}\vv_{b} & \uv_{c}^{\rm H}\Hm_{cc}\vv_{c} & \uv_{c}^{\rm H}\Hm_{cd}\vv_{d}\end{bmatrix}.\nonumber
\end{align}

Now, assuming that  the channel matrices in our cellular network are chosen independently at random from a 
non-degenerate continuous complex distribution (e.g., they have independent elements drawn from 
a complex normal distribution), we can show that ${\bf \tilde H}\in \CC^{4\times 4}$ is full rank with  probability one. 
One can check that the beamforming vectors $\uv_{i}$ and  $\vv_{j}$ do not depend on the above channel realizations, and therefore  the elements of  ${\bf \tilde H}$ can be seen as   independent  random variables   for which   $\PP\big[{\rm det}({\bf \tilde H})\neq 0\big]=1$. 
Hence, the given cell can always decode  the corresponding messages  from %(\ref{vecobs}) . 
${ \bf \tilde y}$ in the required order, as soon as the observations (\ref{obsA}), (\ref{obsB}) and (\ref{obsC}) become available to sectors $a$, $b$ and $c$.

%Recall that the beamforming vectors $\vv_{a}, \vv_{b}, \vv_{c}$ and $\vv_{d}$ as well as the receive projections $\uv_{b}$ and $\uv_{c}$ have been chosen as functions of out-of-cell channel matrices and hence are independent of all $\Hm_{uv}$   
\clearpage
In order to state the decoding process more explicitly, consider $\Qm= [\qv_{1},\qv_{2},\qv_{3},\qv_{4} ]$ to be the unitary matrix obtained by the $QL$-decomposition of ${\bf \tilde H}$ such that  %$\Qm^{\rm H}{\bf \tilde H}$ is lower triangular.   
\begin{equation*}
\Qm^{\rm H}{\bf \tilde H}=
\begin{bmatrix}
\ell_{11} & 0 &0 &0\\
\ell_{21} & \ell_{22} &0 &0\\
\ell_{31} & \ell_{32} &\ell_{33} &0\\
\ell_{41} & \ell_{42} &\ell_{43} &\ell_{44}\\
\end{bmatrix}.
\end{equation*}
The sector receivers $a$, $b$ and $c$ can first project ${ \bf \tilde y}$  along $\qv_{1}$, $\qv_{2}$ and $\qv_{3}$ in order to obtain their corresponding observations in the form
\begin{align}
&y_{a}'  = \ell_{11}s_{a} +z_{1}'\\
&y_{b}'  = \ell_{21}s_{a} +\ell_{22}s_{b} +z_{2}'\\
&y_{c}'  = \ell_{31}s_{a} +\ell_{32}s_{b} +\ell_{33}s_{c}+z_{3}',
\end{align}
and then successively decode their desired messages $s_{a}$, $s_{b}$ and $s_{c}$ according to the specified order.% $\pi^{*}$.

In general, the above observations can be generated for every cell in the network just before their first sector receiver  attempts to decode. Therefore, following the ``left-to-right, top-down'' decoding order $\pi^{*}$, all the  sectors in 
$\hat{\cal G}_{\pi^{*}}\big(\cal V,\hat{\cal E}_{\pi^{*}}\big)$ can 
%successively 
decode their desired messages using the above procedure 
and hence the average (per sector) degrees of freedom $d_{\hat{\cal G},{\pi^{*}}} = 1$ are achievable.

%
%The key observation here is that with the specific beamforming choices described in Section \ref{bfscheme}, all out-of cell interference that comes from sectors $v$ whose messages will be decoded after $b$ and $c$ according to $\pi^{*}$ can be zero-forced. Notice that the only 

%%%%%%%%%%%%%%%%%%%%%%%%%%%%%%%%%%%%%%%%%%%%%%%%%%%%%%%%%
\section{ Topological Robustness}  \label{sec:topo}

In this section we introduce the concept of topological robustness for interference networks. Broadly speaking,  an achievable scheme is said to be {\it robust} with respect to a network topology if its performance does not depend on the existence (or strength) of interference. 
This is a very important property to take into account if we want to apply a communication scheme in practice. Cellular systems  are in principle designed such that  
most interfering links are  weak and hence any scheme that solely depends on the existence (or strength) of interference  will fail whenever the corresponding links are missing (or weak). 

Consider for example the achievable scheme described in Section \ref{sec:intracellA} and assume that for a given channel realization all interference observed at receiver~$a$ is zero. In this network instance, depicted in Fig.~\ref{fig:notrobust}, the equivalent 
%in which all  interfering links adjacent to node $a$ have been removed. 
\begin{figure}[ht]
                \centering
                \includegraphics[width=.6\columnwidth]{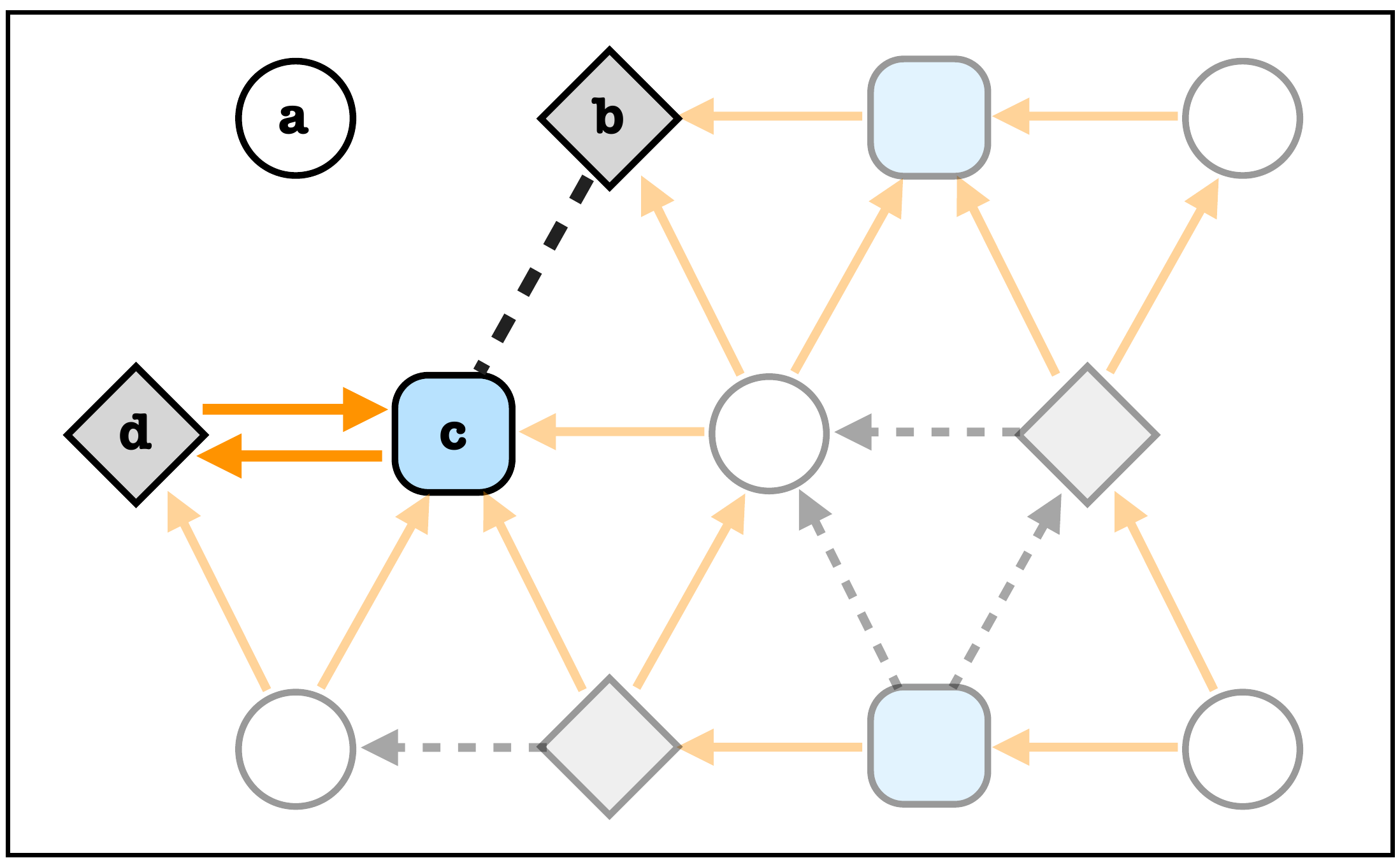}
                \caption{ The achievable scheme of Section \ref{sec:intracellA} is not topologically robust. Sector $b$ will not able to decode its own desired message if Sector $a$ does not observe  interference from its neighbors.
                }
\label{fig:notrobust}
\end{figure}
channel matrix in the joint receiver observation for the cell $\{a,b,c\}$, given by 
\begin{align}
{\bf \tilde H}' = \begin{bmatrix} \Hm_{aa}\vv_{a} & {\bf 0}_{2\times 1} & {\bf 0}_{2\times 1} & {\bf 0}_{2\times 1} \\ 
 0 & \uv_{b}^{\rm H}\Hm_{bb}\vv_{b} & \uv_{b}^{\rm H}\Hm_{bc}\vv_{c} & 0 \\ 
 0 & \uv_{c}^{\rm H}\Hm_{cb}\vv_{b} & \uv_{c}^{\rm H}\Hm_{cc}\vv_{c} & \uv_{c}^{\rm H}\Hm_{cd}\vv_{d}\end{bmatrix},\nonumber
\end{align}
is rank-deficient and hence the desired transmitted messages cannot be resolved from (\ref{vecobs}). Even though sector $a$ can always decode its own message, its observation cannot help sectors $b$ and $c$  eliminate  the remaining interference, and therefore sector $b$ cannot decode its message.

It is not surprising that the above scheme fails in this case; the receiver has been designed to rely on a {\it specific interference topology} (cf. Fig.~\ref{intracell}) in order 
obtain the required linearly-independent observations. Whenever the corresponding links are missing, the decoding process fails and therefore the scheme proposed in Section \ref{sec:intracellA} cannot be considered 
 topologically robust.% according to our definition.
%Notice that even though the receiver $a$ can decode its own message, the corresponding observation is not able to provide a linearly independent view 

%nnn n cc nn

In practice, interference will never be exactly zero as in the previous example. However, 
any communication scheme that critically depends on  sufficiently strong interfering links (e.g. , such that 
the corresponding messages can be decoded and interference can be canceled) 
will suffer from significant noise enhancement in the decoding process whenever the corresponding 
channel gains are below a certain threshold. 
In this case the corresponding receiver will not be able to decode within the operating SNR range 
of the network and the weaker interference links will become the bottleneck in its performance. 
%{\color{red}[ different schemes for different topologies: exponentially many configurations ]}

\clearpage

Under this framework, one could consider all possible channel realizations and design a family of transmission schemes, each one specifically optimized for the corresponding interference topology. Even though this is a tractable approach for small networks, it becomes more challenging as the size of the network increases. Here, we take a unified approach and propose  a topologically robust transmission scheme for large cellular systems that is able to maintain  the same performance for all network configurations, no matter if the interference links are strong or weak.

%%%%%%%%%%%%%%%%%%%%%%%%%%%%%%%%%%%%%%%%%
\subsection{The Compound Cellular Network}

In order to formally capture the concept of topological robustness in our cellular model, we will consider here a compound scenario in which 
any subset of the interfering links could be potentially missing from the network. 
More precisely, we focus on the sectored cellular system  $\hat{\cal G}({\cal V},\hat{\cal E})$ defined in Section~\ref{sec:intracell} and we assume that
every directed edge $[v,u]\in \hat{\cal E}$ is associated with a binary channel-state parameter  $\alpha_{uv}\in\{0,1\}$ that determines whether the corresponding  link will exist in the  network or not. 

The compound channel matrices are generated in the form of ${\alpha_{uv}\cdot\Hm_{uv}},\;\forall [v,u]\in \hat{\cal E}$,  as a function of the  channel-state configuration 
\begin{equation}
{\cal A} \triangleq \big\{\alpha_{uv}\in \{0,1\}: [v,u]\in \hat{\cal E}\big\},
\end{equation}
%\subseteq \{0,1\}^{|\hat{\cal E}|},$$
and the compound cellular network is defined over all possible choices of ${\cal A}\in \{0,1\}^{2|\hat{\cal E}|}$. 
We assume that the  channel-state configuration ${\cal A}$ is known to all  receivers   but is {\em a priori} unavailable\footnote{It is important to note that in the original set-up, the channel state information is available at the transmitters, and therefore obviously the channel is not compound. However, this rather artificial compound model allows us to design a \emph{unified} achievable scheme, which works, independent of the strength of interference links.} to the transmitters in the above compound network, in the sense that %we require that
the interference alignment 
precoding scheme  (although a function of the channel matrices   $\Hm_{uv}$ and of the interference graph $\hat{\cal G}({\cal V},\hat{\cal E})$)
must be designed irrespectively of ${\cal A}$. 
% and the corresponding channel matrices $\Hm_{uv}$,  are assumed to be known, as in our previous setting, to both transmitters and receivers.
%only the generated channel matrices $\Hm_{uv}$ and not the corresponding channel-state parameters $\alpha_{uv}$ are known to the transmitters. 

A topologically robust transmission scheme is required to maintain the same performance for all channel-state parameters ${\cal A}
\in \{0,1\}^{2|\hat{\cal E}|}$. Let $\hat{\cal G}({\cal V},\hat{\cal E}\big|{\cal A})$ be the interference graph  generated in the above  compound  network when the channel-state  is ${\cal A}$ and let $d_{\hat{\cal G}}({\cal A})$ denote the average (per sector) degrees of freedom   achievable in $\hat{\cal G}({\cal V},\hat{\cal E}\big|{\cal A})$.

%\begin{defn}
%A communication scheme designed for a sectored cellular system $\hat{\cal G}({\cal V},\hat{\cal E})$ is  topologically robust   if it can achieve the same average (per sector)  degrees of freedom,  $d_{\hat{\cal G}}({\cal A})$, for all channel-state configurations ${\cal A}
%\in \{0,1\}^{|\hat{\cal E}|}$.
%\end{defn}

\clearpage

\begin{defn}
A communication scheme designed for a sectored cellular system $\hat{\cal G}({\cal V},\hat{\cal E})$ is said to be topologically robust with robustness level $d>0$,  if it can achieve    $d_{\hat{\cal G}}({\cal A})\geq d$, for all channel-state configurations ${\cal A}
\in \{0,1\}^{2|\hat{\cal E}|}$.  \hfill $\lozenge$
\end{defn}

As we have seen before, the achievable scheme described in Section \ref{sec:intracellA} is not topologically robust according to the above definition: 
even though it can achieve $d_{\hat{\cal G}}({\cal A}) = 1$, when $\alpha_{uv}=1,\forall (u,v)\in \hat{\cal E}$, there exists a configuration ${\cal A}'$, shown in Fig.~\ref{fig:notrobust}, in which the decoding process fails.% and $d_{\hat{\cal G}}({\cal A}') = 0$.

Under this framework, we are interested in the design of communication schemes that  maximize the {\it compound} degrees of freedom,
\begin{equation} 
d_{\rm C} \triangleq \min_{ {\cal A} \in \{0,1\}^{|\hat{\cal E}|}} d_{\hat{\cal G}}({\cal A}).
\end{equation}
Notice that a topologically robust scheme with robustness level $d$ achieves (by definition) the compound DoFs $d_{\rm C} = d$.  
%
%A~question of interest is  whether topological robustness comes along with a significant performance loss compared to the degrees of freedom $d_{\hat{\cal G}}$ that can be achieved in the original network. Of course, $d_{\rm C}\leq d_{\hat{\cal G}}$. 
The following theorems show the existence of topologically robust schemes that achieve 
%essentially 
the optimum compound DoFs performance, which coincides with the optimum DoFs performance in the non-compound setting 
with intra-cell interference given in Section \ref{sec:intracellA}.

\begin{thm} \label{thm-comp}
For a compound sectored cellular system  $\left\{\hat{\cal G}\big({\cal V},\hat{\cal E}\big|{\cal A}\big):{\cal A} \in \{0,1\}^{2|\hat{\cal E}|}\right\}$, in which transmitters and receivers are equipped with $M$ antennas each,  there exists a one-shot linear beamforming scheme that achieves  the average
(per sector) compound DoFs
%\vspace{-0.15in}
\begin{equation}d^*_{\rm C} = \begin{cases}\frac{M}{2}, \;\;\;\;\;\;\;\;\;\;\;\mbox{$M$ is even}\\ \frac{M}{2} -\frac{1}{6}, \;\;\;\;\;\mbox{$M$ is odd}\end{cases}
%\frac{M-1}{2} +\frac{1}{3}
\end{equation}
under the network interference cancellation framework, assuming local receiver cooperation within each cell.
%are achievable in the corresponding directed interference graph $\cal G_{\pi^{*}}$ under the network interference cancellation framework with decoding order $\pi^{*}$ . 
\hfill \QED
\end{thm}

\begin{thm}\label{thm:converse2} 
The compound DoFs $d_{\rm C}$ of a sectored cellular system 
$\left\{\hat{\cal G}\big({\cal V},\hat{\cal E}\big|{\cal A}\big):{\cal A} \in \{0,1\}^{2|\hat{\cal E}|}\right\}$ 
achievable by any one-shot linear beamforming scheme under the network 
interference cancellation framework  are bounded by 
$\textstyle
d^{*}_{\rm C} + {\cal O}\left( \scriptstyle{1}/{{\sqrt{|{\cal V}|}}}\right)
$, where $d^{*}_{\rm C}$ is given by Theorem \ref{thm-comp}.
\hfill \QED
\end{thm}

As before, we discuss in detail the case $M = 2$ and sketch the proof of Theorem \ref{thm-comp} in the 
general case in Appendix \ref{comp-ach-proof}. Theorem \ref{thm:converse2} is proved in Appendix \ref{thm:converse2-proof}.

%%%%%%%%%%%%%%%%%%%%%%%%%%%%%%%%%%%%%%%%%%%%%%%%%%%%%%%%%
\subsection{Topologically Robust Achievability}\label{robust:achievability}

In this section, we  focus on the case where $M=2$ and describe a topologically robust transmission scheme for $\hat{\cal G}({\cal V},\hat{\cal E})$ that is able to achieve $d_{\rm C}=1$.
We will consider a scheme very similar to the one described in Section \ref{sec:intracellA}. 
We will use the same beamforming strategy, but consider a new decoding order that is able to guarantee topological robustness.

%We start by distinguishing 
In the terminology to follow, we distinguish between {\it primary} and {\it secondary} sectors in our network according to their 
relative position within each cell. We say that a sector $v \in {\cal V}$ is primary 
if $v \in {\cal V}_{\rm circle}$ (i.e., it is located in the upper-left corner of a cell) and  {secondary} otherwise. 

\begin{figure}[ht]
                \centering
                \includegraphics[width=.7\columnwidth]{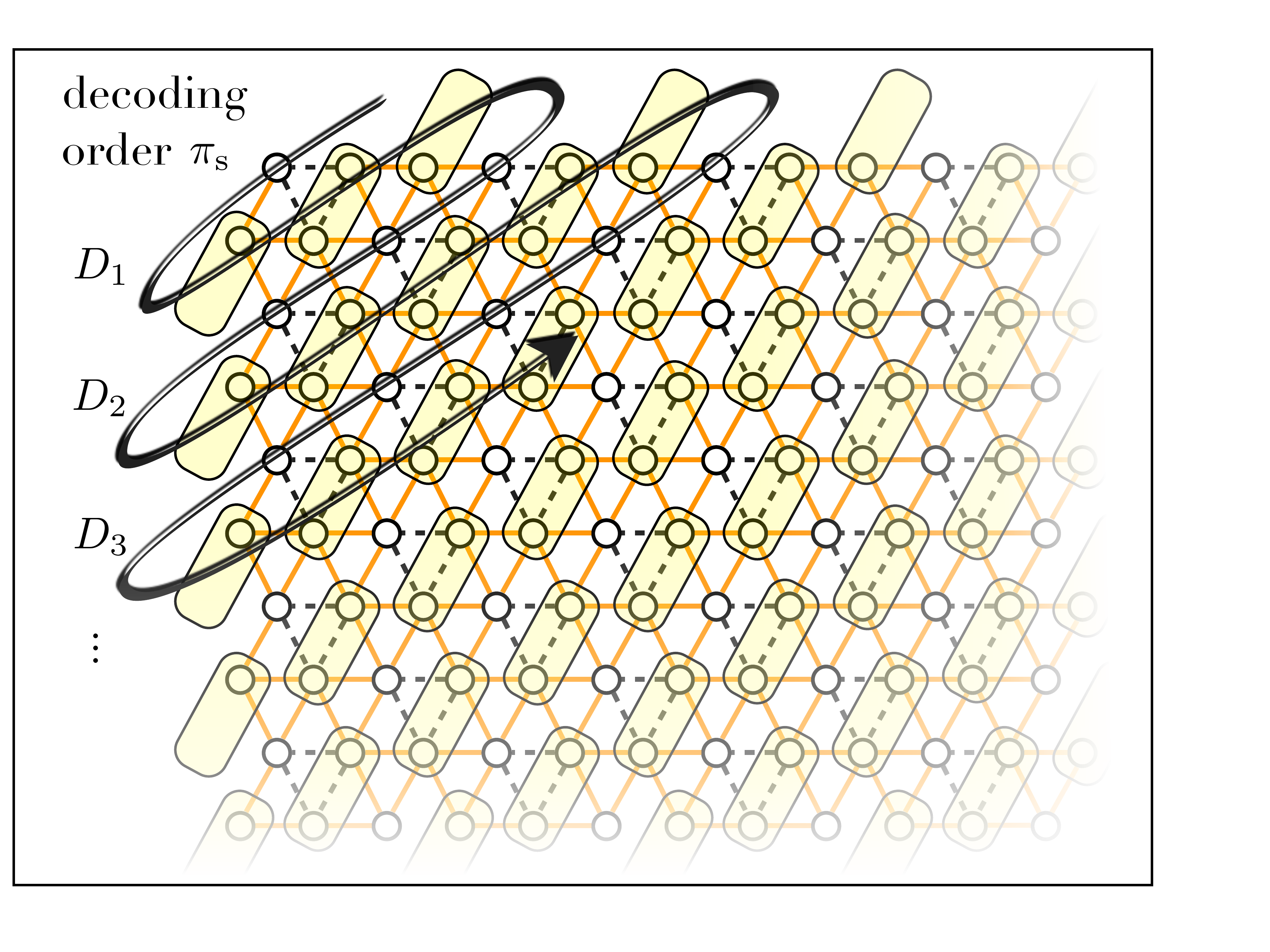}
                \caption{ The  decoding order $\pi_{\rm s}$. The highlighted sector pairs decode their messages simultaneously.
                }
\label{robustorder}

\end{figure} 
We  consider here a new decoding order under the network interference cancellation framework in which  cells  decode their messages in  diagonal groups, starting from the upper-left corner of the network. Within each group,  the cells  first decode their primary  messages (i.e, the ones associated with  primary sectors) following a top-down decoding order and then proceed to  their secondary messages which are decoded in the opposite direction. This  process leads to the ``curly-S'' decoding order shown in Fig.~\ref{robustorder} and will be denoted here as $\pi_{\rm s}$.

An important property of the above decoding order  is that it maintains, under  network interference cancellation, the same out-of-cell interfering link directions as the ``left-to-right, top-down'' decoding order $\pi^{*}$. 
We have that 
$$
{\cal E}_{\pi_{\rm s}} \triangleq \left\{[u,v] : (u,v) \in {\cal E} \mbox{ and }  v\prec_{\pi_{\rm s}}u   \right\} = {\cal E}_{\pi^*}
$$ 
and hence the  beamforming scheme designed for ${\cal E}_{\pi^{*}}$ (Section \ref{bfscheme}) can be directly  applied in this case  and satisfy the  out-of-cell  alignment conditions
\begin{equation}
\uv_{u}^{\rm H}\Hm_{uv}\vv_{v}  = 0,\;\forall [v,u]\in {\cal E}_{\pi_{\rm s}}.
\label{eq:ali}
\end{equation}
Recall the example shown in Fig.~\ref{fig:notrobust} and assume that receiver $a$ has already decoded its own message.  With the previous decoding order, $\pi^{*}$, the receivers $b$ and $c$ were unable to jointly decode their messages due to the existing interference from sector $d$. 
With the new decoding order however, this is no longer  an issue. According to $\pi_{\rm s}$, the receiver in sector $d$  will be decoded before sectors $b$ and $c$, and hence its message will be available to the corresponding receivers for interference cancellation.

\begin{figure}[ht]
                \centering
                \includegraphics[width=.5\columnwidth]{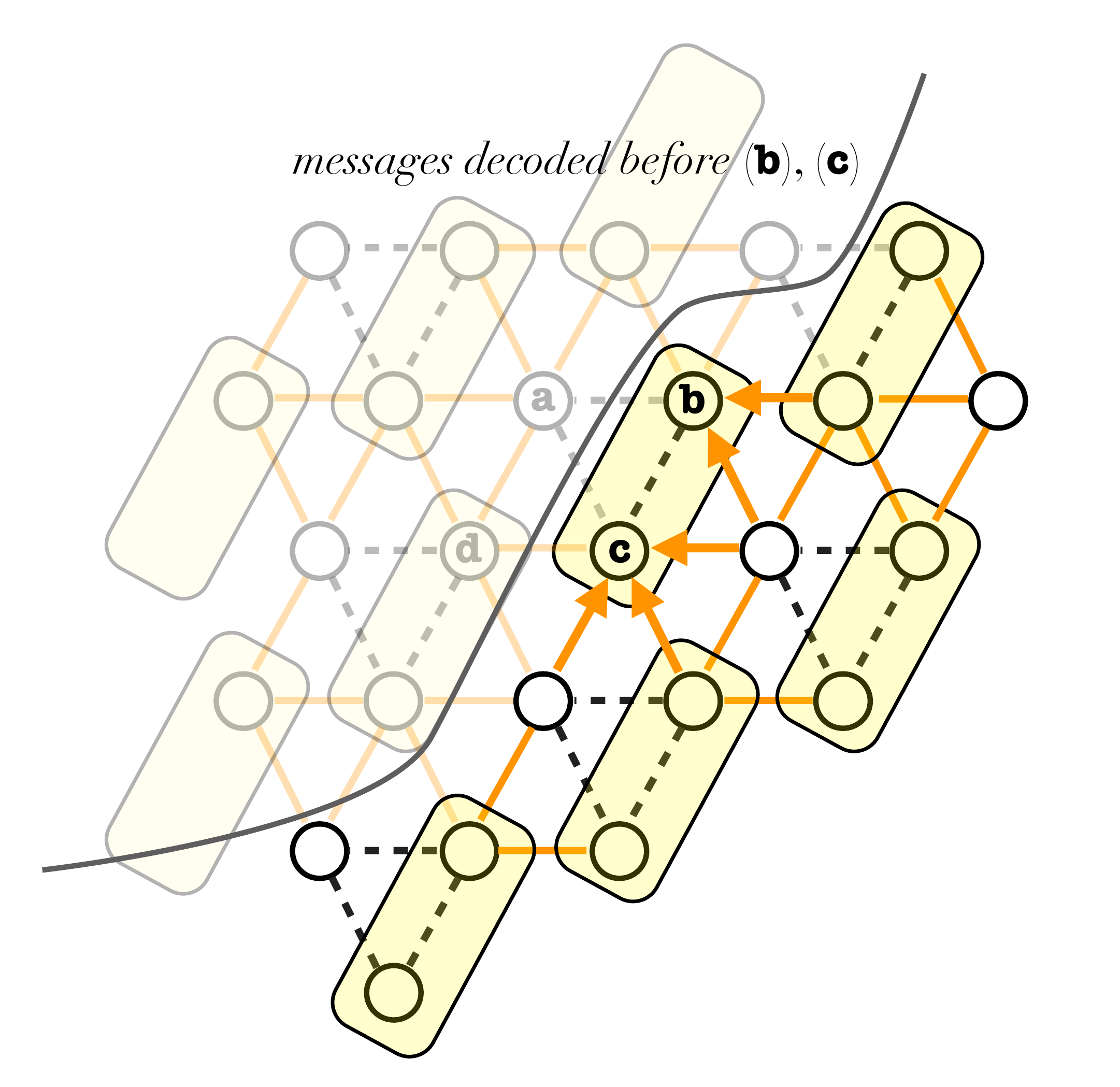}
                \caption{ Robust decoding for sectors $b$ and $c$.
                }
\label{robustdecodbc} 
\end{figure}

The network instance  described above is depicted in Fig.~\ref{robustdecodbc}. Under the network interference cancellation framework with decoding order $\pi_{\rm s}$, the receiver observations at the time when sectors $b$ and  $c$ attempt to decode  are given by
$$\yv_{b} = \Hm_{bb}\vv_{b}s_{b} +\overline\Hm_{bc}\vv_{c}s_{c} + \sum_{v: [v, b]\in {\cal E}_{\pi_{\rm s}} } \overline\Hm_{bv}\vv_{v}s_{v} +\zv_{b},$$
and
$$\yv_{c} = \Hm_{cc}\vv_{c}s_{c} +\overline\Hm_{cb}\vv_{b}s_{b} + \sum_{v: [v, c]\in {\cal E}_{\pi_{\rm s}} } \overline\Hm_{cv}\vv_{v}s_{v} +\zv_{c},$$
where $\overline\Hm_{uv} = \alpha_{uv}\cdot\Hm_{uv}$ are the compound channel matrices  with state parameters $\alpha_{uv}\in\{0,1\}.$
Notice that the secondary sectors $b$ and $c$, no longer need the primary observation from sector $a$ in order to decode their messages. 
From (\ref{eq:ali}) we have that 
$$\uv_{b}^{\rm H}\hspace{-0.1in}\sum_{v: [v, b]\in {\cal E}_{\pi_{\rm s}} } \hspace{-0.1in}\overline\Hm_{bv}\vv_{v}s_{v}=0\;\,\mbox{and}\;\,\uv_{c}^{\rm H}\hspace{-0.1in}\sum_{v: [v, c]\in {\cal E}_{\pi_{\rm s}} }\hspace{-0.1in} \overline\Hm_{cv}\vv_{v}s_{v}=0,$$
for all compound channel states and hence the corresponding observations can be written in vector form as
\begin{equation}
\begin{bmatrix}
\uv_{b}^{\rm H}\yv_{b}\\
\uv_{c}^{\rm H}\yv_{c}
\end{bmatrix}=\underbrace{
\begin{bmatrix}  
\uv_{b}^{\rm H}\Hm_{bb}\vv_{b} & \uv_{b}^{\rm H}\overline\Hm_{bc}\vv_{c}  \\ 
 \uv_{c}^{\rm H}\overline\Hm_{cb}\vv_{b} & \uv_{c}^{\rm H}\Hm_{cc}\vv_{c}\end{bmatrix}}_{\triangleq \tilde\Hm(\alpha_{bc},\alpha_{cb})}
 \begin{bmatrix}
s_{b}\\
s_{c}
\end{bmatrix} 
+\tilde\zv.
\end{equation}

The equivalent channel matrix $\tilde\Hm(\alpha_{bc},\alpha_{cb})\in \CC^{2\times 2}$ given in the above observation depends on the  compound channel-state parameters $\alpha_{bc},\alpha_{cb}\in \{0,1\}$, which determine whether sectors $b$ and $c$ interfere with each other or not. We can see that the resulting channel matrices,
\begin{equation*}
%\resizebox{\hsize}{!}{$
\tilde\Hm(0,0) =  \begin{bmatrix}  
\uv_{b}^{\rm H}\Hm_{bb}\vv_{b} & 0  \\ 
 0& \uv_{c}^{\rm H}\Hm_{cc}\vv_{c}\end{bmatrix}, \;\;
\tilde\Hm(1,0) =  \begin{bmatrix}  
\uv_{b}^{\rm H}\Hm_{bb}\vv_{b} & \uv_{b}^{\rm H}\Hm_{bc}\vv_{c}  \\ 
 0 & \uv_{c}^{\rm H}\Hm_{cc}\vv_{c}\end{bmatrix},
%$}
\end{equation*}
\begin{equation*}
%\resizebox{\hsize}{!}{$
\tilde\Hm(0,1) =  \begin{bmatrix}  
\uv_{b}^{\rm H}\Hm_{bb}\vv_{b} &  0   \\ 
 \uv_{c}^{\rm H}\Hm_{cb}\vv_{b}& \uv_{c}^{\rm H}\Hm_{cc}\vv_{c}\end{bmatrix}, \;\;
\tilde\Hm(1,1) =  \begin{bmatrix}  
\uv_{b}^{\rm H}\Hm_{bb}\vv_{b} & \uv_{b}^{\rm H}\Hm_{bc}\vv_{c}  \\ 
 \uv_{c}^{\rm H}\Hm_{cb}\vv_{b} & \uv_{c}^{\rm H}\Hm_{cc}\vv_{c}\end{bmatrix},
%$}
\end{equation*}
are all full-rank, 
%$\mbox{rank}(\tilde\Hm(0))=\mbox{rank}(\tilde\Hm(1))=2$
and hence the receivers in sectors $b$ and $c$ are always able to decode their messages, irrespective of the compound channel-state parameters.

Similarly, we can show that all transmitted messages in $\hat{\cal G}\big({\cal V},\hat{\cal E}\big|{\cal A}\big)$ associated with  secondary sectors, can be successfully decoded according to $\pi_{\rm s}$, for all compound channel-state configurations ${\cal A} \in \{0,1\}^{2|\hat{\cal E}|}$.
It remains to argue that primary sectors are also able to decode their messages in the above compound network and hence show that the average (per sector) degrees of freedom $d^*_{\rm C} =1$ are achievable.

Consider the cell $\{a,b,c\}$ shown in Fig.~\ref{robustdecoda} just before its primary sector receiver $a$ attempts to decode. According to $\pi_{\rm s}$, the available receiver observations in this cell are given by 
\begin{figure}[ht]
                \centering
                \includegraphics[width=.5\columnwidth]{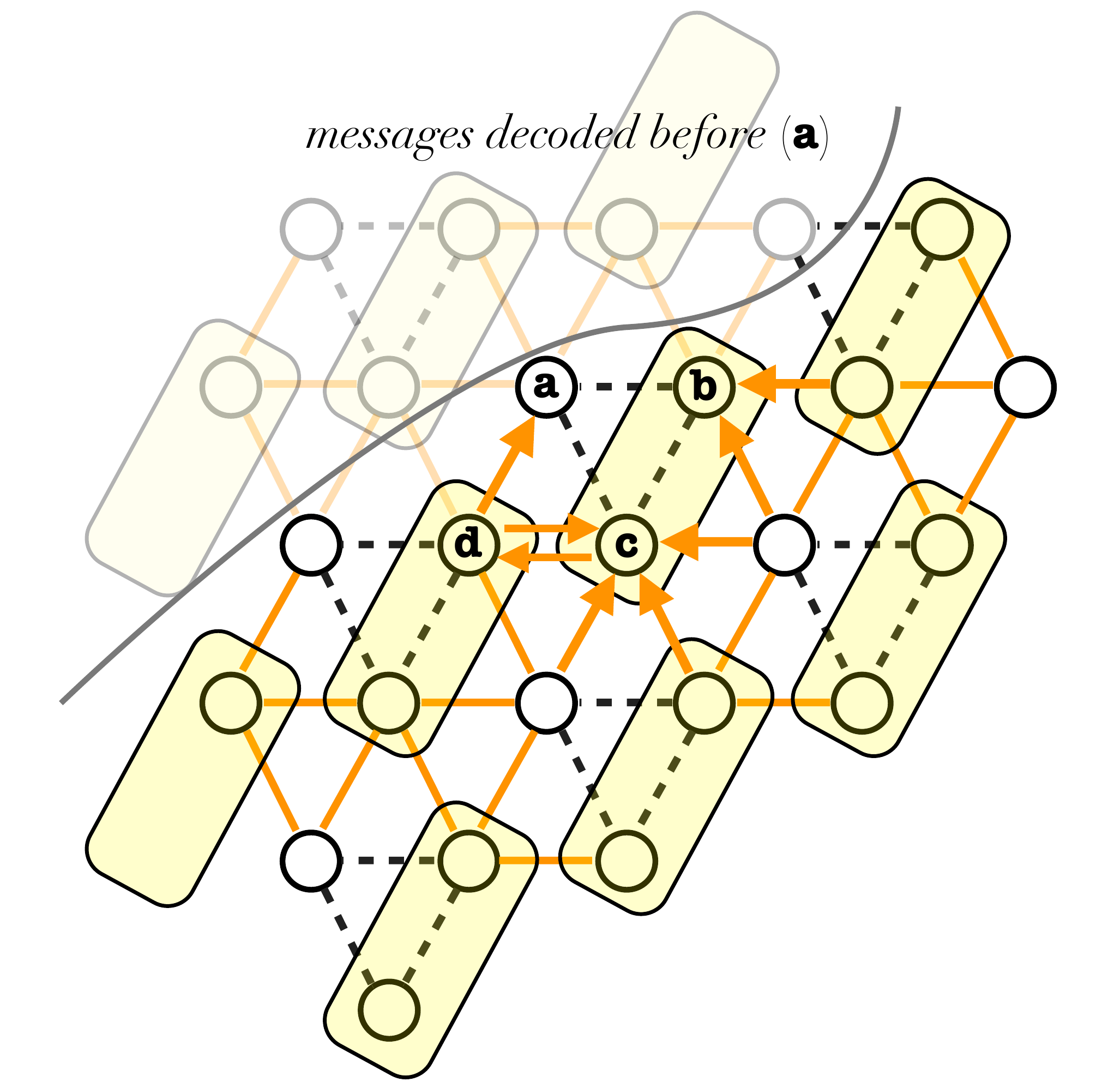}
                \caption{ Robust decoding for sector $a$.
                }
\label{robustdecoda}
\end{figure}
\begin{align*}
&\yv_{a} = \Hm_{aa}\vv_{a}s_{a} + \sum_{v\in \{b,c,d\}}\overline\Hm_{av}\vv_{v}s_{v} + \zv_{a},\\
&\uv_{b}^{\rm H}\yv_{b}= \uv_{b}^{\rm H}\Hm_{bb}\vv_{b}s_{b} + \sum_{v\in \{a,c\}}\uv_{b}^{\rm H}\overline\Hm_{bv}\vv_{v}s_{v} + \uv_{b}^{\rm H}\zv_{b},\\
&\uv_{c}^{\rm H}\yv_{c} = \uv_{c}^{\rm H}\Hm_{cc}\vv_{c}s_{c} + \sum_{v\in \{a,b,d\}}\uv_{c}^{\rm H}\overline\Hm_{cv}\vv_{v}s_{v} + \uv_{c}^{\rm H}\zv_{c},
\end{align*}
where all interference coming from sectors $v\prec_{\pi_{\rm s}}a$ has already been eliminated.
The above observations can be written in vector form (cf. Eq.~\ref{vecobs}) as, 
\begin{equation}
\tilde\yv = {\bf \tilde H}(\overline\alpha)\,\tilde\sv + \tilde\zv,
\label{vecobsa}
\end{equation}
where ${\bf \tilde H}(\overline\alpha)$ depends  on the channel-state parameters $\overline\alpha\triangleq [\alpha_{ab},\alpha_{ac},\alpha_{ad},\alpha_{ba},\alpha_{bc},\alpha_{ca},\alpha_{cb},\alpha_{cd}]\in \{0,1\}^{8}$ and is given by  
\begin{align}
%\resizebox{1\hsize}{!}{$
{\bf \tilde H}(\overline\alpha) =  \begin{bmatrix} \Hm_{aa}\vv_{a} & \overline\Hm_{ab}\vv_{b} & \overline\Hm_{ac}\vv_{c} & \overline\Hm_{ad}\vv_{d} \\ 
 \uv_{b}^{\rm H}\overline\Hm_{ba}\vv_{a} & \uv_{b}^{\rm H}\Hm_{bb}\vv_{b} & \uv_{b}^{\rm H}\overline\Hm_{bc}\vv_{c} & 0 \\ 
 \uv_{c}^{\rm H}\overline\Hm_{ca}\vv_{a} & \uv_{c}^{\rm H}\overline\Hm_{cb}\vv_{b} & \uv_{c}^{\rm H}\Hm_{cc}\vv_{c} & \uv_{c}^{\rm H}\overline\Hm_{cd}\vv_{d}\end{bmatrix}.
% $}
 \nonumber
\end{align}

Notice that  ${\bf \tilde H}(\overline\alpha)$ has the same structure as the matrix $\tilde\Hm$ we considered in Section~\ref{bfscheme}, and as we have already seen in the example of  Fig.~\ref{fig:notrobust}, there exist channel-state configurations (e.g, $\overline\alpha=[0,0,0,0,1,0,1,1])$,  for which ${\bf \tilde H}(\overline\alpha)$ becomes rank-deficient.
However, this is not necessarily a problem here, since we are only interested in decoding the primary sector's message $s_{a}$. In this case, we just have to guarantee that the following condition,
\begin{equation}
[1,0,0,0] \in \mbox{rowspan}\big({\bf \tilde H}(\overline\alpha)\big),
\label{span}
\end{equation}
holds for every channel-state configuration $\overline\alpha$. Of~course, when 
%$\overline\alpha=[1,1,1,1,1,1,1,1]$, 
$\overline\alpha$ is the all-ones vector,
the matrix ${\bf \tilde H}(\overline\alpha)\in \CC^{4\times 4}$ is full-rank and the above condition is automatically satisfied.

In order to show that the primary sector's message can be always be decoded and that (\ref{span}) holds  for all  $\overline\alpha\in \{0,1\}^{8}$, we  will consider here the following  cases:

\vspace{0.1in}
\begin{enumerate}

\item $\alpha_{cd}=0$, for all $[\alpha_{ab},\alpha_{ac},\alpha_{ad},\alpha_{ba},\alpha_{bc},\alpha_{ca},\alpha_{cb}]\in \{0,1\}^{7}$.
\vspace{0.1in}
\item $\alpha_{ad}=1$, for all $[\alpha_{ab},\alpha_{ac},\alpha_{ba},\alpha_{bc},\alpha_{ca},\alpha_{cb},\alpha_{cd}]\in \{0,1\}^{7}$.
\vspace{0.1in}
\item $[\alpha_{cd},\alpha_{ad}]=[1,0]$, for all $[\alpha_{ab},\alpha_{ac},\alpha_{ad},\alpha_{ba},\alpha_{bc},\alpha_{ca},\alpha_{cb},\alpha_{cd}]$ with $\alpha_{ab}\cdot\alpha_{ac}=0$.
\vspace{0.1in}
\item $[\alpha_{cd},\alpha_{ad}]=[1,0]$, for all $[\alpha_{ab},\alpha_{ac},\alpha_{ad},\alpha_{ba},\alpha_{bc},\alpha_{ca},\alpha_{cb},\alpha_{cd}]$ with $\alpha_{ab}\cdot\alpha_{ac}=1$.
\end{enumerate}
\vspace{0.1in}

Notice that these four cases (illustrated in Fig.~\ref{cases})  cover all possible compound channel-state configurations for the interfering links between the sectors $a$, $b$, $c$, and $d$. Before proceeding to examine these cases separately,  we give a lemma that will be repeatedly used. 

\begin{figure}[ht]
                \centering
                \includegraphics[width=.8\columnwidth]{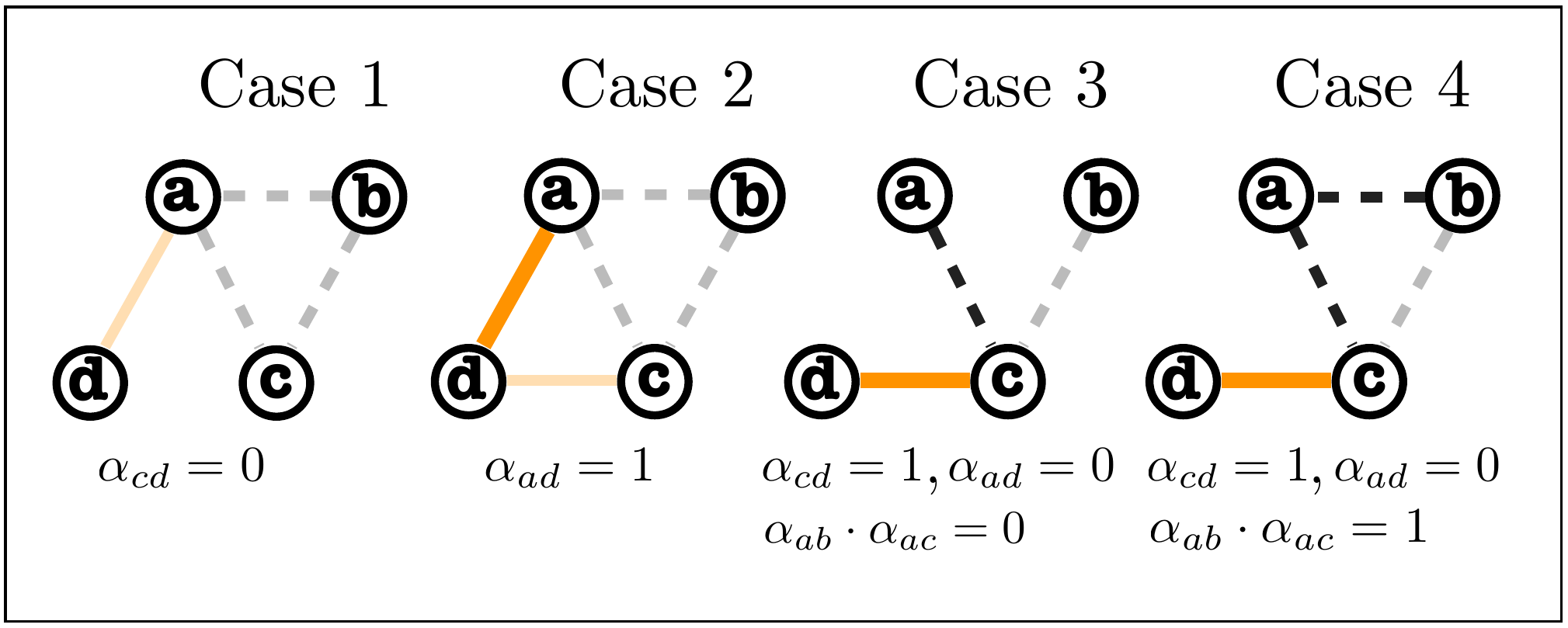}
                \caption{Compound channel-state configurations for sectors $a$, $b$, $c$ and $d$. Case 1 captures all parameter configurations $\overline\alpha$, in which there is no interference between sectors $c$ and $d$, and Case 2 corresponds to configurations 
                %in which sectors $a$ and $d$ interfere with each other. 
                with $\alpha_{ad}=1$.
                In Cases 3 and 4, we assume that $[\alpha_{cd},\alpha_{ad}]=[1,0]$ 
and distinguish between configurations in which sector $a$ observes at most one,  or two interfering signals.                }
\label{cases}
\end{figure}

\clearpage

\begin{lemma}
Let $\Hm$ be an $n\times n$  matrix whose elements are chosen independently at random from a continuous probability distribution. For any binary matrix $\Am\in\{0,1\}^{n\times n}$ with  diagonal elements $a_{ii}=1$, $i=1,...,n$, the rank of the Hadamard (pointwise) product $(\Am\circ\Hm)$ is equal to $n$ with probability one.
\label{lem:hadamard}
\end{lemma}

\begin{proof} 
Let $\Gm = (\Am\circ\Hm)$ and define the multivariate polynomial 
$Q(h_{1,1},h_{1,2},..., h_{n,n})$ as being equal to $\mbox{det}(\Gm)$.
Using the Leibnitz formula for the determinant we have that 
\begin{align}
Q(h_{1,1},h_{1,2},..., h_{n,n}) & = \sum_{\sigma\in S_{n}}\mbox{sgn}(\sigma)\prod_{i=1}^{n}G_{i,\sigma(i)}\\
&=\sum_{\sigma\in S_{n}}\mbox{sgn}(\sigma)\prod_{i=1}^{n}a_{i,\sigma(i)}\prod_{i=1}^{n}h_{i,\sigma(i)}\\
&=\prod_{i=1}^{n}h_{i,i} +\sum_{\sigma\in S_{n}\backslash\{\sigma^{*}\}}\mbox{sgn}(\sigma)\prod_{i=1}^{n}a_{i,\sigma(i)}\prod_{i=1}^{n}h_{i,\sigma(i)}
\end{align}
and hence $Q(h_{1,1},h_{1,2},..., h_{n,n})\not\equiv 0$, for all $\Am$ with $a_{i,i}=1$.
Further, assuming that $h_{i,j}$ are chosen independently at random from a continuous distribution we have that $$\PP[Q(h_{1,1},h_{1,2},..., h_{n,n})\neq0] = 1,$$ and therefore the matrix $\Gm=(\Am\circ\Hm)$ is  full-rank with probability one.
\end{proof}

{\it Case~1:} When $\alpha_{cd}=0$, the receiver $a$ can first zero-force the interference from sector $d$ and obtain 
$$\uv_{a}^{\rm H}\yv_{a}= \uv_{a}^{\rm H}\Hm_{aa}\vv_{a}s_{a} + \sum_{v\in \{b,c\}}\uv_{a}^{\rm H}\overline\Hm_{av}\vv_{v}s_{v} + \uv_{a}^{\rm H}\zv_{a}.$$
Then it can use the projected observations from sectors $b$ and $c$, which are given in this case by
\begin{align*}
&\uv_{b}^{\rm H}\yv_{b}= \uv_{b}^{\rm H}\Hm_{bb}\vv_{b}s_{b} + \sum_{v\in \{a,c\}}\uv_{b}^{\rm H}\overline\Hm_{bv}\vv_{v}s_{v} + \uv_{b}^{\rm H}\zv_{b},\\
&\uv_{c}^{\rm H}\yv_{c} = \uv_{c}^{\rm H}\Hm_{cc}\vv_{c}s_{c} + \sum_{v\in \{a,b\}}\uv_{c}^{\rm H}\overline\Hm_{cv}\vv_{v}s_{v} + \uv_{c}^{\rm H}\zv_{c},
\end{align*}
in order to create a three-dimensional vector observation 
of the form 

\begin{equation*}
%\resizebox{1\hsize}{!}{$
\begin{bmatrix}
\uv_{b}^{\rm H}\yv_{b}\\
\uv_{b}^{\rm H}\yv_{b}\\
\uv_{c}^{\rm H}\yv_{c}
\end{bmatrix}=\underbrace{
\begin{bmatrix}  
\uv_{a}^{\rm H}\Hm_{aa}\vv_{a} &\uv_{a}^{\rm H}\overline\Hm_{ab}\vv_{b} & \uv_{a}^{\rm H}\overline\Hm_{ac}\vv_{c}  \\ 
\uv_{b}^{\rm H}\overline\Hm_{ba}\vv_{a} &\uv_{b}^{\rm H}\Hm_{bb}\vv_{b} & \uv_{b}^{\rm H}\overline\Hm_{bc}\vv_{c}  \\ 
  \uv_{c}^{\rm H}\overline\Hm_{ca}\vv_{a}  &\uv_{c}^{\rm H}\overline\Hm_{cb}\vv_{b} & \uv_{c}^{\rm H}\Hm_{cc}\vv_{c}\end{bmatrix}}_{\triangleq \tilde\Hm(\alpha_{ab},\alpha_{ac},\alpha_{ba},\alpha_{bc},\alpha_{ca},\alpha_{cb})}
 \begin{bmatrix}
 s_{a}\\
s_{b}\\
s_{c}
\end{bmatrix} 
+\tilde\zv.
%$}
\end{equation*}
Notice that $\tilde\Hm(\alpha_{ab},\alpha_{ac},\alpha_{ba},\alpha_{bc},\alpha_{ca},\alpha_{cb})$ 
%can only have zero values in its off-diagonal entries and 
can be written as the pointwise product $(\Am\circ\Hm)=$
\begin{equation*}
%\resizebox{1\hsize}{!}{$
\begin{bmatrix}  
1 &\alpha_{ab} & \alpha_{ac}  \\ 
\alpha_{ba} &1 & \alpha_{bc}  \\ 
  \alpha_{ca}   &\alpha_{cb} & 1\end{bmatrix}\circ
\begin{bmatrix}  
\uv_{a}^{\rm H}\Hm_{aa}\vv_{a} &\uv_{a}^{\rm H}\Hm_{ab}\vv_{b} & \uv_{a}^{\rm H}\Hm_{ac}\vv_{c}  \\ 
\uv_{b}^{\rm H}\Hm_{ba}\vv_{a} &\uv_{b}^{\rm H}\Hm_{bb}\vv_{b} & \uv_{b}^{\rm H}\Hm_{bc}\vv_{c}  \\ 
  \uv_{c}^{\rm H}\Hm_{ca}\vv_{a}  &\uv_{c}^{\rm H}\Hm_{cb}\vv_{b} & \uv_{c}^{\rm H}\Hm_{cc}\vv_{c}\end{bmatrix},
%$}
 \end{equation*}
 where $\Hm$ and $\Am$ satisfy the conditions of Lemma~\ref{lem:hadamard}, and hence  
 it is full-rank for all channel-state parameters $[\alpha_{ab},\alpha_{ac},\alpha_{ba},\alpha_{bc},\alpha_{ca},\alpha_{cb}]$. We can therefore argue that  receiver $a$ is  always able in this case to decode its desired message from the above joint observation.

{\it Case 2:} When $\alpha_{ad}=1$, the equivalent channel matrix $\tilde\Hm(\overline \alpha)$ is 
going to be full-rank for every choice of $[\alpha_{ab},\alpha_{ac},\alpha_{ba},\alpha_{bc},\alpha_{ca},\alpha_{cb},\alpha_{cd}]\in \{0,1\}^{7}$ and hence $s_{a}$ can be decoded directly from (\ref{vecobsa}).  In order to show this we will first write the matrix 
$\tilde\Hm(\overline \alpha)$
%$\tilde\Hm(\alpha_{ab},\alpha_{ac},1,\alpha_{ba},\alpha_{bc},\alpha_{ca},\alpha_{cb},\alpha_{cd})$ 
in its product form $(\Am\circ\tilde \Hm)$, where
\begin{equation*}
\Am=\begin{bmatrix}  
1 &\alpha_{ab} & \alpha_{ac} & 1  \\ 
1 &\alpha_{ab} & \alpha_{ac} & 1 \\ 
\alpha_{ba} &1 & \alpha_{bc} & 0\\ 
  \alpha_{ca}   &\alpha_{cb} & 1 &\alpha_{cd} \end{bmatrix},
 \end{equation*}
 and consider a  permutation matrix $\Pm_{\sigma}$ that reorders the rows of $\Am$ according to  
% $\sigma : \{1,2,3,4\}\rightarrow\{1,2,3,4\}$ with 
 $\sigma(1)=4$,  $\sigma(2)= 1$, $\sigma(3)= 2$, $\sigma(4)= 3$.
%such that the diagonal elements of $\Pm\Am$ are equal to 1.
 We have that \begin{align*} 
 \mbox{rank}(\Am\circ\tilde \Hm) &= \mbox{rank}\big(\Pm_{\sigma}(\Am\circ\tilde \Hm)\big)\\
 &=\mbox{rank}(\Pm_{\sigma}\Am\circ\Pm_{\sigma}\tilde \Hm)
 \end{align*}
 and since $[\Pm_{\sigma}\Am]_{ii} =1$,  $\forall i$,  we can use Lemma~\ref{lem:hadamard} to show that  
 %$\tilde\Hm(\alpha_{ab},\alpha_{ac},1,\alpha_{ba},\alpha_{bc},\alpha_{ca},\alpha_{cb},\alpha_{cd})$
 the above matrix  is indeed full-rank for any choice of channel-state parameters $[\alpha_{ab},\alpha_{ac},\alpha_{ba},\alpha_{bc},\alpha_{ca},\alpha_{cb},\alpha_{cd}]\in \{0,1\}^{7}$.

{\it Case 3:} When $\alpha_{ad}=0$ and $[\alpha_{ab},\alpha_{ac}]=[0,0]$,  receiver $a$ observes no interference and can  directly decode its own message. Now, when $[\alpha_{ab},\alpha_{ac}]\in \{[0,1],$ $[1,0]\}$, the receiver $a$ has only one interfering link which can always be zero-forced from its two-dimensional
observation $\yv_{a}$.  Without loss of generality, assume that $\alpha_{ac}=1$ and choose $\uv_{a}\in \CC^{2}$ such that $\uv_{a}^{\rm H}\Hm_{ac}\vv_{c}=0$. Then, 
\vspace{-0.1in}
$$\uv_{a}^{\rm H}\yv_{a} = \uv_{a}^{\rm H}\Hm_{aa}\vv_{a}s_{a}  + \uv_{a}^{\rm H}\zv_{a},$$ and since $\uv_{a}^{\rm H}\Hm_{aa}\vv_{a}\neq 0$ with probability one,  the message $s_{a}$ can be decoded in this case as well.

{\it Case 4:} In this case, the equivalent channel matrix $\tilde\Hm(\overline \alpha)$ can be written as $(\Am\circ\tilde \Hm)$, with
\begin{equation*}
%\resizebox{0.2\hsize}{!}{$
\Am=\begin{bmatrix}  
1 &1 & 1 & 0  \\ 
1 &1 & 1 & 0 \\ 
\alpha_{ba} &1 & \alpha_{bc} & 0\\ 
\alpha_{ca}   &\alpha_{cb} & 1 &1 \end{bmatrix},
% $}
 \end{equation*}
 and, as in Case~2, we can use Lemma~3 to show that it is full-rank, by swapping the second and third rows of $\Am$. Hence, $s_{a}$ can be decoded from (\ref{vecobsa}) for all $[\alpha_{ba},\alpha_{bc},\alpha_{ca},\alpha_{cb}]\in\{0,1\}^{4}$.

%%%%%%%%%%%%%%%%%%%%%%%%%%%%%%%%%%%%%%%%%%%%%%%%%%%%%%%%
\section{Conclusions} \label{conclusions}

In this work we have shown that the promised DoFs gain of interference alignment can be achieved in cellular networks with straightforward one-shot alignment precoding, without requiring 
symbol extensions over very large number of time-frequency dimensions, or infinite resolution of ``rationally independent'' signal levels. 
In particular, we have shown schemes that achieve $1/2$ DoFs per antenna in the uplink of a cellular system
with three sectors per cell and one active user per sector, where both the user transmitter and the sector receiver have $M$ 
antennas. Our result applies immediately to the case of $M$ even, while it requires extension over two time/frequency varying slots
for $M$ odd. 
%Nevertheless, even in this case we do not require the channel to ``cooperate'' in any special way: the scheme applies verbatim whether the channel is constant or varying over the two precoding slots. 
Furthermore, for the case where there is (possibly) interference between
sectors of the same cell, we have considered a scheme that exploits joint processing (in fact, successive decoding is sufficient)
of the three sectors in the same cell and achieves the same optimal DoFs. Finally, for this scenario we have defined the notion of
``topological robustness'' of a scheme, as the ability to achieve fixed average DoFs irrespectively of the presence/absence 
of the interfering links. We have shown that topologically robust one-shot linear schemes exist, which achieve 
the same optimal DoFs of the original network where all links are present.

The key technology enabler to achieve these results is to allow  base stations  to share their own locally decoded messages with their neighboring 
base station receivers. This framework
%, referred to as Network Interference Cancellation, 
is very different from joint processing of all the cell sites 
as advocated in the so-called ``Wyner model'', which requires  all received signals to be jointly processed at a single central processor. 
As a matter of fact, both joint processing of same-cell sectors and message passing of already (individually) decoded messages to neighboring cells
can be implemented in current cellular technology. 
Therefore, we believe that the results of this paper are not only a step forward in the understanding of the true potential of interference alignment 
in wireless networks, but also provide practical and valuable system design guidelines towards a much more efficient 
interference management in large wireless networks.

%%%%%%%%%%%%%%%%%%%%%%%%%%%%%%%%%%%%%%%%%%%%%%%%%%%%%%%%

\appendices
 
 %%%%%%%%%%%%%%%%%%%%%%%%%%%%%%%%%%%%%%
\section{Proof of Theorem 1}\label{proofthm1}

Consider the directed interference graph ${\cal G}_{\pi^{*}}\big(\cal V,{\cal E}_{\pi^{*}}\big)$
and let $\Vm_{v},\Um_{v}\in \CC^{M\times d_{v}}$ denote the transmit and receive beamforming matrices associated with each node $v\in \cal V$.

We will show here that it is possible to choose $d_{v}$, $\Vm_{v}$ and $\Um_{v}$ for every  $v\in \cal V$ such that the following conditions are satisfied.
\begin{align}
&\Um_{u}^{\rm H}\Hm_{uv}\Vm_{v} = 0,\; \forall [v,u] \in \cal E_{\pi^{*}}\label{eq:condAAA}\\
&\mbox{rank}\left(\Um_{v}^{\rm H}\Hm_{vv}\Vm_{v}\right) = d_{v},\; \forall v\in \cal V, \;\mbox{and}
\label{eq:condBBB}\end{align}
\begin{equation}\frac{1}{|{\cal V}|} \sum_{v\in \cal V}d_{v} \geq \begin{cases}\frac{M}{2}, \;\;\;\;\;\;\;\;\;\;\;\mbox{$M$ is even}\\ \frac{M}{2} -\frac{1}{6}, \;\;\;\;\;\mbox{$M$ is odd.}\end{cases}
\label{eq:condCCC}
\end{equation}

Recall the definitions of  $f(\cdot)$ and $\phi(\cdot)$ that are given in Section~\ref{igraph} and consider the sets
\begin{equation}{\cal V}_{k} = \{v\in {\cal V} : f(\phi(v))=k\},\; k=0,1,2.\label{Vk}\end{equation}
Notice that ${\cal V}_{k}$ satisfy
$${\cal V} = \bigcup_{k=0}^{2}{\cal V}_{k} \;\;\mbox{and}\;\; {\cal V}_{i}\cap{\cal V}_{j} = \emptyset,\; i\neq j$$
and hence form a partition of $\cal V$.

An important observation is that, according to ${\cal E}_{\pi^{*}}$, every receiver associated with a node $u\in{\cal V}_{1}$ has at most one interfering transmitter. 
More specifically, for every $u\in{\cal V}_{1}$ there exist  $[v,u]\in {\cal E}_{\pi^{*}}$ if and only if  there exist $v\in {\cal V}_{2}$ with $\phi(v) = \phi(u)-1-\omega$. 

Similarly, the  receivers associated with the nodes $u\in{\cal V}_{2}$ have at most two interfering transmitters and hence we can argue that for every $u\in{\cal V}_{2}$ there exist $[v,u]\in {\cal E}_{\pi^{*}}$ if and only if  there exist $v_{0}\in {\cal V}_{0}$ with $\phi(v_{0}) = \phi(u)+1$ or  $v_{1}\in {\cal V}_{1}$ with $\phi(v_{1}) = \phi(u)-\omega$. 

Finally, every receiver associated with a node $u\in{\cal V}_{0}$
observes at most three interferers and we have that for every  $u\in{\cal V}_{0}$ there exist $[v,u]\in {\cal E}_{\pi^{*}}$ if and only if there exist $v_{1}\in {\cal V}_{1}$ with $\phi(v_{1}) = \phi(u)+1$ or  $v_{2}\in {\cal V}_{2}$ with $\phi(v_{2}) = \phi(u)-\omega$ or $v_{1}'\in {\cal V}_{1}$ with $\phi(v_{1}') = \phi(u)-1-\omega$.

For any full rank matrix $\Am\in \CC^{m\times n}$ with $m>n$, we let  $\Pm^{\perp}_{\Am}\in \CC^{m\times (m-n)}$ be  a basis for the nullspace of  $\Am^{\rm H}$, such that $(\Pm^{\perp}_{\Am})^{\rm H}\Am = 0$.

%Let $\Pm^{\perp}_{\Am}$ % = {\bf I} - \Am(\Am^{\rm H}\Am)^{-1}\Am^{\rm H}$  
%denote the projection matrix onto the subspace orthogonal to the column span  of a matrix $\Am$  such that $\Pm^{\perp}_{\Am}\Am = 0$.

\subsection{$M$ is even}
Let $d_{v} = \frac{M}{2}$, for all $v \in \cal V$ and 
 consider the following beamforming choices.

\begin{enumerate}

\item[(a)] For all $v_{0} \in {\cal V}_{0}$ such that there exist $v_{1}\in {\cal V}_{1}$ and $u\in {\cal V}_{2}$ with $\phi(v_{1})= \phi(v_{0})-1-\omega$ and $\phi(u)= \phi(v_{0})-1$, set 
$$\Vm_{v_{0}} = \Hm_{uv_{0}}^{-1}\Hm_{uv_{1}}\Vm_{v_{1}}/||\Hm_{uv_{0}}^{-1}\Hm_{uv_{1}}\Vm_{v_{1}}||.$$ Otherwise choose $\Vm_{v_{0}}\in \CC^{M\times d_{v}}$ at random.
\item[(b)] For all $v_{1} \in {\cal V}_{1}$ such that there exists $v_{2}\in {\cal V}_{2}$ and $u\in {\cal V}_{0}$ with $\phi(v_{2})= \phi(v_{1})+1$ and $\phi(u)= \phi(v_{1})-1+\omega$, set 
$$\Vm_{v_{1}} = \Hm_{uv_{1}}^{-1}\Hm_{uv_{2}}\Vm_{v_{2}}/||\Hm_{uv_{1}}^{-1}\Hm_{uv_{2}}\Vm_{v_{2}}||.$$ Otherwise choose $\Vm_{v_{1}} \in \CC^{M\times d_{v}}$ at random.

\item[(c)] For all $v_{2} \in {\cal V}_{2}$ such that there exists $v_{1}\in {\cal V}_{1}$ and $u\in {\cal V}_{0}$ with $\phi(v_{1})= \phi(v_{2})+1+\omega$ and $\phi(u)= \phi(v_{2})+\omega$, set 
$$\Vm_{v_{2}} = \Hm_{uv_{2}}^{-1}\Hm_{uv_{1}}\Vm_{v_{1}}/||\Hm_{uv_{2}}^{-1}\Hm_{uv_{1}}\Vm_{v_{1}}||.$$ Otherwise choose $\Vm_{v_{2}} \in \CC^{M\times d_{v}}$ at random.

\item[(d)] For all $u \in \cal V$ such that there exists an edge $[v,u]\in {\cal E}_{\pi^{*}}$ for some $v \in \cal V$, set  $$\Um_{u} =\Pm^{\perp}_{\Hm_{uv}\Vm_{v}}.$$ Otherwise choose $\Um_{u}\in \CC^{M\times d_{v}}$ at random.

\end{enumerate}

\begin{figure}[ht]
                \centering
                \includegraphics[width=.8\columnwidth]{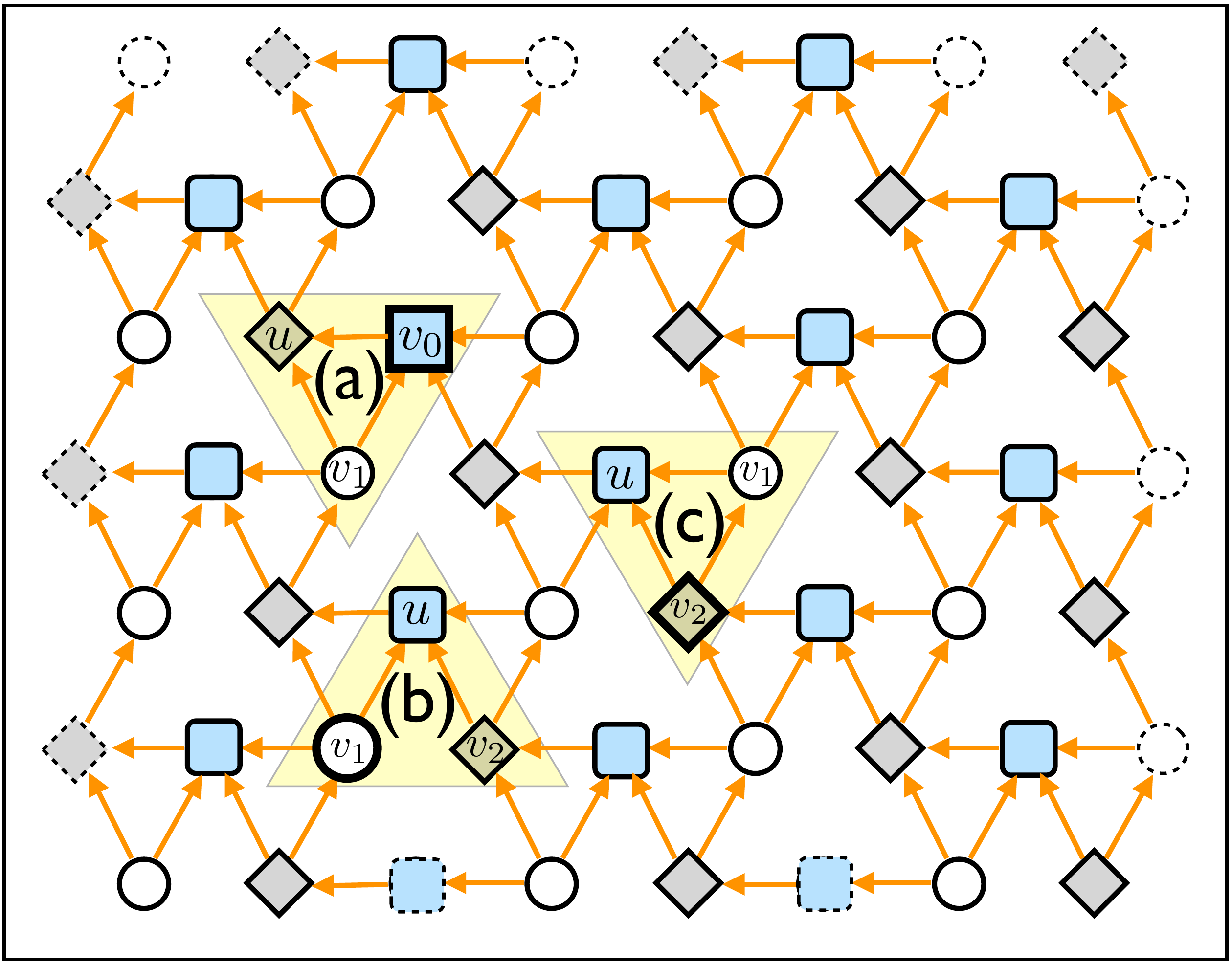}
                \caption{The beamforming choices (a), (b) and (c). In this example, all the nodes with dashed outline have chosen their beamforming vectors at random.  }
\label{cases}
\end{figure}

Notice that  the conditions (\ref{eq:condBBB}) and (\ref{eq:condCCC}) are automatically  satisfied (with probability one) since  $d_{v}=\frac{M}{2},\, \forall v\in \cal V$ and  $\Hm_{uv}$ are chosen at random from a continuous distribution. 
We are going to show next that the conditions (\ref{eq:condAAA}) are also satisfied for all $[v,u] \in \cal E_{\pi^{*}}$. Consider the sets: $${\cal E}^{({k})}_{\pi^{*}} = \{[v,u] \in {\cal E}_{\pi^{*}} : u\in {\cal V}_{k}\}.$$
As we have seen, every receiver associated with $u\in {\cal V}_{1}$ observers at most one interfering transmitter and hence $\Um_{u}^{\rm H}\Hm_{uv}\Vm_{v} = (\Pm^{\perp}_{\Hm_{uv}\Vm_{v}})^{\rm H}\Hm_{uv}\Vm_{v}= 0$ for all $[v,u]\in {\cal E}^{({1})}_{\pi^{*}}$. For every receiver $u\in {\cal V}_{2}$ there exist at most two  interfering transmitters  given by $v_{0}\in {\cal V}_{0}$, $v_{1}\in {\cal V}_{1}$ with $\phi(v_{0}) = \phi(u)+1$   and $\phi(v_{1}) = \phi(u)-\omega$. Notice that in this case $\phi(v_{1})= \phi(v_{0})-1-\omega$ and according to (a), $\Hm_{uv_{0}}\Vm_{v_{0}} = \Hm_{uv_{1}}\Vm_{v_{1}}/||\Hm_{uv_{0}}^{-1}\Hm_{uv_{1}}\Vm_{v_{1}}||\in \mbox{span}(\Hm_{uv_{1}}\Vm_{v_{1}})$. Hence we can also argue that $\Um_{u}^{\rm H}\Hm_{uv}\Vm_{v} = 0$ for all $[v,u]\in {\cal E}^{({2})}_{\pi^{*}}$.
Now consider the set ${\cal E}^{({2})}_{\pi^{*}}$. In a similar fashion, we can see that according to  the beamforming choices (b) and (c),  all interference observed by receivers $u\in {\cal V}_{0}$ aligns in $M/2$ dimensions. That is for every $u\in {\cal V}_{0}$ that observes interference from the transmitters 
$v_{1}\in {\cal V}_{1}$, $v_{2}\in {\cal V}_{2}$ and $v_{1}'\in {\cal V}_{1}$ with $\phi(v_{1}) = \phi(u)+1$ , $\phi(v_{2}) = \phi(u)-\omega$ and $\phi(v_{1}') = \phi(u)-1-\omega$ we have that $\mbox{span}(\Hm_{uv_{1}}\Vm_{v_{1}})=\mbox{span}(\Hm_{uv_{1}'}\Vm_{v_{1}'})= \mbox{span}(\Hm_{uv_{2}}\Vm_{v_{2}})$ and hence $\Um_{u}^{\rm H}\Hm_{uv}\Vm_{v} = 0$ for all $[v,u]\in {\cal E}^{({3})}_{\pi^{*}}$ as well.
Since by definition ${\cal E}^{({1})}_{\pi^{*}}\cup{\cal E}^{({2})}_{\pi^{*}}\cup{\cal E}^{({3})}_{\pi^{*}}= {\cal E}_{\pi^{*}}$ we conclude that the conditions (\ref{eq:condAAA}) are satisfied for all $[v,u] \in \cal E_{\pi^{*}}$.

\subsection{$M$ is odd}

Let $\tilde d_{v} = \frac{M-1}{2},\forall v\in \cal V$ and consider the beamforming matrices $\tilde \Um_{v}\in \CC^{M\times \frac{M+1}{2}}$ and  $\tilde \Vm_{v}\in \CC^{M\times \frac{M-1}{2}}$ given by (a), (b), (c) and (d). Following the same arguments as before we can see that if we use the above beamforming subspaces for transmission, every receiver $u\in \cal V$ will observe interference aligned in $\frac{M-1}{2}$ dimensions and hence we could directly achieve $\frac{1}{|{\cal V}|} \sum_{v\in \cal V}\tilde d_{v} =\frac{M-1}{2}$. 

Notice however that in this case, any receiver 
that zero-forces $\frac{M-1}{2}$ out of $M$ dimensions  can in principle support one extra dimension for transmission since $M-\frac{M-1}{2}=\tilde d_{v} +1$. Furthermore, any receiver that uses  only $\tilde d_{v}=\frac{M-1}{2}$ dimensions for desired symbols can zero-force the remaining $\frac{M+1}{2}$ dimensions and can hence tolerate one additional interfering stream from its neighbors.

Let ${\cal V}_{*}\subseteq\cal V$ be a set of nodes such that  the following two conditions are satisfied:
\begin{align}
&(u,v)\notin{\cal E}, \forall u,v \in \cal  {\cal V}_{*} \label{a1}\\
&\big|\{v \in {\cal V}_{*}: [v,u]\in  {\cal E}_{\pi^{*}} \}\big|\leq 1, \forall u\notin {\cal V}_{*}.
\label{a2}\end{align} 

The first condition requires that $\cal V_{*}$ is an independent set in $\cal G(V,E)$ and the second one  states that for every $u\notin {\cal V}_{*}$ there is at most one $v \in {\cal V}_{*}$ such that 
$[v,u]\in  {\cal E}_{\pi^{*}}$.
Consider the following beamforming choices given in terms of $\tilde\Vm_{v}$ and $\tilde\Um_{v}$:
\begin{itemize}
\item For all $v\in {\cal V}_{*}$ set $\Vm_{v} = [\tilde \Vm_{v}, \vv_{v}]/||[\tilde \Vm_{v}, \vv_{v}]||$, for some $\vv_{v}\notin \mbox{span}(\tilde \Vm_{v})$ and let  $\Um_{v}^{\rm H} = \tilde \Um_{v}^{\rm H}$.

\item For all $u\notin {\cal V}_{*}$ set $\Vm_{u} = \tilde \Vm_{u}$. If there exists $v\in \cal V_{*}$ such that $[v,u]\in {\cal E}_{\pi^{*}}$ set  and $\Um_{u}^{\rm H} = \Pm^{\perp}_{\tilde\Um_{u}^{\rm H}\Hm_{uv}\vv_{v}}\tilde\Um_{u}^{\rm H}$. Otherwise set  $\Um_{u}^{\rm H}= \tilde\Um_{u}^{\rm H}$.

\end{itemize}

We have  that $d_{v}=\tilde d_{v}+1$ for all $v\in \cal V_{*}$ and $d_{v}=\tilde d_{v}$ for all $v\notin \cal V_{*}$. We are going to show next that with the above beamforming choices 
the interference alignment conditions (\ref{eq:condAAA}) and (\ref{eq:condBBB})
%\begin{align*}
%&\Um_{u}^{\rm H}\Hm_{uv}\Vm_{v} = 0,\; \forall [v,u] \in \cal E_{\pi^{*}},\;\mbox{and}\\
%&\mbox{rank}\left(\Um_{v}^{\rm H}\Hm_{vv}\Vm_{v}\right) = d_{v},\; \forall v\in \cal V. 
%\end{align*}
are satisfied and hence the  average (per sector) degrees of freedom
\begin{equation}\frac{1}{|{\cal V}|} \sum_{v\in \cal V} d_{v} =\frac{M-1}{2} + \frac{|{\cal V_{*}}|}{|{\cal V}|}\end{equation}
are achievable. Then we are going to show that it is always possible to find a set $\cal V_{*}\subseteq \cal V$ that satisfies the properties (\ref{a1}) and (\ref{a2}) with $|{\cal V_{*}}| \geq \frac{|{\cal V}|}{3}$ and hence show that $\frac{1}{|{\cal V}|} \sum_{v\in \cal V} d_{v} \geq \frac{M}{2}-\frac{1}{6}$ as required by (\ref{eq:condCCC}).

First notice that the conditions (\ref{eq:condBBB}) are automatically satisfied (with probability one) since all the channel matrices $\Hm_{uv}$  have been chosen at random from a continuous distribution. 
In order to show that the zero-forcing conditions (\ref{eq:condAAA}) are also satisfied, consider the sets 
$$\overline{{\cal V}}_{*}^{(0)} \triangleq \{u\in {\cal V}\backslash {\cal V}_{*}:\big|\{v \in {\cal V}_{*}: [v,u]\in  {\cal E}_{\pi^{*}} \}\big|= 0 \},$$ 
$$\overline{{\cal V}}_{*}^{(1)} \triangleq \{u\in {\cal V}\backslash {\cal V}_{*}:\big|\{v \in {\cal V}_{*}: [v,u]\in  {\cal E}_{\pi^{*}} \}\big|= 1 \}.$$
Notice that according to (\ref{a2}), the sets ${\cal V}_{*}$, $\overline{{\cal V}}_{*}^{(0)}$ and $\overline{{\cal V}}_{*}^{(1)}$ form a partition of $\cal V$.
According to (\ref{a1}), every receiver associated with $u\in \cal V_{*}$ will only observe interference from transmitters $v\notin \cal V_{*}$ and since $\Um_{u} = \tilde \Um_{u}$ for all $u\in \cal V_{*}$ and  $\Vm_{v} = \tilde \Vm_{v}$ for all $v\notin \cal V_{*}$, we have that 
$\Um_{u}^{\rm H}\Hm_{uv}\Vm_{v}=\tilde\Um_{u}^{\rm H}\Hm_{uv}\tilde\Vm_{v} = 0$, $\forall \{[u,v]\in {\cal E}_{\pi^{*}}:u\in {\cal V}_{*}\}$. Similarly, $\Um_{u}^{\rm H}\Hm_{uv}\Vm_{v}=\tilde\Um_{u}^{\rm H}\Hm_{uv}\tilde\Vm_{v} = 0$, $\forall \{[u,v]\in {\cal E}_{\pi^{*}}:u\in \overline{{\cal V}}_{*}^{(0)}\}$. Now, consider the receivers associated with the nodes $u\in \overline{{\cal V}}_{*}^{(1)}$ and let  $v_{0}\in \{v \in {\cal V}_{*}: [v,u]\in  {\cal E}_{\pi^{*}}\}$. By construction, we have that $\mbox{span}(\Hm_{uv}\Vm_{v})\subseteq \mbox{span}(\Hm_{uv_{0}}\Vm_{v_{0}})=\mbox{span}(\Hm_{uv_{0}}\tilde\Vm_{v_{0}})\cup\mbox{span}(\Hm_{uv_{0}}\vv_{v_{0}})$ for all $v \in {\cal V}$ such that $[v,u]\in  {\cal E}_{\pi^{*}}$ and  
since $\Um_{u}^{\rm H}\Hm_{uv_{0}}\Vm_{v_{0}} = [\tilde\Um_{u}^{\rm H}\Hm_{uv_{0}}\tilde\Vm_{v_{0}}, \Um_{u}^{\rm H}\Hm_{uv_{0}}\vv_{v_{0}}]= [\Pm^{\perp}_{\tilde\Um_{u}^{\rm H}\Hm_{uv_{0}}\vv_{v_{0}}}\tilde\Um_{u}^{\rm H}\Hm_{uv_{0}}\tilde\Vm_{v_{0}}, \Pm^{\perp}_{\tilde\Um_{u}^{\rm H}\Hm_{uv_{0}}\vv_{v_{0}}}\tilde\Um_{u}^{\rm H}\Hm_{uv_{0}}\vv_{v_{0}}] = 0$, we get $\Um_{u}^{\rm H}\Hm_{uv}\Vm_{v}= 0$, $\forall \{[u,v]\in {\cal E}_{\pi^{*}}:u\in \overline{{\cal V}}_{*}^{(1)}\}$. 
Putting everything together, since the sets ${\cal V}_{*}$, $\overline{{\cal V}}_{*}^{(0)}$ and $\overline{{\cal V}}_{*}^{(1)}$ form a partition of $\cal V$, we can argue that $\Um_{u}^{\rm H}\Hm_{uv}\Vm_{v}= 0$  for all  $[u,v]\in {\cal E}_{\pi^{*}}$ and hence show that the conditions (\ref{eq:condAAA}) are satisfied. 

For the last part of the proof consider the sets ${\cal V}_{k}$ given in (\ref{Vk}) and recall that they form a partition of $\cal V$. First notice that since $|{\cal V}| = |{\cal V}_{0}|+|{\cal V}_{1}|+|{\cal V}_{2}|$, there must exist some $k^{*}\in \{0,1,2\}$ such that  
$|{\cal V}_{k^{*}}|\geq \frac{|{\cal V}|}{3}$. {By symmetry, we have that $|{\cal V}_{1}|=|{\cal V}_{2}|$ since  $f(z)=1 \Leftrightarrow f(-z)=2,\forall z\in \ZZ(\omega)$, and hence we can assume without loss of generality that $k^{*}$ is either $0$ or $2$.} %and observe that for any $v_{1}\in {\cal V}_{1}$ the vertex given by $\phi^{-1}(-\phi(v_{1}))$

Furthermore the set ${\cal V}_{k^{*}}$ will satisfy (\ref{a1}) since  for every $(u,v)\in \cal E$ we can write $\phi(u)=\phi(v) + \delta$, for some $\delta\in\{\pm 1,\pm\omega,\pm(\omega+1)\}$ and hence $f(\phi(u)) \neq f(\phi(v)),\forall(u,v)\in \cal E$. 

Finally, recall that 1) for every $u\in{\cal V}_{1}$ there exist  $[v,u]\in {\cal E}_{\pi^{*}}$ if and only if  there exist $v\in {\cal V}_{2}$ with $\phi(v) = \phi(u)-1-\omega$,  2) for every $u\in{\cal V}_{2}$ there exist $[v,u]\in {\cal E}_{\pi^{*}}$ if and only if  there exist $v_{0}\in {\cal V}_{0}$ with $\phi(v_{0}) = \phi(u)+1$ or  $v_{1}\in {\cal V}_{1}$ with $\phi(v_{1}) = \phi(u)-\omega$ and 3) for every $u\in{\cal V}_{0}$ there exist $[v,u]\in {\cal E}_{\pi^{*}}$ if and only if there exist $v_{1}\in {\cal V}_{1}$ with $\phi(v_{1}) = \phi(u)+1$ or  $v_{2}\in {\cal V}_{2}$ with $\phi(v_{2}) = \phi(u)-\omega$ or $v_{1}'\in {\cal V}_{1}$ with $\phi(v_{1}') = \phi(u)-1-\omega$. Therefore, 
$\big|\{v \in {\cal V}_{k^{*}}: [v,u]\in  {\cal E}_{\pi^{*}} \}\big|\leq 1, \forall u\in {\cal V}$ and hence the set ${\cal V}_{k^{*}}$, $k^{*}\in\{0,2\}$ will also satisfy (\ref{a2}).

In order to complete the proof we  set ${\cal V}_{*} = {\cal V}_{k^{*}}$ and obtain  $\frac{1}{|{\cal V}|} \sum_{v\in \cal V} d_{v} = \frac{M-1}{2} + \frac{|{{\cal V}_{k^{*}}}|}{|{\cal V}|} \geq \frac{M}{2}-\frac{1}{6}$ as required.

 %Let ${\cal V}_{*}\subseteq\cal V$ be a set of nodes such that  $(u,v)\notin\cal E$, for all $u,v \in \cal  {\cal V}_{*}$ and consider the following beamforming choices:
%\begin{itemize}
%\item $\Vm_{v} = [\tilde \Vm_{v}, \vv_{v}]$ and $\Um_{v}^{\rm H} = \tilde \Um_{v}^{\rm H}$, $\forall v\in {\cal V}_{*}$
%
%\end{itemize}

%$d_{v} = \tilde d_{v} +1,v \in \cal  {\cal V}_{*}$ and 
%Similarly, a receiver that uses  $\tilde d_{v}=\frac{M-1}{2}$ dimensions for desired symbols can tolerate one extra interfering dimension. 

%%%%%%%%%%%%%%%%%%%%%%%%%%%%%%%%%%%%%%%%%%%%%%%%%
\section{Proof of Lemma \ref{lem:2}}  \label{proof-lemma2}

Recall that the set of vertices $\cal V$ of the graph $\cal G(V,E)$ is defined in terms of a parameter $r\geq 1$ as 
\begin{equation*}
{\cal V} = \left\{ \phi^{-1}(z) : z\in \mathbb{Z}(\omega)\cap{\cal B}_{r}\right\},
\end{equation*}
where
$${\cal B}_{r} \triangleq \left\{z\in \mathbb{C}:|{\rm Re}(z)|\leq r , |{\rm Im}(z)|\leq \frac{\sqrt{3}r}{2}\right\}.$$
Since the size of the graph depends on the choice of  $r$, we will consider here the sequence of graphs ${\cal G}^{(r)}({\cal V}^{(r)},{\cal E}^{(r)})$,  indexed by $r \in \ZZ^{+}$ and provide the corresponding results in terms of the above parameter.

%%%%%%%%%%%%%%%%%%%%%%%%%%%%%%%
\subsection{The cardinality of ${\cal V}^{(r)}$}

By definition $|{\cal V}^{(r)}| =  |\mathbb{Z}(\omega)\cap{\cal B}_{r}|$. Hence, our goal is to count the number of Eisenstein integers that belong to the set $ \mathbb{Z}(\omega)\cap{\cal B}_{r}$. 
We define the sets
\begin{equation}
L(k) = \left\{z\in\mathbb{Z}(\omega)\cap{\cal B}_{r} : |{\rm Im}(z)| = \frac{\sqrt{3}k}{2}\right\}\end{equation}
for all $k\in\{-r,...,0,...,r\}$. Notice that the sets $L(k)$ contain all the Eisenstein integers that  lie on the same horizontal line on the complex plane and hence $\bigcup_{k}L(k)$  forms a partition of the set $\mathbb{Z}(\omega)\cap{\cal B}_{r}$. Therefore,
$$|\mathbb{Z}(\omega)\cap{\cal B}_{r}| = \sum_{k=-r}^{r}|L(k)|.$$
A key observation coming  from the triangular structure of $\ZZ(\omega)$ is that 
\begin{equation*}
|L(k)| = \begin{cases}|L(0)|, \;\;k \mbox{ is even}\\|L(1)|, \;\;k \mbox{ is odd.} \end{cases}
\end{equation*}
Hence, we can write
$$|\mathbb{Z}(\omega)\cap{\cal B}_{r}| = K_{\rm even}^{[r]}|L(0)| + K_{\rm odd}^{[r]}|L(1)|.$$
where $K_{\rm even}^{[r]},K_{\rm odd}^{[r]}$ denote the cardinalities of even and odd integers in $\{-r,...,0,...,r\}$. 

If $r$ is even then $K_{\rm even}^{[r]} = r+1$ and $K_{\rm even}^{[r]} = r$, whereas if $r$ is odd then $K_{\rm even}^{[r]} = r$ and $K_{\rm even}^{[r]} = r+1$. Since $|L(0)|=2r+1$ and $|L(1)|=2r$ for all $r\geq 1$ we have that
\begin{equation}
|{\cal V}^{(r)}| = |\mathbb{Z}(\omega)\cap{\cal B}_{r}| = \begin{cases}4r^{2}+3r+1, \;\;r \mbox{ is even}\\4r^{2}+3r, \;\;\;\;\;\;\;\;\,r \mbox{ is odd.} \end{cases}
\label{cardV}
\end{equation}

%%%%%%%%%%%%%%%%%%%%%%%%%%%%%%
\subsection{The cardinality of ${\cal T}^{(r)}$}

We will associate here each ordered vertex triplet $[u,v,w]\in {\cal T}^{(r)}$ with its leading vertex $u\in {\cal V}^{(r)}$ in an one-to-one fashion and define the set
$${\cal A}^{(r)} =\{\phi^{-1}(u)\in \mathbb{Z}(\omega)\cap{\cal B}_{r} : [u,v,w]\in {\cal T}^{(r)}\}.$$
In order to determine the cardinality of ${\cal T}^{(r)}$,  it suffices to count the number of Eisenstein integers that belong to the set ${\cal A}^{(r)}$, since $|{\cal T}^{(r)}| =|{\cal A}^{(r)}|$ by definition.
Consider the sets

$$S(k)= {\cal A}^{(r)}\cap L(k)$$
for all $k\in\{-r,...,0,...,r-1\}$. The set $S(k)$ contains all the Eisenstein integers that are associated with a leading vertex of a triangle and lie on the same horizontal line $L(k)$.  
As before, $\bigcup_{k}S(k)$
forms a partition of ${\cal A}^{(r)}$ and hence
$$|{\cal A}^{(r)}| = \sum_{k=-r}^{r-1} |S(k)|.$$
Intuitively $|S(k)|$ counts the number of triangles that are formed between the lines $L(k)$ and $L(k+1)$ and hence the total number of triangles can be obtained by adding all $|S(k)|$ up to $k=r-1$.

It is not hard to verify that 
\begin{equation*}
|S(k)| = \begin{cases}|S(0)|, \;\;k \mbox{ is even}\\|S(1)|, \;\;k \mbox{ is odd.} \end{cases}
\end{equation*}
for all $r\geq 2$ and hence 
$$|{\cal A}^{(r)}|  = K_{\rm even}^{[r]}|S(0)| + K_{\rm odd}^{[r]}|S(1)|$$
where $K_{\rm even}^{[r]},K_{\rm odd}^{[r]}$ denote the cardinalities of even and odd integers in $\{-r,...,0,...,r-1\}$. We have that $K_{\rm even}^{[r]}=K_{\rm odd}^{[r]}=r$ and hence
$$|{\cal A}^{(r)}|  = r\left(|S(0)| + |S(1)|\right).$$
It follows from the definitions of ${\cal T}^{(r)}$, ${\cal A}^{(r)}$ and $S(0)$ that  
\begin{equation*}
z\in S(0) \Leftrightarrow \begin{cases} z\in L(0), f(z)\neq 0\; \mbox{and}\\ z+\omega,z+\omega+1 \in L(1).\end{cases}
\end{equation*}
We can argue hence that the set $S(0)$ hence contains the integers $a \in \{-r+1,...r-1\}$ for which $f(a)=[a]\mbox{mod}\,3 \neq 0$.

Similarly,
\begin{equation*}
z\in S(1) \Leftrightarrow \begin{cases} z\in L(1), f(z)\neq 0\; \mbox{and}\\ z+\omega,z+\omega+1 \in L(2).\end{cases}
\end{equation*}
And hence the set $S(1)$ contains the Eisenstein integers $z=a + \omega$, for all $a\in \{-r+1,r\}$ that satisfy $f(z)=[a+1]\mbox{mod}\,3 \neq 0$.

It follows that 
\begin{align} 
&|S(0)| = 2\left(r-1-\left \lfloor\frac{r-1}{3}\right \rfloor\right), \;\mbox{and}\nonumber\\
&|S(1)| = 2r-1-\left \lfloor\frac{r+1}{3}\right \rfloor-\left \lfloor\frac{r-2}{3}\right \rfloor.\nonumber
\end{align}
We can hence conclude that 
%\begin{align} %split equation
%|{\cal T}^{(r)}| =r  \Bigg(&4r - 
%\left \lfloor\frac{r-2}{3}\right\rfloor -
%2\left \lfloor\frac{r-1}{3}\right\rfloor \nonumber\\&-
%\left \lfloor\frac{r+1}{3}\right\rfloor -3\Bigg). 
%\label{cardT}
%\end{align}
\begin{align}
|{\cal T}^{(r)}| =r  \Bigg(4r - 
\left \lfloor\frac{r-2}{3}\right\rfloor -
2\left \lfloor\frac{r-1}{3}\right\rfloor -
\left \lfloor\frac{r+1}{3}\right\rfloor -3\Bigg). 
\label{cardT}
\end{align}

%%%%%%%%%%%%%%%%%%%%%%%%%%%%%%%%%%%%%
\subsection{The cardinality of ${\cal V}_{\rm ex}^{(r)}$}

We will upper bound $|{\cal V}_{\rm ex}^{(r)}|$ as follows.
From Lemma~\ref{lem1} we have that
\begin{equation*}
\sum_{u\in {\cal V}_{\rm ex}^{(r)}}\hspace{-0.1in}\left(1-\frac{n_{u}}{2}\right)x_{u}\hspace{-0.05in} = \hspace{-0.1in}\sum_{v\in{\cal V}^{(r)}}x_{v} - \hspace{-0.2in} \sum_{[i,j,k]\in {\cal T}^{(r)}}\hspace{-0.1in}\left(\frac{x_{i}+x_{j}+x_{k}}{2}\right),
\end{equation*}
for any $\{x_{v}: v\in {\cal V}^{(r)}\}$. Setting $x_{v}=1,\,\forall v\in {\cal V}^{(r)}$, we obtain

\begin{equation*}
|{\cal V}_{\rm ex}^{(r)}|-\sum_{u\in {\cal V}_{\rm ex}}\frac{n_{u}}{2} = |{\cal V}^{(r)}| - \frac{3}{2}|{\cal T}^{(r)}|.
\end{equation*}
Since $n_{v}\leq 1$ for all $v\in {\cal V}_{\rm ex}^{(r)}$ we have that 
$$\sum_{u\in {\cal V}_{\rm ex}}\frac{n_{u}}{2} \leq \frac{1}{2}|{\cal V}_{\rm ex}^{(r)}|, $$
and hence 
\begin{equation}
|{\cal V}_{\rm ex}^{(r)}| \leq 2|{\cal V}^{(r)}| - 3|{\cal T}^{(r)}|.
\label{cardVexTMP}
\end{equation}

We can lower bound $|{\cal T}^{(r)}|$ from (\ref{cardT}) as 
\begin{eqnarray}
|{\cal T}^{(r)}| &\geq& r  \left(4r - 
\frac{r-2}{3} -
 \frac{2r-2}{3}-
\frac{r+1}{3} -3\right) \nonumber \\
&=& r  \left(4r - 
\frac{4r-3}{3} -3\right) \nonumber \\
&=& \frac{8}{3}r^{2} - 2r.
\label{cardTbound}
\end{eqnarray}

From (\ref{cardV}), (\ref{cardVexTMP}) and (\ref{cardTbound}), we have that 
\begin{eqnarray}
|{\cal V}_{\rm ex}^{(r)}| &\leq&  2\left(4r^{2} + 3r + 1\right) - 3\left(\frac{8r^{2}}{3} - 2r\right)\nonumber\\&=& 8r^{2} + 6r + 2 - 8r^{2} + 6r \nonumber\\
&=&12r + 2.
\label{cardVex}
\end{eqnarray}

%%%%%%%%%%%%%%%%%%%%%%%%%%%%%%%%%%%%%%%%%%
\subsection{Proof of $|{\cal T}^{(r)}| = \frac{2}{3}|{\cal V}^{(r)}|$}

First we will upper bound $|{\cal T}^{(r)}|$ using the  inequality

$$\left \lfloor\frac{x}{3}\right\rfloor \geq \frac{x-2}{3}, \; \forall x\in \RR.$$
Applying the above inequality in (\ref{cardT}) we obtain 
\begin{eqnarray}
|{\cal T}^{(r)}| &\leq& r  \left(4r - \frac{4r-11}{3} -3 \right)=\frac{8r^{2}+2r}{3}.\nonumber
\end{eqnarray}
From (\ref{cardV}) we can see that 
$$2|{\cal V}^{(r)}|\geq 8r^{2}+2r,$$
and hence it follows that
$|{\cal T}^{(r)}| \leq  \frac{2|{\cal V}^{(r)}|}{3}$, which completes the proof.

%%%%%%%%%%%%%%%%%%%%%%%%%%%%%%%%%%%%%%%%%%%%%%%%
\subsection{Proof of $|{\cal V}_{\rm ex}^{(r)}| = {\cal O}\left(\sqrt{|{\cal V}^{(r)}|}\right)$}

From (\ref{cardV}) it follows that $\sqrt{|{\cal V}^{(r)}|} \geq 2r$ for all $r\geq 1$.
From (\ref{cardVex}) we have that 
$$|{\cal V}_{\rm ex}^{(r)}| \leq 12r +2 \leq 7\sqrt{|{\cal V}^{(r)}|},\; \forall r\geq 1,$$
and hence $|{\cal V}_{\rm ex}^{(r)}| = {\cal O}\left(\sqrt{|{\cal V}^{(r)}|}\right)$.

%%%%%%%%%%%%%%%%%%%%%%%%%%%%%%%%%%%%%%%%%%%%
\section{Proof of Theorem 2}\label{proof:thm2}

Consider the directed interference graph $\cal G_{\pi}(\cal V,\cal E_{\pi})$ and assume that there exist full rank matrices $\Um_{v},\Vm_{v}\in \CC^{M\times d_{v}}$, $v\in \cal V$ such that  
\begin{align}
&\Um_{v}^{\rm H}\Hm_{vu}\Vm_{u} = 0,\; \forall [u,v] \in \cal E_{\pi}\label{eq:condAA}\\
&\mbox{rank}\left(\Um_{v}^{\rm H}\Hm_{vv}\Vm_{v}\right) = d_{v},\; \forall v\in \cal V,\label{eq:condBB}
\end{align}
where $\Hm_{uv}\in \CC^{M\times M}$ have been chosen at random from a continuous distribution. 
Then, $\{d_{v}: v\in \cal V\}$ must satisfy 
\begin{align}
&d_{v}\in \{0,1,...,M\}\,,\;\forall v\in\cal V\label{eq:cond11}\\
&d_{v} + d_{u} \leq M \,,\;\forall (u,v)\in\cal E.\label{eq:cond22}
\end{align}

The first condition follows trivially from the fact that $\mbox{rank}\left(\Um_{v}^{\rm H}\Hm_{vv}\Vm_{v}\right)\leq M$.  The second condition follows from (\ref{eq:condAA}): The columns of the matrices 
$\Hm_{vu}\Vm_{u}$ and $\Um_{v}$ span two orthogonal subspaces of $\CC^{M}$. Since $\mbox{rank}\left(\Hm_{vu}\Vm_{u}\right) = d_{u}$ and $\mbox{rank}\left(\Um_{v}\right) = d_{v}$, the columns of the composite matrix $ [\Hm_{vu}\Vm_{u},\Um_{v}]$ span a $(d_{u}+d_{v})$-dimensional subspace of $\CC^{M}$ and hence $d_{u}+d_{v}\leq M$ , for all $[u,v]\in \cal E_{\pi}$. Now, for any $(u,v)\in \cal E$ and any $\pi$, either $[u,v]$ or $[v,u]$ must be in  $\cal E_{\pi}$. Since, $d_{u}+d_{v}\leq M$ is symmetric in $(d_{u},d_{v})$, we can write the above inequalities for all $(u,v)\in \cal E$.

The above conditions are {\it necessary} for all $\{d_{v}: v\in \cal V\}$ that can be  achieved in $\cal G_{\pi}(\cal V,\cal E_{\pi})$ for any decoding order $\pi$ and  any linear beamforming scheme that does not use symbol extensions. We will use these conditions here to upper bound the average (per sector) achievable degrees of freedom in our framework. 
%Further, notice that (\ref{eq:cond22}) does not depend on the specific choice of $\pi$ and is equivalent  to the (redundant) set of inequalities $d_{v} + d_{u} \leq M \,,\;\forall (u,v)\in\cal E$.  

From Lemma~\ref{lem1} we can write 
\begin{equation}
\frac{1}{|{\cal V}|} \sum_{v\in \cal V}d_{v} = \frac{1}{2|{\cal V}|}\sum_{[i,j,k]\in \cal T}s_{(d_{i},d_{j},d_{k})} + \frac{D_{\rm ex}}{|\cal V|}\, ,
\end{equation}  
where 
\begin{equation}
s_{(d_{i},d_{j},d_{k})}\triangleq d_{i}+d_{j}+d_{k}
\end{equation} and $$D_{\rm ex} = \sum_{v\in {\cal V}_{\rm ex}}\left(1-\frac{n_{v}}{2}\right)d_{v}.$$

Since $d_{v}\leq M$ and $n_{v}\geq 0$ for all $v\in \cal V$, we have that 
$$D_{\rm ex} \leq M|\cal V_{\rm ex}|,$$
and hence 
\begin{equation}
\frac{1}{|{\cal V}|} \sum_{v\in \cal V}d_{v} \leq \frac{1}{2|{\cal V}|}\sum_{[i,j,k]\in \cal T}s_{(d_{i},d_{j},d_{k})}  + \frac{2|\cal V_{\rm ex}|}{|\cal V|},
\end{equation} 
for all $\{d_{v}: v\in \cal V\}$ that satisfy (\ref{eq:cond11}) and (\ref{eq:cond22}).
Letting

\begin{equation}
s^{*} \triangleq \max_{(d_{i},d_{j},d_{k})\in {\cal T_{\rm D}}} s_{(d_{i},d_{j},d_{k})}
\end{equation}
where

\begin{equation*}
{\cal T_{\rm D}} =\left\{(d_{i},d_{j},d_{k}): 
\begin{aligned} &d_{i},d_{j},d_{k} \in \{0,1,...,M\}\\ &d_{i} + d_{j}\leq M \\
&d_{j} + d_{k}\leq M \\
&d_{k} + d_{i}\leq M  \end{aligned}\right\},
\end{equation*}
we can conclude that any degrees of freedom 
$\{d_{v}: v\in \cal V\}$ that are achievable in $\cal G_{\pi}(V,E_{\pi})$ for any $\pi$ must satisfy

\begin{equation}
\frac{1}{|{\cal V}|} \sum_{v\in \cal V}d_{v} \leq \frac{|\cal T|}{2|{\cal V}|}s^{*} + \frac{M|\cal V_{\rm ex}|}{|\cal V|}.
\label{sbound}
\end{equation} 

\begin{lemma} We have that
\begin{equation}s^{*} = \begin{cases}\frac{3M}{2}, \;\;\;\;\;\;\;\;\;\;\;\mbox{$M$ is even}\\ \frac{3M-1}{2} , \;\;\;\;\;\;\;\;\mbox{$M$ is odd}\end{cases}
\end{equation}
\label{lem:sstar}
\end{lemma}
\begin{proof}
By definition, every $(d_{i},d_{j},d_{k})\in {\cal T_{\rm D}}$ must satisfy $d_{i} + d_{j}\leq M$,
$d_{j} + d_{k}\leq M$, $d_{k} + d_{i}\leq M$. Adding these inequalities together we obtain
\begin{equation}
s_{(d_{i},d_{j},d_{k})} \leq \frac{3M}{2}, \;\forall(d_{i},d_{j},d_{k})\in {\cal T_{\rm D}}. 
\label{3m2}
\end{equation}

When $M$ is even, the tuple $\left(\frac{M}{2},\frac{M}{2},\frac{M}{2}\right) \in {\cal T_{\rm D}}$ achieves $s_{\left(\frac{M}{2},\frac{M}{2},\frac{M}{2}\right)} =\mathsmaller{{3M}/{2}}$ and from (\ref{3m2}) we can conclude  that $s^{*} = s_{\left(\frac{M}{2},\frac{M}{2},\frac{M}{2}\right)} = \mathsmaller{{3M}/{2}}$. 

When $M$ is odd,  consider  $\left(\frac{M-1}{2},\frac{M-1}{2},\frac{M+1}{2}\right) \in {\cal T_{\rm D}}$  with $s_{\left(\frac{M-1}{2},\frac{M-1}{2},\frac{M+1}{2}\right)} =\mathsmaller{({3M-1})/{2}}$. Assume that $s^{*}>\mathsmaller{({3M-1})/{2}}$. Since $s_{(d_{i},d_{j},d_{k})}$ is integer, there must exist a tuple  ($d_{i},d_{j},d_{k})\in {\cal T_{\rm D}}$ with  $s_{(d_{i},d_{j},d_{k})}\geq \frac{3M-1}{2}+1 = \frac{3M+1}{2}$. This is a contradiction due to (\ref{3m2}) and hence $s^{*} = s_{\left(\frac{M-1}{2},\frac{M-1}{2},\frac{M+1}{2}\right)} = \mathsmaller{({3M-1})/{2}}$. \end{proof}
\vspace{0.1in}

Combining the results of Lemma 2 and 3 with the bound in (\ref{sbound}) we arrive at 
\begin{equation}
\frac{1}{|{\cal V}|} \sum_{v\in \cal V}d_{v} \leq  \begin{cases} \frac{M}{2} + {\cal O}\left( \scriptstyle{1}/{{\sqrt{|{\cal V}|}}}\right), \;\;\;\;\;\;\;M\mbox{ is even}  \\
\frac{M}{2}-\frac{1}{6} + {\cal O}\left( \scriptstyle{1}/{{\sqrt{|{\cal V}|}}}\right),\;M\mbox{ is odd,}\end{cases}
\end{equation}
which completes the proof of Theorem 2.

%%%%%%%%%%%%%%%%%%%%%%%%%%%%%%%%%%%%%%%%%%%%%%%%%%%
\section{Proof of Theorem \ref{thm:sectors}}\label{proof:thm3}

% Let $\hat{\cal G}\big(\cal V,\hat{\cal E}\big)$ denote the interference graph that includes both out-of-cell and intra-cell interference edges. The set $\cal V$ is defined as in Section \ref{igraph} and the set of edges is given by 
%$$\hat{\cal E} = \big\{(u,v): |\phi(u)-\phi(v)| = 1, \;u,v\in \cal V\big\}$$
%
%

Consider the set $\cal V$ as defined as in Section \ref{igraph} and let 
\begin{equation}
{\cal D_{\rm in}}\triangleq \hspace{-0.2in}\underset{\substack{ a,b\in\mathbb{Z}: \\ \left[a+b\right]{\rm mod}\,3 = 0}}{\bigcup}\hspace{-0.2in}{\Delta}(a+b\omega),
\end{equation}
and
\begin{equation}
{\cal D_{\rm out}}\triangleq \hspace{-0.2in}\underset{\substack{ a,b\in\mathbb{Z}: \\ \left[a+b\right]{\rm mod}\,3 \neq 0}}{\bigcup}\hspace{-0.2in}{\Delta}(a+b\omega)
\end{equation}
where
$${\Delta}(z) = \{(z,z+\omega),\,(z,z+\omega+1),\,(z+\omega,z+\omega+1) \}.$$

The set of out-of-cell edges can be defined as 
\begin{equation}
{\cal E^{\rm out}} = \left\{(u,v) : u,v\in{\cal V} \mbox{ and} \left(\phi(u),\phi(v)\right)\in {\cal D_{\rm out}}\right\},
\end{equation} 
and the set of intra-cell edges as
\begin{equation}
{\cal E}^{\rm in} = \left\{(u,v) : u,v\in{\cal V} \mbox{ and} \left(\phi(u),\phi(v)\right)\in {\cal D_{\rm in}}\right\}.
\end{equation}

The interference graph in this case  is given  by $\hat{\cal G}\big(\cal V,\hat{\cal E}\big)$, where $\hat{\cal E} = {\cal E^{\rm out}} \cup {\cal E^{\rm in}}$. 
We further define the sets $$C(z)=\{z,z+1,z-\omega,z-\omega-1\}\cap {\cal B}_{r}$$  for all $z\in \ZZ(\omega)$ such that  $f(z)=1$. Notice that if $|C(z)|=4$,  the set $C(z)$ corresponds to the labels of four vertices $a,b,c,d \in \cal V$ with 
\begin{align*}
&\phi(a)= z,\\
&\phi(b)= z+1,\\
&\phi(c)= z-\omega,\\
&\phi(d)= z-\omega-1.
\end{align*}
%$\phi(a)= z$, $\phi(b)= z+1$, $\phi(c)= z-\omega$ and $\phi(d)= z-\omega-1$. 
Moreover, the vertices $\{a,b,c\}$ are connected in $\hat{\cal G}\big(\cal V,\hat{\cal E}\big)$ only with edges in ${\cal E}^{\rm in}$ and hence correspond to  sectors of the same cell (cf. Fig.~\ref{intracell}). 

First we are going to show that with the beamforming choices given in Appendix \ref{proofthm1}, the above cell $\{a,b,c\}$ can jointly decode its corresponding messages according to the decoding order $\pi^{*}$. Notice that at the time when receiver $a$ wants to decode, all the sectors that correspond to vertices $v\in {\cal V} :v\prec_{\pi^{*}}a$ have already decoded their messages and no longer cause interference to their neighbors. Hence, the received signal for a sector associated with $u\in \cal V$
can be written as
\begin{align}\yv_{u} = \Hm_{uu}\Vm_{u}\sv_{u}+ \underset{\substack{ {(u, v)\in \hat{\cal E}:} \\ a\prec_{\pi^{*}}v}}\sum\Hm_{uv}\Vm_{v}\sv_{v} + \zv_{u}.\end{align}

The interfering transmitters for receiver $a$ are given by $\{  v : (a, v)\in \hat{\cal E}, a\prec_{\pi^{*}}v\}=\{b,c,d\}$. In order to identify the interfering transmitters for receivers $b$ and $c$ notice that for any $u\in\{b,c\}$ the set $\{  v : (u, v)\in \hat{\cal E}, a\prec_{\pi^{*}}v\}$ can be written as
\begin{align}
%\{  v\in {\cal V} : (a, v)\in \hat{\cal E}, &a\prec_{\pi^{*}}v\}=\{b,c,d\},\nonumber\\
\{  v : (u, v)\in {\cal E}^{\rm in}\}\bigcup\{  v : (u, v)\in \hat{\cal E}^{\rm out}, a\prec_{\pi^{*}}v\}.\nonumber
\end{align}

For receiver $b$ the set $\{  v : (b, v)\in \hat{\cal E}^{\rm out}, a\prec_{\pi^{*}}v\} = \{  v : [v, b]\in \hat{\cal E}^{\rm out}_{\pi^{*}}\}$ and for receiver $c$ we have that $\{  v : (c, v)\in \hat{\cal E}^{\rm out}, a\prec_{\pi^{*}}v\} = \{d\}\cup\{  v : [v, c]\in \hat{\cal E}^{\rm out}_{\pi^{*}}\}$. Putting everything together, the interfering transmitters for receivers $b$ and $c$
are given by
%\begin{align} %%split
%\{  v : (b, v)\in \hat{\cal E}, &a\prec_{\pi^{*}}v\}=\nonumber\\&\{a,c\}\cup\{  v : [v, b]\in \hat{\cal E}^{\rm out}_{\pi^{*}}\},\nonumber
%\end{align}
\begin{align}
\{  v : (b, v)\in \hat{\cal E}, &a\prec_{\pi^{*}}v\}=\{a,c\}\cup\{  v : [v, b]\in \hat{\cal E}^{\rm out}_{\pi^{*}}\},\nonumber
\end{align}
and
%\begin{align} %% split
%\{  v : (c, v)\in \hat{\cal E}, &a\prec_{\pi^{*}}v\}=\nonumber\\&\{a,c,d\}\cup\{  v : [v, c]\in \hat{\cal E}^{\rm out}_{\pi^{*}}\}.\nonumber
%\end{align}
\begin{align}
\{  v : (c, v)\in \hat{\cal E}, &a\prec_{\pi^{*}}v\}=\{a,c,d\}\cup\{  v : [v, c]\in \hat{\cal E}^{\rm out}_{\pi^{*}}\}.\nonumber
\end{align}

A key observation is that according to (\ref{eq:condAAA}) and the achievability scheme of Appendix \ref{proofthm1} all the interference from the transmitters in $\{  v : [v, b]\in \hat{\cal E}^{\rm out}_{\pi^{*}}\}$ and $\{  v : [v, c]\in \hat{\cal E}^{\rm out}_{\pi^{*}}\}$ can be zero-forced at receivers $b$ and $c$ by projecting along $\Um^{\rm H}_{b}$ and $\Um^{\rm H}_{c}$ respectively. The corresponding observations are given by
\begin{align}
\Um^{\rm H}_{b}\yv_{b} = \Um^{\rm H}_{b}\Hm_{bb}\Vm_{b}\sv_{b} + \sum_{v\in\{a,c\}}\Um^{\rm H}_{b}\Hm_{bv}\Vm_{v}\sv_{v} + \zv_{b},
\end{align} 
\begin{align}
\Um^{\rm H}_{c}\yv_{c} = \Um^{\rm H}_{c}\Hm_{cc}\Vm_{c}\sv_{c} + \sum_{v\in\{a,c,d\}}\Um^{\rm H}_{c}\Hm_{\cv}\Vm_{v}\sv_{v} + \zv_{c}.
\end{align} 

We are going to show next that it is possible for the cell $\{a,b,c\}$ to jointly decode the desired messages $\sv_{a}$, $\sv_{b}$ and $\sv_{c}$ from 
the received signals $\yv_{a}$, $\Um^{\rm H}_{b}\yv_{b}$ and $\Um^{\rm H}_{c}\yv_{c}$.
Let 
\begin{align}
\tilde\sv = \begin{bmatrix} \sv_{a}\\ \sv_{b}\\ \sv_{c} \\ \sv_{d}\end{bmatrix},\; 
\tilde\yv = \begin{bmatrix} \yv_{a}\\ \Um^{\rm H}_{b}\yv_{b}\\ \Um^{\rm H}_{c}\yv_{c} \end{bmatrix},
\; \tilde\zv = \begin{bmatrix} \zv_{a}\\ \Um^{\rm H}_{b}\zv_{b}\\ \Um^{\rm H}_{c}\zv_{c} \end{bmatrix},\nonumber
\end{align}
and %${\bf \tilde H} =$
\begin{align}
 {\bf \tilde H} =\begin{bmatrix} \Hm_{aa}\Vm_{a} & \Hm_{ab}\Vm_{b} & \Hm_{ac}\Vm_{c} & \Hm_{ad}\Vm_{d} \\ \Um_{b}^{\rm H}\Hm_{ba}\Vm_{a} & \Um_{b}^{\rm H}\Hm_{bb}\Vm_{b} & \Um_{b}^{\rm H}\Hm_{bc}\Vm_{c} & {\bf 0}_{d_{b}\times d_{d}} 
\\ 
 \Um_{c}^{\rm H}\Hm_{ca}\Vm_{a} & \uv_{c}^{\rm H}\Hm_{cb}\Vm_{b} & \Um_{c}^{\rm H}\Hm_{cc}\Vm_{c} & \Um_{c}^{\rm H}\Hm_{cd}\Vm_{d}\end{bmatrix},\nonumber
\end{align}
such that the available observations in the cell $\{a,b,c\}$ can be written in vector form as

\begin{equation}
{ \bf \tilde y} = {\bf \tilde H} \tilde\sv + { \bf \tilde z}.
\label{vecobs2}
\end{equation} 

\begin{lemma}
If the channel gains $\Hm_{uv}\in \CC^{M\times M}$, $(u,v)\in \hat{\cal E}$  are chosen independently at random from a Gaussian distribution, the matrix ${\bf \tilde H}$ has full column rank with probability one.
\end{lemma}
\begin{proof} The matrix $\tilde \Hm$ has $M+d_{b}+d_{c}$ rows and $d_{a}+d_{b}+d_{c}+d_{d}$ columns. Since $d_{a}+d_{d}\leq M$ for all $M$ (odd or even), we have  to show that $Pr\big[\mbox{rank}(\tilde \Hm) = d_{a}+d_{b}+d_{c}+d_{d}\big] =1$. Recall that the beamforming matrices $\Vm_{v}$, $\Um_{u}$ do not depend on the channel realizations $\Hm_{ij}$ used in the definition of the above matrix. Hence, assuming that  all channel gains are chosen independently at random  from a Gaussian distribution, we can argue that the non-zero entries in $\tilde \Hm$ (given by the projections $\Um^{\rm H}_{u}\Hm_{ij}\Vm_{v}$) are also independent. Let $\Fm=\Pm\tilde\Hm$ be the matrix obtained by rearranging the rows of $\tilde \Hm$ such that all zero entries are in the upper-right corner and consider $\tilde \Fm$ to be the square sub-matrix defined by the $d_{a}+d_{b}+d_{c}+d_{d}$ rows of $\Fm$. We can see that the diagonal elements of $\tilde \Fm$ are going to be non-zero, and hence, we can show by Lemma~\ref{lem:hadamard} that the rank of $\tilde \Fm$ is full with probability one. Therefore, $\tilde\Hm$ will always have $d_{a}+d_{b}+d_{c}+d_{d}$ linearly independent rows and  $Pr\big[\mbox{rank}(\tilde \Hm) = d_{a}+d_{b}+d_{c}+d_{d}\big] =1$.
%\begin{align}
% {\bf \tilde H} =\begin{bmatrix} 
% \Um_{b}^{\rm H}\Hm_{ba}\Vm_{a} & \Um_{b}^{\rm H}\Hm_{bb}\Vm_{b} & \Um_{b}^{\rm H}\Hm_{bc}\Vm_{c} & {\bf 0}_{d_{b}\times d_{d}} 
%\\ 
% \Um_{c}^{\rm H}\Hm_{ca}\Vm_{a} & \uv_{c}^{\rm H}\Hm_{cb}\Vm_{b} & \Um_{c}^{\rm H}\Hm_{cc}\Vm_{c} & \Um_{c}^{\rm H}\Hm_{cd}\Vm_{d} \\
%  \Hm_{aa}\Vm_{a} & \Hm_{ab}\Vm_{b} & \Hm_{ac}\Vm_{c} & \Hm_{ad}\Vm_{d} \end{bmatrix},\nonumber
%\end{align}

\end{proof}

In view of the above lemma, the vector observation in (\ref{vecobs2}) can be used to decode the symbols in $\tilde \sv$ and hence the cell $\{a,b,c\}$ is able to recover the desired messages $\sv_{a}$, $\sv_{b}$ and $\sv_{c}$. 

Applying the above procedure successively, according to the decoding order $\pi^{*}$, we can argue that that all the cells $\{a,b,c\}\subseteq \cal V$  whose labels  correspond to a set $C(z)$ with $|C(z)|=4$,   can decode their desired messages using the beamforming choices of Appendix \ref{proofthm1}. 

In order to conclude the proof it remains to consider all the degenerate cases for cells that lie on the boundary of $\hat{\cal G}\big(\cal V,\hat{\cal E}\big)$ and correspond to $C(z)$ with $|C(z)|\leq 3$.  When $C(z)=1$, there is only out-of-cell interference and hence the scheme works as described in Appendix \ref{proofthm1}. This is also the case when $|C(z)| = 2$ and $\phi(d)=z-1-\omega\in C(z)$. If $|C(z)| = 2$ and $\phi(d)=z-1-\omega \notin C(z)$ the two sectors $u,v$ of the given cell can zero-force all out-of-cell interference and use their vector observation 
\begin{equation*}
\begin{bmatrix}
\Um^{\rm H}_{u}\yv_{u}\\ \Um^{\rm H}_{v}\yv_{v}\end{bmatrix}
=
\begin{bmatrix} 
\Um^{\rm H}_{v}\Hm_{vv}\Vm_{v} & \Um_{v}^{\rm H}\Hm_{vu}\Vm_{u}\\ 
\Um_{u}^{\rm H}\Hm_{uv}\Vm_{v} & \Um_{u}^{\rm H}\Hm_{uu}\Vm_{u} \end{bmatrix}
\begin{bmatrix}
\sv_{u}\\ \sv_{v}\end{bmatrix} + \tilde \zv
\end{equation*}
to jointly decode the desired messages $\sv_{u}$ and $\sv_{v}$. 
In a similar fashion, all the cells $\{a,b,c\}$ that correspond to a set $C(z)$ with $|C(z)| = 3$ and $\phi(d)=z-1-\omega \notin C(z)$ can decode their messages using the projected observations $\Um^{\rm H}_{a}\yv_{a}$, $\Um^{\rm H}_{b}\yv_{b}$ and  $\Um^{\rm H}_{c}\yv_{c}$. Finally, the only possible cell configuration with $|C(z)| = 3$ and $\phi(d)=z-1-\omega \in C(z)$, is $\{a,c\}$ with $\phi(a)=z$ and $\phi(c)=z-\omega$. The corresponding vector observation is given by
%\begin{equation*} %split
%\begin{bmatrix}
%\yv_{a}\\ \Um^{\rm H}_{c}\yv_{c}\end{bmatrix}
%=
%\tilde\Hm
%\begin{bmatrix}
%\sv_{a}\\ \sv_{c}\\ \sv_{d}\end{bmatrix} + \tilde \zv, \;\; \mbox{where}
%\end{equation*}
%
%\begin{equation*}
%\tilde\Hm=
%\begin{bmatrix} 
%\Hm_{aa}\Vm_{a} & \Hm_{ac}\Vm_{c}& \Hm_{ad}\Vm_{d} \\ 
%\Um_{c}^{\rm H}\Hm_{ca}\Vm_{a} & \Um_{c}^{\rm H}\Hm_{cc}\Vm_{c}&  \Um_{c}^{\rm H}\Hm_{cd}\Vm_{d} \end{bmatrix}.
%\end{equation*}
\begin{equation*}
\begin{bmatrix}
\yv_{a}\\ \Um^{\rm H}_{c}\yv_{c}\end{bmatrix}
=
{\begin{bmatrix} 
\Hm_{aa}\Vm_{a} & \Hm_{ac}\Vm_{c}& \Hm_{ad}\Vm_{d} \\ 
\Um_{c}^{\rm H}\Hm_{ca}\Vm_{a} & \Um_{c}^{\rm H}\Hm_{cc}\Vm_{c}&  \Um_{c}^{\rm H}\Hm_{cd}\Vm_{d} \end{bmatrix}}
\begin{bmatrix}
\sv_{a}\\ \sv_{c}\\ \sv_{d}\end{bmatrix} + \tilde \zv.
\end{equation*}

Notice that the above channel matrix has $M+d_{c}$ rows and $d_{a}+d_{c}+d_{d}$ columns. According to the beamforming choices of Appendix \ref{proofthm1},  we have that $d_{a}+ d_{d} \leq M$  for all $M$ and hence we can argue as before that the above matrix has full column rank with probability one. Therefore,
the receivers $a$ and $c$ can jointly decode their desired messages in this case as well.

%%%%%%%%%%%%%%%%%%%%%%%%%%%%%%%%%%%%%%%%% 
\section{Proof of Theorem~\ref{thm-comp}}  \label{comp-ach-proof}

The proof can be obtained as a straightforward generalization of the proof described in Section \ref{robust:achievability} using the beamforming design of Theorem \ref{thm:sectors} given in Appendix \ref{proof:thm3}. Applying Lemma~\ref{lem:hadamard},  we can show that primary  and secondary sectors are always able to decode their messages from the available observations
\begin{equation}
\begin{bmatrix}
\yv_{a}\\
\Um_{b}^{\rm H}\yv_{b}\\
\Um_{c}^{\rm H}\yv_{c}
\end{bmatrix}=
\begin{bmatrix} \Hm_{aa}\Vm_{a} & \overline\Hm_{ab}\Vm_{b} & \overline\Hm_{ac}\Vm_{c} & \overline\Hm_{ad}\Vm_{d} \\ 
 \Um_{b}^{\rm H}\overline\Hm_{ba}\Vm_{a} & \Um_{b}^{\rm H}\Hm_{bb}\Vm_{b} & \Um_{b}^{\rm H}\overline\Hm_{bc}\Vm_{c} &  {\bf 0}_{d_{b}\times d_{d}}  \\ 
 \Um_{c}^{\rm H}\overline\Hm_{ca}\Vm_{a} & \Um_{c}^{\rm H}\overline\Hm_{cb}\Vm_{b} & \Um_{c}^{\rm H}\Hm_{cc}\Vm_{c} & \Um_{c}^{\rm H}\overline\Hm_{cd}\Vm_{d}\end{bmatrix}
 \begin{bmatrix}
s_{a}\\
s_{b}\\
s_{c}\\
s_{d}
\end{bmatrix} 
+\tilde\zv,
\end{equation}
and
\begin{equation}
\begin{bmatrix}
\Um_{b}^{\rm H}\yv_{b}\\
\Um_{c}^{\rm H}\yv_{c}
\end{bmatrix}={
\begin{bmatrix}  
\Um_{b}^{\rm H}\Hm_{bb}\Vm_{b} & \Um_{b}^{\rm H}\overline\Hm_{bc}\Vm_{c}  \\ 
 \Um_{c}^{\rm H}\overline\Hm_{cb}\Vm_{b} & \Um_{c}^{\rm H}\Hm_{cc}\Vm_{c}\end{bmatrix}}
 \begin{bmatrix}
s_{b}\\
s_{c}
\end{bmatrix} 
+\tilde\zv,
\end{equation}
for all channel-state  configurations given in Section \ref{robust:achievability}. We omit the details here for brevity.

%
%\begin{align}
%%\resizebox{1\hsize}{!}{$
%{\bf \tilde H}(\overline\alpha) =  \begin{bmatrix} \Hm_{aa}\Vm_{a} & \overline\Hm_{ab}\Vm_{b} & \overline\Hm_{ac}\Vm_{c} & \overline\Hm_{ad}\Vm_{d} \\ 
% \Um_{b}^{\rm H}\overline\Hm_{ba}\Vm_{a} & \Um_{b}^{\rm H}\Hm_{bb}\Vm_{b} & \Um_{b}^{\rm H}\overline\Hm_{bc}\Vm_{c} &  {\bf 0}_{d_{b}\times d_{d}}  \\ 
% \Um_{c}^{\rm H}\overline\Hm_{ca}\Vm_{a} & \Um_{c}^{\rm H}\overline\Hm_{cb}\Vm_{b} & \Um_{c}^{\rm H}\Hm_{cc}\Vm_{c} & \Um_{c}^{\rm H}\overline\Hm_{cd}\Vm_{d}\end{bmatrix}.
%% $}
% \nonumber
%\end{align}

%%%%%%%%%%%%%%%%%%%%%%%%%%%%%%%%%%%%%%%%% 
\section{Proof of Theorem~\ref{thm:converse2}}  \label{thm:converse2-proof}

Here, we are going to follow an approach similar to the one in Section \ref{sec:converse1} and show that for any decoding 
order $\pi$, any linear scheme for the system  $\left\{\hat{\cal G}\big({\cal V},\hat{\cal E}\big|{\cal A}\big):{\cal A} \in \{0,1\}^{2|\hat{\cal E}|}\right\}$ 
achieves compound DoFs $d_{\rm C}$ upper bounded by $d_{\rm C}^{*} + {\cal O}\left( \scriptstyle{1}/{{\sqrt{|{\cal V}|}}}\right)$, 
as stated in Theorem~\ref{thm:converse2}.

First we will upper bound $d_{\rm C}$ by conditioning on a specific channel-state configuration 
${\cal A}^{*}$ shown in Fig.~\ref{fig:triangles2}. We have that 
\begin{equation}
d_{\rm C} = \min_{ {\cal A} \in \{0,1\}^{2|\hat{\cal E}|}} d_{\hat{\cal G}}({\cal A})\leq d_{\hat{\cal G}}({\cal A}^{*}), 
\end{equation}
where ${\cal A}^{*}$ is given by  setting $\alpha_{ij} = 1$ for all  edges $[i,j]\in \hat{\cal E}$ that belong  
to the triangles, 
%{\RED [QUESTION: I don't understand what you say here ... the edges $(u,v)$ corresponding to 
%the triangle $[u,v,w]$ may mean just *one* edge for each triangle, or three edges $(u,v), (v,w), (w,u)$ for the triangle. 
%in this case, don't say ``edge $(u,v)$" since it is understood as just the *first* edge $(u,v)$ of the triangle, i.e., 
%one edge per triangle. In contrast, from the figure it seems that you mean **all three edges** of the triangle $[u,v,w]$.
%This sentence has to be re-phrased]}
{\begin{equation*}
\hat{\cal T} \hspace{-0.05in}=\hspace{-0.05in} \{[u,v,w]: [\phi(u),\phi(v),\phi(w)]\in \hat{\cal P}, u,v,w\in {\cal V}\},
\end{equation*}
with $\hat{ \cal P} = \{[z,z+\omega,z+\omega+1]: z\in\Lambda_0 - 1 \},$}
and $\alpha_{ij}=0$ otherwise. 

%{\RED [Again, for the purpose of limiting proliferation of symbols and notation, can we define $\hat{\cal T}$ as follows:
%$\hat{\Tc}$ is the set of triplets of vertices such that the corresponding  triangles are  $\{\Delta(z) :  z \in \Lambda_0 - 1\}$ or, equivalently, 
%$\hat{\Tc}$ is the set of triplets of triangle vertices $[u,v,w] \in \Tc$ such that $u \in \Vc_{\rm diamond}$]} 

\begin{figure}[ht]
                \centering
                \includegraphics[width=.5\columnwidth]{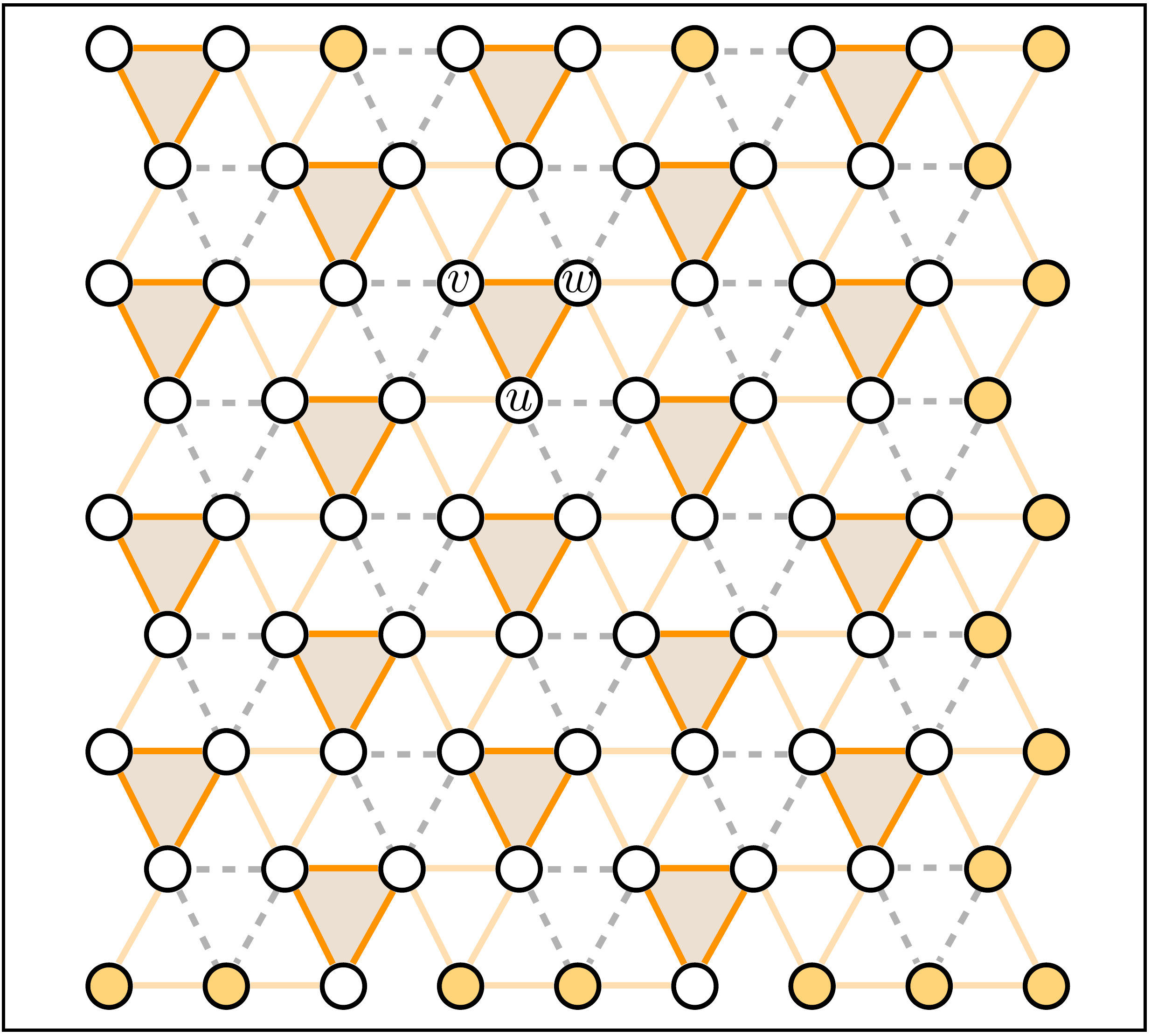}
                \caption{The set of triangles $[u,v,w]\in \hat{\cal T}$ for  $\hat{\cal G}\big({\cal V},\hat{\cal E}\big|{\cal A}^{*}\big)$. All the circle 
                nodes belong to $\hat{\cal V}_{\rm in}$ and participate in one triangle ($n_{v}=1$). The set $\hat{\cal V}_{\rm ex}$ contains the  
                colored nodes on the  boundary for which $n_{v}=0$.}
\label{fig:triangles2}
\end{figure}

Notice, that the set $\hat{\cal T}$ has been chosen such that for all $[u,v,w] \in \hat{\cal T}$, the sectors associated with the nodes $u$, $v$, $w$ belong to different cells and hence cannot be jointly decoded. 
Arguing as in Lemma~\ref{lem:sstar} in Appendix~\ref{proof:thm2}, we can show that for any decoding order, the total degrees of freedom achievable in each triangle cannot be more than  
\begin{equation} s^{*}= \begin{cases}\frac{3M}{2}, \;\;\;\;\;\;\;\;\;\;\;\mbox{$M$ is even}\\ \frac{3M-1}{2} , \;\;\;\;\;\;\;\;\mbox{$M$ is odd.}\end{cases}
\end{equation}
Further, observe that apart from some vertices on the external boundary of the graph ($v\in \hat{\cal V}_{\rm ex}$), all other nodes ($v\in\hat{\cal V}_{\rm in}$) participate in exactly one triangle in $\hat{\cal T}$, and therefore we have that $\sum_{v\in \cal V}d_{v} = \sum_{[i,j,k]\in \hat{\cal T}} (d_{u}+d_{v}+d_{w}) + \sum_{v\in \hat{\cal V}_{\rm ex}}d_{v}$.
The average (per sector) degrees of freedom achievable in $\hat{\cal G}\big({\cal V},\hat{\cal E}\big|{\cal A}^{*}\big)$ can be bounded as 
\begin{align}
d_{\hat{\cal G}}({\cal A}^{*}) &= \frac{1}{|\cal V|}\sum_{v\in \cal V}d_{v}\\
& = \frac{1}{|\cal V|}\sum_{[i,j,k]\in \hat{\cal T}} (d_{u}+d_{v}+d_{w})  + \frac{1}{|\cal V|}\sum_{v\in \hat{\cal V}_{\rm ex}}d_{v}\\
&\leq \frac{|\hat{\cal T}|}{|{\cal V}|}s^{*} + \frac{|\hat{\cal V}_{\rm ex}|}{|{\cal V}|}.
\end{align}
To conclude the proof we can argue as, in Lemma~\ref{lem:2}, that ${|\hat{\cal T}|}\leq \frac{1}{3}{|{\cal V}|}$ and $|{\cal V_{\rm ex}}| = {\cal O}\left(\sqrt{|\cal V|}\right)$, and show that $d_{\rm C} = \frac{1}{3}s^{*}+{\cal O}\left( \scriptstyle{1}/{{\sqrt{|{\cal V}|}}}\right)$.

\bibliographystyle{ieeetr}
\bibliography{referencesIT}

\end{document}